\documentclass[onecolumn,fleqn,usenatbib]{mnras}
\usepackage{newtxtext,newtxmath}

\usepackage[T1]{fontenc}

\DeclareRobustCommand{\VAN}[3]{#2}
\let\VANthebibliography\thebibliography
\def\thebibliography{\DeclareRobustCommand{\VAN}[3]{##3}\VANthebibliography}

\usepackage{graphicx}	
\usepackage{amsmath}	

\usepackage{natbib}
\usepackage{float}
\usepackage{ulem}
\usepackage{relsize}
\pdfoutput=1
\usepackage{color}
\usepackage{commath}
\usepackage{microtype}

\usepackage{natbib,caption} 	
\usepackage{threeparttablex}
\usepackage{tabularx}  			
\usepackage{booktabs,amsfonts} 

\title[Quiescent Radio Occurrence Rate for UCDs]{The Occurrence Rate of Quiescent Radio Emission for Ultracool Dwarfs  using a Generalized Semi-Analytical Bayesian Framework}

\author[Kao \& Shkolnik]{
Melodie M. Kao,$^{1,2,3,4}$\thanks{E-mail: melodie.kao@ucsc.edu}
Evgenya L. Shkolnik$^{1}$
\\
$^{1}$Arizona State University, School of Earth and Space Exploration, 550 E Tyler Mall PSF 686, Tempe, AZ 85287\\
$^{2}$University of California Santa Cruz, Department of Astronomy amd Astrophysics, 1156 High Street, Santa Cruz, CA 95064\\
$^{3}$NASA Hubble Postdoctoral Fellow\\
$^{4}$Heising-Simons 51 Pegasi b Fellow\\
}

\date{Accepted 2023 June 28. Received 2023 June 28; in original form 2022 April 11}

\pubyear{2022}

\begin{document}
\label{firstpage}
\pagerange{\pageref{firstpage}--\pageref{lastpage}}
\maketitle

\begin{abstract}
We present a generalized analytical Bayesian framework for calculating the occurrence rate of steady emission (or absorption) in astrophysical objects.  As a proof-of-concept, we apply this framework to non-flaring quiescent radio emission in ultracool ($\leq$ M7) dwarfs.  Using simulations, we show that our framework recovers the simulated radio occurrence rate to within 1--5\% for sample sizes of 10--100 objects when averaged over an ensemble of trials and simulated occurrence rates for our assumed luminosity distribution models. In contrast, existing detection rate studies may under-predict the simulated rate by 51--66\% because of sensitivity limits.  Using all available literature results for samples of 82 ultracool M dwarfs, 74 L dwarfs, and 23 T/Y dwarfs, we find that the maximum-likelihood quiescent radio occurrence rate is between  $15^{+4}_{-4}$ -- $20^{+6}_{-5}$\%, depending on the luminosity prior that we assume. Comparing each spectral type, we find occurrence rates of $17^{+9}_{-7}$ -- $25^{+13}_{-10}$\%  for M dwarfs, $10^{+5}_{-4}$ -- $13^{+7}_{-5}$\% for L dwarfs, and $23^{+11}_{-9}$ -- $29^{+13}_{-11}$\%  for T/Y dwarfs.  We rule out potential selection effects and speculate that age and/or rotation may account for tentative evidence that the quiescent radio occurrence rate of L dwarfs may be suppressed compared to M and T/Y dwarfs and phenomenon.  Finally, we discuss how we can harness our occurrence rate framework to carefully assess the possible physics that may be contributing to observed occurrence rate trends.

\end{abstract}




\section{Introduction}\label{sec.intro}
Observations of gigahertz radio emission from ultracool dwarfs (M7 and later spectral type) demonstrate that these objects emit two distinct radio components.  The first is periodically flaring and highly circularly polarized electron cyclotron maser emission that traces kiloGauss magnetic fields \citep{Hallinan2007ApJ...663L..25H, Hallinan2008ApJ...684..644H}.  This emission is the radio manifestation of aurorae \citep{Hallinan2015Natur.523..568H,Kao2016ApJ...818...24K, Pineda2017ApJ...846...75P} and occurs in ultracool dwarfs at least as late as T6.5 \citep{RouteWolszczan2012ApJ...747L..22R, RouteWolszczan2016ApJ...821L..21R, Williams2015ApJ...808..189W, Kao2016ApJ...818...24K, Kao2018ApJS..237...25K}. 

The second component is quasi-steady and weakly circularly polarized ``quiescent" emission.  This nonthermal and incoherent radio emission accompanies all known detections of ultracool dwarf aurorae \citep{Pineda2017ApJ...846...75P}.  While the source of the magnetospheric plasma responsible for quiescent radio emission in ultracool dwarfs is unknown, the emission itself has been attributed to optically thin (gyro)synchrotron emission \citep[e.g.][]{Berger2005ApJ...627..960B, Osten2006ApJ...637..518O, Williams2015ApJ...815...64W, Lynch2016MNRAS.457.1224L, Kao2023}. It persists for months to years \citep[e.g.][]{Berger2008ApJ...676.1307B, Kao2016ApJ...818...24K, Kao2018ApJS..237...25K}, pointing to long-lived trapped populations of relativistic electrons in ultracool dwarf magnetospheres that cannot be accelerated via flare processes \citep{Kao2023}.  Instead, recent resolved imaging at 8.4~GHz and 4.9~GHz show that such emission may trace radiation belt analogs around ultracool dwarfs \citep{Kao2023, Climent2023arXiv230306453C}.

The presence of incoherent nonthermal radio emission requires both a source of radiating electrons and magnetic fields to accelerate these electrons. Thus, the absence of quiescent radio emission from a single object cannot distinguish between a lack of strong magnetic fields or magnetospheric plasma.  In contrast, examining the occurrence rates of quiescent radio emission can potentially reveal how ultracool dwarf radio activity depends on various fundamental properties such as $T_{\mathrm{eff}}$, mass,  rotation rate, and age. Statistical studies of the occurrence rate of quiescent ultracool dwarf radio emission may yield insights into the nature of this emission. 

However, no such analyses currently exist. Instead, volume-limited radio surveys yield detection rates between $\sim$5--10\% for late-M, L, and T ultracool dwarfs \citep{Antonova2013AA...549A.131A, Lynch2016MNRAS.457.1224L, RouteWolszczan2016ApJ...830...85R}.  Examining detections as a function of $v\sin i$ measurements shows that the detection fraction increases for more rapidly rotating objects \citep{Pineda2017ApJ...846...75P}.  These studies do not distinguish between flaring versus quiescent radio emission nor do they separately treat binaries from single objects.   Furthermore, they likely underpredict the occurrence rate of ultracool dwarf radio emission, since detection rates cannot account for observation effects like sensitivity and time variability.  Consequently, the degree to which detection fractions underpredict the true radio occurrence rate is unknown.  Finally, as sample sizes increase, the cumulative impact of unaccounted-for systematics grows.  One consequence is that confidence intervals on detection rates are less robust at recovering emission rates and therefore act as increasingly poor estimators of occurrence rates.  In other words, more data can increase the precision of detection rates at the cost of decreasing their accuracy.  Thus, ultracool dwarf magnetic activity trends identified with detection rate studies that do not account for completeness \citep[e.g.][]{Pineda2016ApJ...826...73P, Pineda2017ApJ...846...75P, Schmidt2015AJ....149..158S} must be treated with caution.

A different approach to studying magnetic activity on ultracool dwarfs is to quantify the occurrence rate of flaring behavior.  Using a Monte Carlo simulation to model the population of radio flaring ultracool dwarfs, \citet{Route2017ApJ...845...66R} find that $\sim$4.6\% of ultracool dwarfs are magnetically active, which agrees with detection rates from volume-limited surveys.  However, the rarity of ultracool dwarf radio detections together with the complex nature of flaring ultracool dwarf radio emission require their framework to rely on simplifying assumptions that are not self-consistent.  In particular, without measured ultracool dwarf radio flare luminosity distributions, they assume luminosity distributions that follow power law indices from extreme UV and X-ray flare studies of mid-M dwarfs. This assumption is at odds with observations that demonstrate that auroral activity is distinct from stellar flares  \citep{Hallinan2015Natur.523..568H, Kao2016ApJ...818...24K, Pineda2017ApJ...846...75P}.  Indeed, X-ray activity in confirmed and candidate auroral objects departs dramatically from that of M dwarfs  \citep{Berger2010ApJ...709..332B, Williams2014ApJ...785....9W, Callingham2021NatAs...5.1233C} and a transition from stellar-like to planetary-like magnetic activity appears to occur over the ultracool dwarf regime \citep{Pineda2017ApJ...846...75P}. Additionally, radio aurorae are beamed rather than isotropically emitted \citep{Zarka1998JGR...10320159Z}, but \citet{Route2017ApJ...845...66R} does not account for the geometric effects of beaming. Finally, they assume that off-target phase calibration scans do not affect the observability of periodically flaring radio aurorae. Thus, their framework cannot capitalize on the majority of ultracool dwarf radio observations that are available in the literature.  

Quantifying ultracool dwarf magnetic activity using their quiescent rather than periodically flaring auroral radio emission has several benefits.  First, detections of quiescent emission do not depend strongly on rotational phase coverage or serendipitously timed calibration scans.  This is in contrast to periodically flaring ultracool dwarf aurorae, which can have short flare lifetimes. Second, we do not have to account for beaming effects. Third, we show in this paper that calculating the quiescent radio occurrence rate can be done analytically from first principles with few physical assumptions.

\section{Method}\label{sec.Occurrence_Rate}
\newcommand{\rms}{\sigma_{\mathrm{rms}}}
\newcommand{\rmsi}{\sigma_{\mathrm{rms_i}}}
\newcommand{\derr}{\sigma_{\mathrm{d}}}
\newcommand{\derri}{\sigma_{\mathrm{d_i}}}
\newcommand{\pierri}{\sigma_{\mathrm{\pi_i}}}
\newcommand{\pierr}{\sigma_{\mathrm{\pi}}}
\newcommand{\detect}{\mathrm{detect}}
\newcommand{\thr}{\mathrm{thr}}
\newcommand{\prob}{\mathbb{P}}
\newcommand{\pdf}{\mathcal{P}}
\newcommand{\normal}{\mathcal{N}}
\newcommand{\vin}{V_{\mathrm{in}}}
\newcommand{\vout}{V_{\mathrm{out}}}
\newcommand{\vtot}{V_{\mathrm{tot}}}

\subsection{$\pdf(\theta \mid D)$: the probability density distribution that the occurrence rate is $\theta$ given observed data $D$}

Let $\theta$ be the occurrence rate of radio emission for which we seek a probability density.  
We first derive an analytic expression for the probability $\prob(D \mid \theta)$ of observing our data $D$ given a fixed occurrence rate $\theta = \Theta$.  Using data compiled from the literature, we numerically calculate $\prob(D \mid \theta  = \Theta)$ for $\theta \in [0,1]$.  We then apply Bayes' Theorem to calculate the probability density distribution $\pdf(\theta \mid D)$.  

Our data $D$ consists of a set of $N$ observations that we have compiled from all available radio observations of ultracool dwarfs in the literature, summarized in Table \ref{table:literatureBD}.  Let the $i^\mathrm{th}$ observation $D_i = \{ \detect_i, \rmsi, d_i, \derri, L_{\nu, i} \}$ consist of the values for the following random variables: 
\begin{enumerate}
\item detect (0 if undetected, 1 if detected)
\item $\rms$ (rms noise for each object image)
\item $d$ or $\pi$ (distance to or parallax of each object)
\item $\derr$ or $\pierr$ (error for the measured $d$ or $\pi$)
\item $L_{\nu}$ (assumed specific luminosity for quiescent radio emission).
\end{enumerate}
We use parallax measurements when available or other distance estimates when not.

We assume that the data from each observation are independent from other observations, such that
\begin{equation} \label{eqn:p(D|theta)}
\prob(D \mid \theta=\Theta) = \prod_{i=1}^N \prob(D_i \mid \theta=\Theta) \quad .
\end{equation}
In contrast to an interpretation that a constant fraction of ultracool dwarfs emit radio emission, our assumption states that every individual ultracool dwarf independently has some intrinsic probability $\theta$ of emitting radio emission. We discuss how to construct a dataset that satisfies our assumption of independence in \S \ref{sec.DataInclusion}.

For each observation, $D_i$ consists of the intersection of observed random variables:
\begin{equation} 
	\prob(D_i \mid \theta=\Theta)  \,\, =  \,\, \prob(\detect = \detect_i \cap \rms = \rmsi \cap \derr = \derri \cap d = d_i \cap L_{\nu} = L_{\nu,i} \mid \theta=\Theta) \quad .
\end{equation} 
For the sake of brevity, we omit the random variables and instead write:
\begin{equation} \label{eqn:p_Di_short}
 	\prob(D_i \mid\theta= \Theta) = \prob(\detect_i \cap \rmsi \cap \derri \cap d_i \cap L_{\nu, i} \mid \theta=\Theta) \quad.
\end{equation}\label{eqn:conditioning}
Conditioning on each characteristic and simplification by independence yield
\begin{equation} \label{eqn:p_Di_conditioned_simplified}
 \prob(D_i \mid \theta = \Theta) =   \prob(\rmsi)\times \prob(\derri)  \times\prob(d_i) \times \prob(L_{\nu,i}) 
 									 \times \prob(\detect_i \mid \theta = \Theta, \rmsi, \derri,  d_i, L_{\nu,i} ) 
				    	      \quad .
\end{equation}
For objects with parallax measurements, $\pi = \pi_i $ replaces $d = d_i$ for the remainder of the discussion, except where noted.

Eq. \ref{eqn:p_Di_conditioned_simplified} asserts that distances and intrinsic luminosities are independent, which is physically true. However, little is known about the radio emission that we aim to study, other than that it is non-thermal, not auroral, and can trace synchrotron radiation belts \citep{Kao2023, Climent2023arXiv230306453C}. As such, we cannot use physical arguments to select an intrinsic luminosity.  Instead, we only know if an observation corresponds to a detection or not, and this \textit{does} depend on distance and rms noise. We can account for observational realities by using the law of total probability to combine the last two quantities in the right-hand product of  Eq. \ref{eqn:p_Di_conditioned_simplified} and consider the mutually exclusive cases where the luminosity of an object is and is not detectable:
\begin{equation} \label{eqn:p_Di_conditioned_totalprob}
\begin{split}
 \prob(D_i \mid & \theta=\Theta) = \prob(\rmsi) \times  \prob(\derri)  \times\prob(d_i) \times \\
  \Bigg[\bigg( & \prob(\detect_i \mid \theta= \Theta, \rmsi, \derri, d_i, L_{\nu,i} \geq L_{\mathrm{thr, i}} ) 
					   \prob(L_{\nu,i} \geq L_{\mathrm{thr, i}} \mid \rmsi, \derri, d_i)\bigg) \\
				     + \bigg(&\prob(\detect_i \mid \theta= \Theta, \rmsi,\derri, d_i, L_{\nu,i} < L_{\mathrm{thr, i}} )  
				    \times \prob(L_{\nu,i} < L_{\mathrm{thr, i}} \mid  \rmsi,\derri, d_i) \bigg)\Bigg] \quad ,
\end{split}
\end{equation}
for which  $L_{\mathrm{thr, i}}$ is the minimum detectable luminosity.  It is a function of the given $\rmsi$, an assumed minimum signal-to-noise ratio threshold $\thr$ for a detection, and an assumed finite distance $x_i$: 
\begin{equation}\label{eqn:L_thri}
 L_{\mathrm{thr, i}}  =    L_{\mathrm{thr, i}}(x = x_i \mid \rmsi)  = \thr \times 4\pi x_i ^2  \times \rmsi  \quad .
\end{equation}
Here, $x_i$ is a random variable from the probability density distribution $\pdf(x_i \mid d_i, \derri)$ that describes the object's location in space given its measured $d_i$ and $\derri$. The rms noise in interferometric images are not Gaussian, and 3$\rms$ detections are considered tentative. Bootstrapping demonstrates that 4$\rms$ corresponds to a $<$1\%  false positive rate \citep[e.g.][]{Kao2018ApJS..237...25K}. Therefore, we assume a constant value for $\thr = 4$ such that $\detect_i = 1$ for objects with flux densities  $\geq$4$\rms$. 

Rather than assuming a single intrinsic luminosity for a given observation, we can now incorporate our ignorance by
assuming a distribution  $\pdf(L_{\nu})$  of possible luminosities to calculate the probability $\prob(L_{\nu,i} \geq L_{\mathrm{thr, i}} \mid \rmsi, \derri, d_i)$  that an object's luminosity is detectable.  We can then marginalize over all possible intrinsic luminosities and distances  (\S\ref{sec.lum}).  This choice allows one to explore the impact of different luminosity distributions as well as incorporate the current state of knowledge in a given science application (\S \ref{sec.modelValidation}).  

The detectability of an object also depends on whether or not it is emitting, so we introduce a new random variable $e \in \{0,1\}$, where $e=0$ is not emitting and $e=1$ is emitting.  The probability that an object is or is not emitting $\prob(e \mid \theta= \Theta)$ depends only on the occurrence rate:
\begin{align} \label{eqn:emitting}
 		\prob &(e = 1 \mid \theta= \Theta) = \Theta \\
 		\prob &(e = 0 \mid \theta= \Theta) = 1-\Theta \quad.
\end{align}
We then invoke the law of total probability once more to consider both cases:
\begin{equation} \label{eqn:p_Di_conditioned_totalprob_emit}
\begin{split}
 \prob(D_i \mid  \theta=\Theta) = \prob(\rmsi) & \times   \prob(\derri)  \times\prob(d_i) \times \\
 \sum_{e \in  \{0,1\}} 
 \prob(e \mid \theta= \Theta) \times 
 \Bigg[\bigg(  &\prob(\detect_i \mid \theta= \Theta, \rmsi, \derri, d_i, L_{\nu,i} \geq L_{\mathrm{thr, i}}, e ) 
					 \times \prob(L_{\nu,i} \geq L_{\mathrm{thr, i}} \mid \rmsi, \derri, d_i, e)\bigg) \\
				     + \bigg(&\prob(\detect_i \mid \theta= \Theta, \rmsi,\derri, d_i, L_{\nu,i} < L_{\mathrm{thr, i}},e )  
				    \times \prob(L_{\nu,i} < L_{\mathrm{thr, i}} \mid  \rmsi,\derri, d_i,e) \bigg)\Bigg] \quad .
\end{split}
\end{equation}

An object is detectable if and only if it is emitting and its luminosity is at least the minimum detectable luminosity.  
Conversely, an object is undetectable if it is emitting and its luminosity is less than the minimum detectable luminosity, as well as if it is not emitting:
\begin{align} \label{eqn:detect}
	\prob(\detect_i = 1 \mid \theta=\Theta, \rmsi,\derri, d_i, L_{\nu,i} \geq  L_{\mathrm{thr, i}}, e=1 )  & = 1 \\
	\prob(\detect_i = 0 \mid \theta=\Theta, \rmsi,\derri, d_i, L_{\nu,i} <  L_{\mathrm{thr, i}}, e=1 ) & = 1 \\
	\prob(\detect_i = 0 \mid \theta=\Theta, \rmsi,\derri, d_i, L_{\nu,i} \geq  L_{\mathrm{thr, i}}, e=0 ) & = 1 \\
	\prob(\detect_i = 0 \mid \theta=\Theta, \rmsi,\derri, d_i, L_{\nu,i} <  L_{\mathrm{thr, i}}, e=0 ) & = 1 \\
\end{align}

Combining Eq. \ref{eqn:emitting},  Eq.  \ref{eqn:p_Di_conditioned_totalprob_emit}, and Eq. \ref{eqn:detect} gives the piecewise analytic function for the probability of observing our data given an occurrence rate 
\begin{equation} \label{eqn:p_Di_final}
\prob(D \mid \theta=\Theta)= \displaystyle{\prod_{i=1}^N  }  \prob(\rmsi) \, \prob(\derri) \, \prob(d_i)  \times
\begin{cases} 
      \Theta \, \prob(L_{\nu,i} \geq L_{\mathrm{thr, i}} \mid \rmsi,\derri, d_i, e=1) 	 	&, \,\,   \detect_i = 1 \\
      1-\Theta \,  \prob(L_{\nu,i} \geq L_{\mathrm{thr, i}} \mid \rmsi,\derri, d_i, e=1)      &, \,\, \detect_i = 0 \quad .
\end{cases}
\end{equation}


\subsection{$\prob(L_{\nu,i} \geq L_{\mathrm{thr, i}} \mid  \rmsi,\derri, d_i)$: Calculating the probability that an object's emitted luminosity is detectable 
} \label{sec.lum}

\subsubsection{Expression using distance $d_i$}

Recall that the minimum detectable luminosity $L_{\mathrm{thr, i}}=4\pi x_i^2 \,\, \thr \, \rmsi$ is a function of an assumed distance $x_i$ for the object in our observation (Eq. \ref{eqn:L_thri}).  However, $x_i$ is not known exactly and instead is drawn from the distribution $\pdf(x_i \mid d_i, \derri)$ describing an object's location in space given its measured $d_i$ and $\derri$.  We therefore invoke the law of total probability and integrate over the joint probability distribution that an object's luminosity $L_{\nu,i}$ is detectable at distance $x_i$ and that it is located at distance $x_i$: 
\begin{equation} \label{eqn:Prob_L_nu}
	\prob(L_{\nu,i}  \geq  L_{\mathrm{thr, i}}   \mid \rmsi,\derri, d_i, e=1)   = 
	\displaystyle{\int\limits_{0}^{d_{\mathrm{thr},{\mathrm{max_i}}}}} \quad
	\displaystyle{\int\limits_{L_{\mathrm{thr, i}}}^{ L_{\nu,\max} } }
	\pdf(L_{\nu,i} \cap x_i \mid \rmsi,\derri, d_i, e=1) 
	\quad  dL_{\nu,i}  \, dx_i   \quad.
\end{equation}

For the outer integral,  we choose a lower limit of $0$ because objects cannot have negative distances.   The cutoff distance for detecting the maximum possible luminosity defines the upper limit
\begin{equation}
		d_{\mathrm{thr},{\mathrm{max_i}}} =  \bigg( \frac{ L_{\nu,\mathrm{max}}}{ 4 \pi \,\, (\thr)  \,  \rmsi} \bigg)^{1/2}  \quad .
\end{equation}
 For the inner integral, the lower limit is by definition the minimum detectable luminosity $L_{\mathrm{thr, i}}$.  The upper limit is the maximum possible luminosity $L_{\nu,\mathrm{max}}$, which we discuss in \S \ref{sec:lum_distrib}.

By conditioning and independence, Eq.  \ref{eqn:Prob_L_nu} separates to:
\begin{equation} \label{eqn:Prob_L_nu_simplified}
		\prob(L_{\nu,i}  \geq  L_{\mathrm{thr, i}}   \mid  \rmsi,\derri, d_i, e=1) =  
		\displaystyle{\int\limits_{0}^{d_{\mathrm{thr},{\mathrm{max_i}}}}}   
			\pdf(x_i \mid \derri, d_i)  
		\displaystyle{\int\limits_{L_{\mathrm{thr, i}}}^{ L_{\nu,\max} } }   
			\pdf(L_{\nu,i} \mid e=1)  \quad dL_{\nu,i}  \,dx_i  \quad,
\end{equation}
and we assume that the errors on the measured distance are Gaussian, such that 
\begin{equation}\label{eqn:gauss_distance}
 \pdf(x_i \mid d_i, \derri)  	= \frac{1}{\derri \sqrt{2\pi}} e^{-\frac{1}{2} \big(\frac{x_i-d_i}{\derri}\big)^2}  
 							 	= \normal(x_i; d_i, \derri)  \quad .  
\end{equation}

Some objects only have estimated distance ranges available.  For these objects, we assume a uniform distribution over the reported distance range $\Delta d_i = d_{i,\mathrm{max}} - d_{i,\mathrm{min}}$ such that
\begin{equation}
\pdf(x_i \mid d_i) = 
\begin{cases} 
	1 / \Delta d_i & \, , \,\,  d_{i,\mathrm{min}} \leq x_i \leq d_{i,\mathrm{max}} \\[5pt]
	0  	& \, , \,\,  \mathrm{otherwise}
\end{cases}
\quad.
\end{equation}

\subsubsection{Defining the luminosity distribution of emitting objects $\pdf(L_{\nu,i} \mid  e=1)$} \label{sec:lum_distrib}
\setlength{\tabcolsep}{0.07in}
	\begin{table*}\centering 
		\begin{ThreePartTable}
			\caption{Specific luminosities for detected quiescent radio emission in single-object ultracool dwarf systems \label{table:quiescent}}
			\begin{tabularx}{\textwidth}{lclr@{\hspace{0.01in}}c@{\hspace{0.01in}}lr@{\hspace{0.01in}}c@{\hspace{0.01in}}lcr@{\hspace{0.01in}}c@{\hspace{0.01in}}lrl}
				\toprule \vspace{2pt}
Object Name						&	
Other Name						&	
SpT								&	
\multicolumn{3}{c}{$\pi$}		&	
\multicolumn{3}{c}{$d$}			&	
ref								&	
\multicolumn{3}{c}{$F_{\nu}$}	&	
$\log_{10}(L_{\nu})$ 			&
ref								\\[-3pt]		
								&	
								&
								&
\multicolumn{3}{c}{(mas)}		&
\multicolumn{3}{c}{(pc)}		&
								&	
\multicolumn{3}{c}{($\mu$Jy)}	&	
								&	
								\\%
\midrule  
2MASS J00242463-0158201	&	BRI 0021-0214			&	M9.5V	&	79.9653		&	$\pm$	&	0.2212	&	12.5054	&	$\pm$	&	0.0346	&	23	&	25		&	$\pm$	&	15	&	12.7	&	2	\\
2MASS J00361617+1821104	&							&	L3.5	&	114.4167	&	$\pm$	&	0.2088	&	8.7400	&	$\pm$	&	0.0159	&	23	&	134		&	$\pm$	&	16	&	13.1	&	3	\\
2MASS J00361617+1821104	&							&	L3.5	&	114.4167	&	$\pm$	&	0.2088	&	8.7400	&	$\pm$	&	0.0159	&	23	&	152		&	$\pm$	&	9	&	13.1	&	3	\\
2MASS J00361617+1821104	&	LSPM J0036+1821			&	L3.5	&	114.4167	&	$\pm$	&	0.2088	&	8.7400	&	$\pm$	&	0.0159	&	23	&	240		&	$\pm$	&	11	&	13.3	&	2	\\
2MASS J00361617+1821104	&							&	L3.5	&	114.4167	&	$\pm$	&	0.2088	&	8.7400	&	$\pm$	&	0.0159	&	23	&	259		&	$\pm$	&	19	&	13.4	&	3	\\
2MASS J01365662+0933473	&	SIMP J013656.5+093347.3	&	T2.5	&	163.6824	&	$\pm$	&	0.7223	&	6.1094	&	$\pm$	&	0.0270	&	23	&	11.5	&	$\pm$	&	1.2	&	11.7	&	16	\\
2MASS J01365662+0933473	&	SIMP J013656.5+093347.3	&	T2.5	&	163.6824	&	$\pm$	&	0.7223	&	6.1094	&	$\pm$	&	0.0270	&	23	&	33.3	&	$\pm$	&	5.9	&	12.2	&	15	\\
2MASS J03393521-3525440	&	LP 944-20				&	L0		&	155.7590	&	$\pm$	&	0.0991	&	6.4202	&	$\pm$	&	0.0041	&	23	&	75		&	$\pm$	&	23	&	12.6	&	1	\\
2MASS J03393521-3525440 &	LP 944-20				&	L0		&	155.7590	&	$\pm$	&	0.0991	&	6.4202	&	$\pm$	&	0.0041	&	23	&	90.7	&	$\pm$	&	17.6&	12.7	&	17	\\
2MASS J03393521-3525440 &	LP 944-20				&	L0		&	155.7590	&	$\pm$	&	0.0991	&	6.4202	&	$\pm$	&	0.0041	&	23	&	137.6	&	$\pm$	&	10.5&	12.8	&	17	\\
2MASS J07464256+2000321	&	LSPM J0746+2000			&	L0.5V	&	80.9		&	$\pm$	&	0.8		&	12.4	&	$\pm$	&	0.1		&	10	&	149		&	$\pm$	&	15	&	13.4	&	7	\\
2MASS J07464256+2000321	&	LSPM J0746+2000			&	L0.5V	&	80.9		&	$\pm$	&	0.8		&	12.4	&	$\pm$	&	0.1		&	10	&	224		&	$\pm$	&	15	&	13.6	&	7	\\
2MASS J10430758+2225236	&							&	L8.5	&				&	--		&			&	16.4	&	$\pm$	&	3.2		&	20	&	16.3	&	$\pm$	&	2.5	&	12.7	&	15	\\
2MASS J10475385+2124234	&							&	T6.5	&	94.73		&	$\pm$	&	3.81	&	10.56	&	$\pm$	&	0.42	&	11	&	7.4		&	$\pm$	&	2.2	&	12.0	&	16	\\
2MASS J10475385+2124234	&							&	T6.5	&	94.73		&	$\pm$	&	3.81	&	10.56	&	$\pm$	&	0.42	&	11	&	9.3		&	$\pm$	&	1.5	&	12.1	&	22	\\
2MASS J10475385+2124234	&							&	T6.5	&	94.73		&	$\pm$	&	3.81	&	10.56	&	$\pm$	&	0.42	&	11	&	16.5	&	$\pm$	&	5.1	&	12.3	&	21	\\
2MASS J10475385+2124234	&							&	T6.5	&	94.73		&	$\pm$	&	3.81	&	10.56	&	$\pm$	&	0.42	&	11	&	17.5	&	$\pm$	&	3.6	&	12.4	&	15	\\
2MASS J10481463-3956062	&	DENIS J1048.0-3956		&	M8.5Ve:	&	247.2157	&	$\pm$	&	0.1236	&	4.0451	&	$\pm$	&	0.0020	&	23	&	131.4	&	$\pm$	&	10.8&	12.4	&	17	\\
2MASS J10481463-3956062	&	DENIS J1048.0-3956		&	M8.5Ve:	&	247.2157	&	$\pm$	&	0.1236	&	4.0451	&	$\pm$	&	0.0020	&	23	&	140		&	$\pm$	&	40	&	12.4	&	8	\\
2MASS J10481463-3956062	&	DENIS J1048.0-3956		&	M8.5Ve:	&	247.2157	&	$\pm$	&	0.1236	&	4.0451	&	$\pm$	&	0.0020	&	23	&	211.9	&	$\pm$	&	8.2	&	12.6	&	17	\\
2MASS J12373919+6526148	&							&	T6.5	&	96.07		&	$\pm$	&	4.78	&	10.41	&	$\pm$	&	0.52	&	11	&	27.8	&	$\pm$	&	1.3	&	12.6	&	16	\\
2MASS J12373919+6526148	&							&	T6.5	&	96.07		&	$\pm$	&	4.78	&	10.41	&	$\pm$	&	0.52	&	11	&	43.3	&	$\pm$	&	7.3	&	12.7	&	15	\\
2MASS J14563831-2809473	&	LHS 3003, LP 914-54		&	M7.0Ve	&	141.6865	&	$\pm$	&	0.1063	&	7.0578	&	$\pm$	&	0.0053	&	23	&	270		&	$\pm$	&	40	&	13.2	&	8	\\
2MASS J15010818+2250020	&	TVLM 513-46546			&	M8.5V	&	93.4497		&	$\pm$	&	0.1945	&	10.7009	&	$\pm$	&	0.0223	&	23	&	190		&	$\pm$	&	15	&	13.4	&	2	\\
2MASS J15010818+2250020	&	TVLM 513-46546			&	M8.5V	&	93.4497		&	$\pm$	&	0.1945	&	10.7009	&	$\pm$	&	0.0223	&	23	&	\multicolumn{3}{c}{$\sim$200$^{a}$}	&	$\sim$13.4	&	14	\\
2MASS J15010818+2250020	&	TVLM 513-46546			&	M8.5V	&	93.4497		&	$\pm$	&	0.1945	&	10.7009	&	$\pm$	&	0.0223	&	23	&	\multicolumn{3}{c}{$\sim$200$^{a}$}	&	$\sim$13.4	&	14	\\
2MASS J15010818+2250020	&	TVLM 513-46546			&	M8.5V	&	93.4497		&	$\pm$	&	0.1945	&	10.7009	&	$\pm$	&	0.0223	&	23	&	208		&	$\pm$	&	18		&	13.5	&	5	\\
2MASS J15010818+2250020	&	TVLM 513-46546			&	M8.5V	&	93.4497		&	$\pm$	&	0.1945	&	10.7009	&	$\pm$	&	0.0223	&	23	&	228		&	$\pm$	&	11$^{b}$	&	13.5	&	18	\\
2MASS J15010818+2250020	&	TVLM 513-46546			&	M8.5V	&	93.4497		&	$\pm$	&	0.1945	&	10.7009	&	$\pm$	&	0.0223	&	23	&	260		&	$\pm$	&	46$^{b}$	&	13.6	&	18	\\
2MASS J15010818+2250020	&	TVLM 513-46546			&	M8.5V	&	93.4497		&	$\pm$	&	0.1945	&	10.7009	&	$\pm$	&	0.0223	&	23	&	284		&	$\pm$	&	13$^{b}$	&	13.6	&	18	\\
2MASS J17502484-0016151	&							&	L5		&	108.2676	&	$\pm$	&	0.2552	&	9.2364	&	$\pm$	&	0.0218	&	23	&	56.4	&	$\pm$	&	5.5	&	12.8	&	19	\\
2MASS J18353790+3259545$^{c}$	&	LSR J1835+3259	&	M8.5V	&	175.8244	&	$\pm$	&	0.0905	&	5.6875	&	$\pm$	&	0.0029	&	23	&	464		&	$\pm$	&	10	&	13.3	&	6	\\
2MASS J18353790+3259545$^{c}$	&	LSR J1835+3259	&	M8.5V	&	175.8244	&	$\pm$	&	0.0905	&	5.6875	&	$\pm$	&	0.0029	&	23	&	525		&	$\pm$	&	15	&	13.3	&	4	\\
\bottomrule
\end{tabularx}
\begin{tablenotes}[]\footnotesize
\item[]\textit{Note} --- This table does not include a measurement for the T6 dwarf WISEP J112254.73+255021.5 and the L2.5 dwarf 2MASS J05233822-1403022.  Flaring radio emission was detected from the T6 dwarf WISEP J112254.73+255021.5 \citep{RouteWolszczan2016ApJ...821L..21R} but \citet{Williams2017ApJ...834..117W} were unable to measure a decisive quiescent radio flux density due to ongoing low-level variability and the possibility of polarization leakage. For the L2.5 dwarf 2MASS J05233822-1403022, \citet{Antonova2007AA...472..257A} report that its radio emission can vary by a factor of $\sim$5 from $\leq$45 to 230$\pm$17 $\mu$Jy with no evidence of short-duration flares during 2-hour observing blocks. Although the low circular polarization of this object rules out coherent aurorae, non-auroral flares at both radio and optical frequencies can persist for at least several hours \citep[e.g.][]{Villadsen2019ApJ...871..214V, Paudel2018ApJ...861...76P}. We exclude this outlier object on the basis of the uncertain nature of its radio emission.
\item[$a$]\citet{Hallinan2006ApJ...653..690H} do not give measured quiescent emission flux densities, but they note that the persistent radio emission does not drop below $\sim$200 $\mu$Jy for TVLM 513-46546.
\item[$b$]\citet{Osten2006ApJ...637..518O} note that low-level variability is present but not strong or periodic flares.
\item[$c$] We do not include quiescent emission flux densities measured from resolved imaging at 8.4~GHz \citep{Kao2023} or 4.9~GHz \citep{Climent2023arXiv230306453C}, since phase errors can reduce integrated measured flux densities.
\item[] \textit{References} ---			
(1)		\citet{Berger2001Natur.410..338B}	;
(2)		\citet{Berger2002ApJ...572..503B}	;
(3)		\citet{Berger2005ApJ...627..960B}	;
(4)		\citet{Berger2006ApJ...648..629B}	;
(5)		\citet{Berger2008ApJ...673.1080B}	;
(6)		\citet{Berger2008ApJ...676.1307B}	;
(7)		\citet{Berger2009ApJ...695..310B}	;
(8)		\citet{BurgasserPutnam2005ApJ...626..486B}	;
(9)		\citet{Dupuy2016ApJ...827...23D}	;
(10)		\citet{Dupuy2017ApJS..231...15D}	;
(11)		\citet{Faherty2012ApJ...752...56F}	;
(12)		\citet{Gauza2015ApJ...804...96G}	;
(13)		\citet{Guirado2018AA...610A..23G}	;
(14)		\citet{Hallinan2006ApJ...653..690H}	;
(15)		\citet{Kao2016ApJ...818...24K}	;
(16)		\citet{Kao2018ApJS..237...25K}	;
(17)		\citet{Lynch2016MNRAS.457.1224L}	;
(18)		\citet{Osten2006ApJ...637..518O}	;
(19)		\citet{Richey-Yowell2020}	;
(20)		\citet{Schmidt2010AJ....139.1808S}	;
(21)		\citet{Williams2013ApJ...767L..30W}	;
(22)		\citet{Williams2015ApJ...808..189W}	;
(23)		\citet{Gaia2018yCat.1345....0G}	
\end{tablenotes}
	\end{ThreePartTable}  
		\end{table*}

We compile all available measurements of quiescent radio emission from single-object ultracool dwarf systems in Table \ref{table:quiescent}, including repeat observations of the same target.  Fifteen single ultracool dwarfs spanning spectral types from M7--T6.5 have been detected at gigahertz radio frequencies, thirteen of which have measurements of quiescent radio emission.  

These low numbers are insufficient for including a well-constrained spectral type dependence for the probability density distribution $\pdf( L_{\nu}  \mid e = 1)$ of an emitting object's luminosity. Furthermore, ultracool dwarf radio emission may depend on other object properties, such as rotation, age, or multiplicity.   No existing studies conclusively document such effects for ultracool dwarfs.  In follow-up papers, we investigate if and how these other factors influence quiescent radio emission.
Here, we consider two distributions  for emitting objects both to demonstrate the generalizability of our framework and to motivate our proof-of-concept science application, quiescent radio emission from ultracool dwarfs.

\begin{figure*}
	\begin{centering}
		\includegraphics[width=\columnwidth]{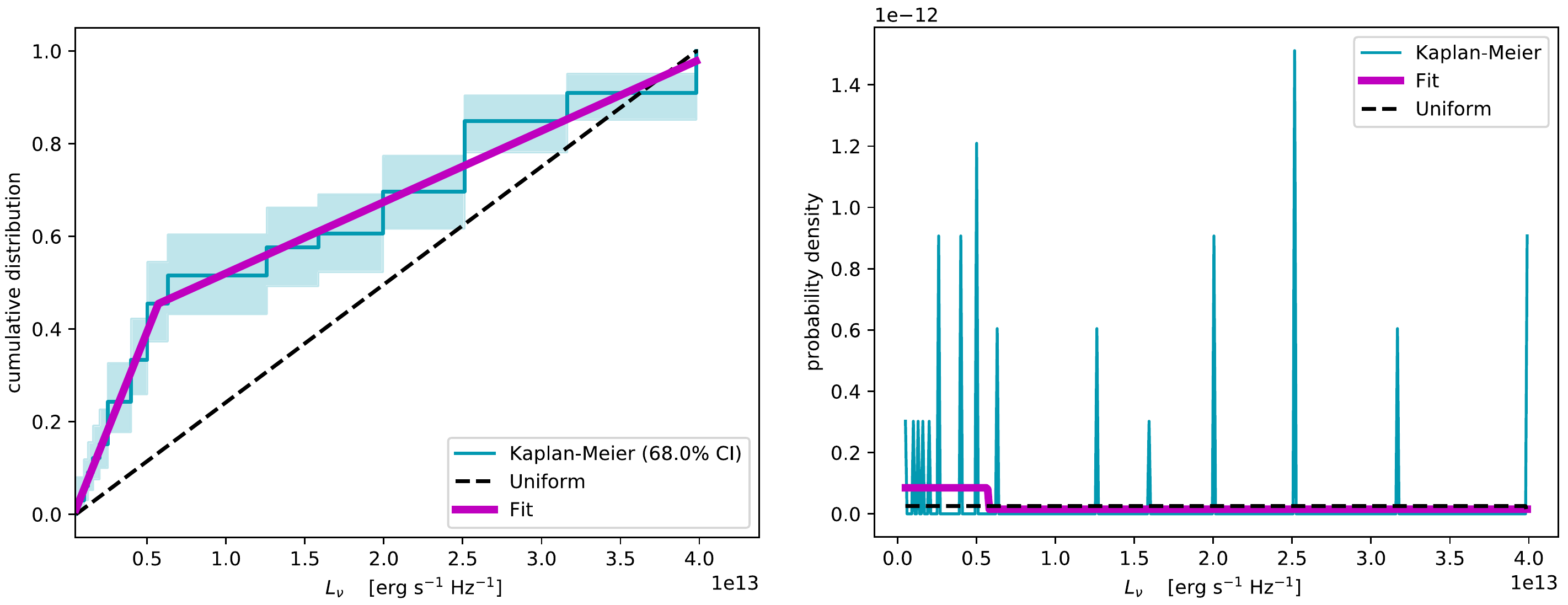}
	\end{centering}
   	\caption{\label{fig:lum_fit} ---Left: The empirical cumulative  distribution function for $L_{\nu}$ for all radio detections of ultracool dwarfs, calculated using the Kaplan-Meier estimator (teal). ---Right: The probability density distribution calculated from the cumulative distribution function. We run simulations for both the distribution fitted by a piecewise linear function (magenta) and a uniform distribution (black).}
\end{figure*}

First, we construct the empirical cumulative distribution function (CDF) for detected single-object luminosities using the Kaplan-Meier estimator. We then fit the empirical CDF with a piece-wise linear function and calculate the probability density distribution by differentiating the fitted CDF.  Figure \ref{fig:lum_fit} shows the CDF and $\pdf( L_{\nu}  \mid e = 1)$.  The Kaplan-Meier estimator is a non-parametric statistic used in survival analysis \citep{Kaplan-Meier1958, Lee_survival}.  For our science case, we can use it to  estimate the probability that an object ``survives" to a given luminosity. The Kaplan-Meier estimator can take into account upper limits by treating such measurements as left-censored data.  However, we do not do so because including non-detections in our calculation implies that all objects are emitting at some level, which may not be the case.  This estimator also captures differences in the unknown spectral behaviors of objects in our sample.  For the purposes of our calculations, employing the luminosity distribution calculated here assumes that emitting objects follow the same luminosity distribution of all \textit{measured} detections.

The small set of detected objects may not sample the true underlying luminosity distribution well, so we also consider a uniform distribution over the minimum and maximum of previously observed levels of quiescent radio emission in single ultracool dwarfs 
	\begin{align}\label{eqn:pdf_L}
		\pdf( & L_{\nu}  \mid e = 1) =  \\
		&
		\begin{cases}
			{\displaystyle \frac{1}{\Delta L} = \frac{1}{L_{\nu,\mathrm{max}} - L_{\nu,\mathrm{min}}} }  \, &, \,\,  L_{\nu,\mathrm{min}} \leq L_{\nu} \leq L_{\nu,\mathrm{max}} \\
			0  \, &, \,\, \text{otherwise} \,\, 
		\end{cases}
	\end{align} 
where $L_{\nu,\mathrm{min}} = 10^{11.7}$ and $ L_{\nu,\mathrm{max}} = 10^{13.6}$~erg~s$^{-1}$~Hz$^{-1}$, respectively.  This prior assumes that we have no information about how $ L_{\nu}$ is distributed within the interval bounds.

Notice that we do not have to make any assumptions about how the probability mass is distributed for  $L_{\nu} \not\in [L_{\nu,\mathrm{min}},   L_{\nu,\mathrm{max}}]$.  Instead, we have defined detect = 1 for objects with  $L_{\nu} \in [L_{\nu,\mathrm{min}},L_{\nu,\mathrm{max}} ]$. Consequently, our framework calculates the occurrence rate of quiescent emission between the assumed luminosity ranges.  

When setting the luminosity range, we must consider the fact that quiescent ultracool dwarf radio emission is not truly steady-state and in fact can vary (Table \ref{table:quiescent} and references therein).   Rather than modeling the time variability of this radio emission, we model it as a range of luminosities that an UCD can have during any given observation.  This approach is analogous to the data that we compile in \S \ref{sec.DataInclusion} for our proof-of-concept application; for each observation, we do not know exactly how its quiescent emission was varying or where in its variability cycle it was.  Instead, we assume that emitting objects have radio luminosities that fall somewhere between our luminosity range of interest, which encompasses all measurements of detected ultracool dwarf quiescent radio emission.  For this assumption to hold, ``on" objects must remain on and ``off" objects remain off. Observations spanning more than 10 years confirm that ultracool dwarf quiescent emission can persist for years-long timescales \citep[e.g.][]{Hallinan2006ApJ...653..690H, gawronski2017,  Hallinan2008ApJ...684..644H, Hallinan2015Natur.523..568H, Osten2009ApJ...700.1750O, Williams2013ApJ...767L..30W, Kao2016ApJ...818...24K, Kao2018ApJS..237...25K, Kao2023, Climent2023arXiv230306453C}.

\subsection{$\prob(d_i)$: Calculating the probability that the object lies at distance $d_i$ }
In the previous section, we employed a probability density distribution of an \textit{individual object's}  measured or estimated distance $\pdf(x_i \mid d_i, \derri)$.  In contrast, this section deals with the distribution that describes the distances of a \textit{population} of objects $\pdf(d)$ to calculate the probability that an individual object lies within some particular distance $\prob(d_i) = \prob(d \in [d_{\min,i},d_{\max,i}])$.

Ultracool dwarf number density studies use Poisson distributions to describe the space density of objects as a population \citep[e.g.][]{Gagne2017ApJS..228...18G, BardalezGagliuffi2019ApJ...883..205B}. Thus, we assume that the population of ultracool dwarfs are radially distributed in a Poisson manner and show in Appendix \ref{appendix:p(d)} that 
\begin{equation}
	\prob(d_i \in [d_{\min,i},d_{\max,i}]) = \vin / \vtot \, . 
\end{equation}

$\vin$ is the spherical shell within which an ultracool dwarf may lie given its distance $d_i$ and the uncertainties on its distance 
$\sigma_{d_{i}}$:
\begin{equation}
	\vin = \frac{4}{3} \pi \bigg( (d_i+ 3\derri)^3 - (d_i - 3\derri)^3 \bigg) \quad. 
\end{equation} 
Uniform distribution over a fixed interval is a fundamental Poisson assumption. However, we previously assumed that the errors on the measured distance are Gaussian. Since $>$99.7\% of the probability mass of a Gaussian distribution is within $\pm3\sigma$, we assume that an object is equally likely to be located anywhere within $\pm3\derri$ of its measured distance $d_i$.     
For objects with measured parallaxes $\pi_i$ and the uncertainty $\sigma_{\pi,i}$:
\begin{equation}
	\vin = \frac{4}{3} \pi \Bigg( \bigg(\frac{1}{\pi_i- 3\sigma_{\pi,i}}\bigg)^3 - \bigg(\frac{1}{\pi_i+ 3\sigma_{\pi,i}}\bigg)^3 \Bigg) \quad. 
\end{equation}

$\vtot$ is the total spherical volume encompassing all ultracool dwarfs under consideration, up to a maximum distance $d_{\mathrm{max}}$: 
\begin{equation}
	\vtot = \frac{4}{3} \pi d_{\mathrm{max}}^3  \,.
\end{equation}

 Current number density surveys extend only to $\sim$20-25 pc with incomplete volume coverage \citep[e.g.][]{Kirkpatrick2012ApJ...753..156K, BardalezGagliuffi2019ApJ...883..205B,Best2020AJ....159..257B}, because surveys are typically magnitude limited for ultracool dwarfs.
Consequently, number density may not remain constant with distance and our framework may need future modifications to account for ultracool dwarf number densities that vary with distance.  

\subsection{$\prob(\rmsi)$: Calculating the probability that the sensitivity for an observation is $\rmsi$ }

\begin{figure*}
	\begin{centering}
		\includegraphics[width=\columnwidth]{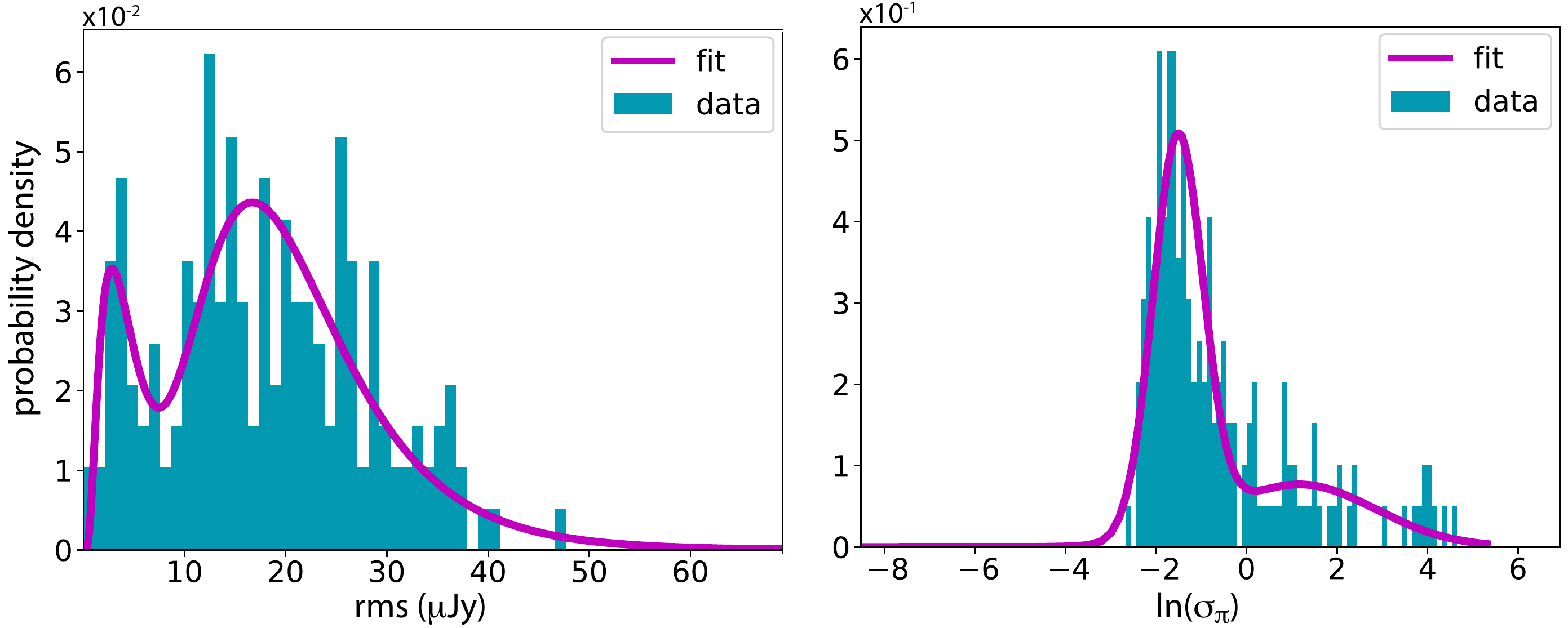}
		\end{centering}
		\caption{\label{fig:rms_fit} ---Left: The empirical probability density distribution for observation sensitivities  $	\pdf(\rms)$ and the fitted double-lognormal for M, L and T/Y ultracool dwarfs in single-object systems.   ---Right: The empirical density distribution for the natural logarithms of measured parallax errors  $\pdf(\ln \derri)$ and fitted Gaussian mixture model for the same population with histogram bin widths corresponding to $b$ as defined in Eq. \ref{eqn:prms}.}
\end{figure*}

We calculate $\prob(\rmsi)$ by first fitting a normalized analytic function $\pdf(\rms)$ to all $\rmsi$ for the full set of observations that we are considering.  The $\rms$ for any image is bounded by instrumental noise, such that $\pdf(\rms \leq 0) = 0$. We model this as a mixture of lognormal distributions to account for heterogenous datasets that may be dominated by several large surveys with individual observations subject to random sources of noise (e.g. telescope failures, weather, calibration errors, etc). A choice of lognormal distributions also enables the use of existing software tools, but readers are free to use other distributions as they see fit.  We therefore take the logarithm of our $\rms$ data and fit it with an n-component Gaussian mixture model
\begin{equation}
	\pdf(\ln \rms) = \sum_{k=1}^{n}  w_k \frac{1}{\sigma_{k} \sqrt{2 \pi} } \exp \bigg( -  \frac{1}{2} \left(\frac{\rms -\mu_k}{\sigma_k}\right)^2  \bigg) \quad,   
\end{equation}
with free parameters weight $w_k$, mean $\mu_k$, and standard deviation $\sigma_k$.
We use the \texttt{GaussianMixture} class from the Python package \texttt{sklearn.mixture} \citep{scikit-learn}, which implements expectation-maximization for fitting.    

Since we include data from many studies that use various telescopes, multiple Gaussian components may best represent our $\rms$ data.  We avoid over-fitting by selecting the number of components that correspond to the first local minimum in the Bayesian inference criterion (BIC)
\begin{equation*}
\mathrm{BIC}  = k \ln(n) - 2 \ln(\hat{L}) \quad, 
\end{equation*}
where $n$ is the number of data points, $k$ is the number of parameters in the Gaussian mixture model, and $\hat{L}$ is the maximized value of the likelihood function that measures the goodness of fit of our model to our $\rms$ data.  In most cases, 2 components are sufficient to fit our data. 

We then integrate over the characteristic length-scale $b$ of the fitted function, centered at $\rmsi$:
\begin{equation} \label{eqn:prms}
	\prob(\rmsi)  =  \int_{\ln \rmsi - b/2}^{\ln \rmsi + b/2}  \pdf(\ln \rms)  \quad d(\ln \rms) \quad.
\end{equation}
We define $b$ as one-fifth of the standard deviation of the narrowest component of our fitted Gaussian mixture model.   Varying the divisor between [1.0, 10.0] does not impact the final calculated maximum-likelihood occurrence rate.

Our fitting procedure chooses a summed double lognormal distribution (Figure \ref{fig:rms_fit}), which corresponds to a p-value of 1.0 for the KS test.  This indicates that the analytical and empirical distributions are indistinguishable. In manual fits, we find that a single lognormal returns a p-value of 0.0 for the KS test, indicating that it is a poor description of our data. This is expected, since we include data from studies before and after the VLA upgrade, which resulted in a factor of 10 increase in sensitivity.

\subsection{$\prob(\derri)$: Calculating the probability of that the uncertainty on the distance measurement is $\derri$}

We treat $\prob(\derri)$ using the same procedure that we outline for $\prob(\rmsi)$.  77.3\% of objects included in our samples have \textit{Gaia} parallaxes, and 13.3\% have parallaxes measured by \citet{Faherty2012ApJ...752...56F}.  5.8\% of literature objects in our sample do not have parallax measurements and instead have photometric distance estimates. For these objects, we convert estimated distance errors to parallax errors and include them in our distribution fit.  Our fitting procedure chose a summed double lognorm distribution (Figure \ref{fig:rms_fit}).

\subsection{Comparison to a similar existing Bayesian occurrence rate framework}
The occurrence rate framework that we present here has been independently developed from but is similar in spirit to that presented in \citet{Radigan2014ApJ...793...75R}, with which the authors calculated the occurrence rate of high-amplitude photometric O/IR variability in brown dwarfs spanning the L/T spectral type transition. However, two key differences generalize our framework to cases that do not have strong time dependence and ease of access to reduced data.  

The first difference is our treatment of calculating the probability than an object's  characteristic of interest, ($L_{\nu}$ in this work, variability amplitude in their work) is detectable. Rather than marginalizing over relevant random variables such as $\rms$, they model an empirical joint probability distribution by using injection tests of uniformly distributed synthetic signals into their raw reference star lightcurves and fitting an error function to the recovered detection fraction.  This approach is not feasible for the application that we present in this paper, since reducing hundreds of archival radio observations from multiple radio observatories is computationally intractable for most astronomers. Instead, we choose to marginalize over $\rms$, $d$, and $\derr$ (Eq. \ref{eqn:conditioning}) so that we can individually model each distribution and analytically treat the detectability.  Furthermore, marginalizing over these variables allows us to more easily consider different distributions for $L_{\nu}$, since we do not have to re-model the joint detectability function.  

Second, \citet{Radigan2014ApJ...793...75R} do not explicitly account for object distances.  Instead, they focus on relative variability amplitudes, which are analogous to our use of a steady $L_{\nu}$ that has been converted to flux density.  This amounts to implicit assumptions that all targets included in a sample have detectable fluxes and that the characteristic of interest does not depend on object distance.  While this may have been the case for \citet{Radigan2014ApJ...793...75R},  accounting for distance is essential in any general-purpose occurrence rate calculation. One example use case to which our more general-purpose framework applies is  ultracool dwarf magnetic activity studies that examine how H$\alpha$ detection fractions change with spectral type \citep[e.g.][]{Schmidt2015AJ....149..158S}.  In the future, our framework may also be applied to compiled studies of H$_3^+$ emission from hot Jupiter atmospheres \citep[e.g.][]{Shkolnik2006AJ....132.1267S, Lenz2016AA...589A..99L} as one means of studying the occurrence rate of conditions that are sufficient to produce aurorae in that population of planets.

\section{Model validation} \label{sec.modelValidation}


\subsection{Simulation set-up}\label{sec.sim_setup}
Our framework simplifies to a binomial distribution when observations are infinitely sensitive (i.e. rms~=~0~$\mu$Jy). In reality, sensitivity limits bias observations toward non-detections when emission can fall below such limits.  For faint ultracool dwarf quiescent radio emission, this is especially the case.  Our framework attempts to correct for this bias, so we expect our calculated occurrence rates to be systematically higher than observed detection rates.  To assess the fidelity of our recovered occurrence rate compared to the actual occurrence rate, we simulate suites of experiments with sample sizes of $N=[10, 20, 50, 100]$ observations.  For comparison, our ultracool M dwarf sample has 82 objects, our L dwarf sample has 74 objects, and our T/Y dwarf sample has 23 objects.    

We run 1000 trials for each sample size $N$ at each simulated true quiescent radio occurrence rate, which we vary from $\theta_{\mathrm{true}} = [0.1, 1.0]$ in increments of 0.1.  In a frequentist interpretation, the radio occurrence rate is the empirical emission rate, or the fraction of emitting objects in a given sample $\theta_{\mathrm{emit}} = n_{\text{emit}} / N$. In our framework, however, $\theta_{\text{true}}$ is the probability that an individual object is intrinsically emitting at radio frequencies.  Thus, $\theta_{\mathrm{emit}}$ can vary between samples for the same $\theta_{\text{true}}$. 

We fit distributions to the parallaxes, parallax errors, and observation sensitivities from the literature.  Using these fitted distributions, we randomly draw parameters for each synthetic observation in our simulated sample. We then randomly draw $\theta_{\mathrm{emit}}$ for each sample using a binomial distribution described by $\theta_{\mathrm{true}}$  and randomly select emitting objects in each sample to correspond to this emission rate.  For each emitting object, we randomly draw a radio luminosity from our assumed uniform distribution.  Based on these randomly drawn parameters, we assign a boolean for the detection of each object.

We run two sets of three simulation suites to assess how our occurrence rate model performs under different assumptions, for six total simulation suites.  In the first set of simulation suites, we explore the impact of luminosity priors on calculated occurrence rates by using a luminosity prior that is:
\begin{enumerate}
\item \textbf{KM:} Reconstructed using the Kaplan-Meier survival analysis estimator (Fig. \ref{fig:lum_fit}), where the luminosities vary between the minimum and maximum observed  quiescent radio luminosities for all ultracool dwarfs, such that $10^{11.7} \leq L_{\nu,i} \leq 10^{13.6}$ erg s$^{-1}$ Hz$^{-1}$. 
\item \textbf{Uniform:} A uniform distribution between the minimum and maximum observed quiescent radio luminosities in single-object systems (Eqn. \ref{eqn:pdf_L}). 
\item \textbf{Low luminosity:} A low luminosity distribution that is truncated to $10^{11.7} \leq L_{\nu,i} \leq 10^{12.7}$ erg s$^{-1}$ Hz$^{-1}$.  This luminosity range is consistent with reported quiescent radio luminosities for the set of three detected T dwarfs with measured quiescent radio luminosities (Table \ref{table:quiescent}). This distribution is uniform, since the break in the Kaplan-Meier distribution occurs at a higher luminosity than the maximum luminosity that we use for this case. 
\end{enumerate}
The second set of simulation suites explores a hypothetical dataset that has high observational sensitivities.  For this set, we use these same luminosity distributions but randomly draw from a uniform distribution of $\rmsi = [1.0, 3.0]$ $\mu$Jy and keep the distributions for the other parameters the same as the literature.  The sensitivity range that we choose is representative of typical sensitivities that VLA surveys can achieve today at 4--8 GHz frequencies for $\geq1$ hour observing blocks. 

Intriguingly, a recent 144 MHz detection of coherent emission from a T dwarf suggests that a population of T dwarfs with aurorae more luminous than previously detected at GHz frequencies may exist \citep{Vedantham2020ApJ...903L..33V}. At present, it is unknown if such a population would also produce similarly super-luminous quiescent radio emission. In the event that such a population exists at GHz frequencies, literature data suggest that they may be rare, as objects more luminous than existing detections would have been detected. As such, we do not expect a hypothetical population of super-luminous quiescent GHz radio emitters to meaningfully impact the luminosity priors that we implement.  However, we note that future observations may yet demonstrate that some objects exhibit very luminous quiescent radio emission, at which point the luminosity prior will automatically be updated with the inclusion of new data.

\begin{figure*}
	\begin{centering}
		\includegraphics[width=0.93\columnwidth]{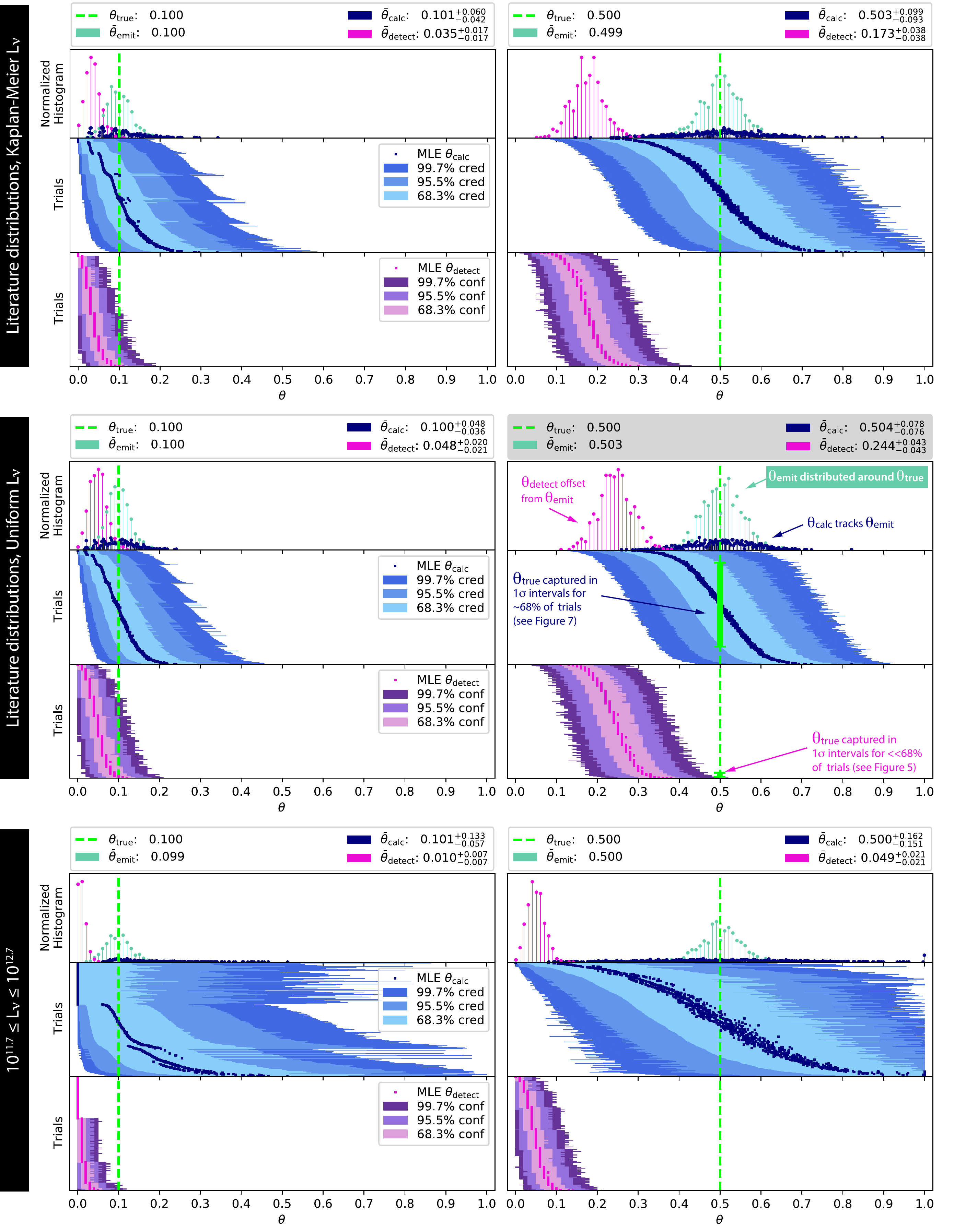}
			\caption{\label{fig:simulation_Lum} Simulations of 1000 trials per true quiescent radio occurrence rate $\theta_{\text{true}}$ with $N=100$ synthetic observations per trial using literature distributions for $d$, $\derri$, and $\rms$. Luminosity distributions vary by row: (top) KM literature distribution defined in \S\ref{sec.sim_setup}; (middle) uniform distribution between  1.0 $\mu$Jy $\leq \rms \leq 3.0$ $\mu$Jy; and (bottom) a low-luminosity distribution between $11.7 \leq [L_{\nu}] \leq 12.7$ 10$^{13.6}$ erg s$^{-1}$ Hz$^{-1}$. \,\,
			\textbf{Individual plot panels:} (top) Histograms and mean values of empirical emission rates $\theta_{\text{emit}}$ (green), detection rates $\theta_{\text{detect}}$ (magenta), and maximum-likelihood calculated occurrence rates $\theta_{\text{calc}}$ (blue).  (middle) $\theta_{\text{calc}}$ and 68.3\%, 95.5\%, and 99.7\% credible intervals for each trial. Trials are ordered by their lower 68.3\% credible intervals to show the distribution of $\pdf(\theta_{\text{calc}} \mid D)$.  (bottom)  $\theta_{\text{detect}}$ and bootstrapped confidence intervals, calculated from 1000 trials each, for each trial.  \textbf{Interpretation key:} To aid the reader in interpreting these plot panels, we annotate the middle-right panel with the grey legend.}  
	\end{centering}
\end{figure*}

\begin{figure*}
	\begin{centering}
		\includegraphics[width=0.93\columnwidth]{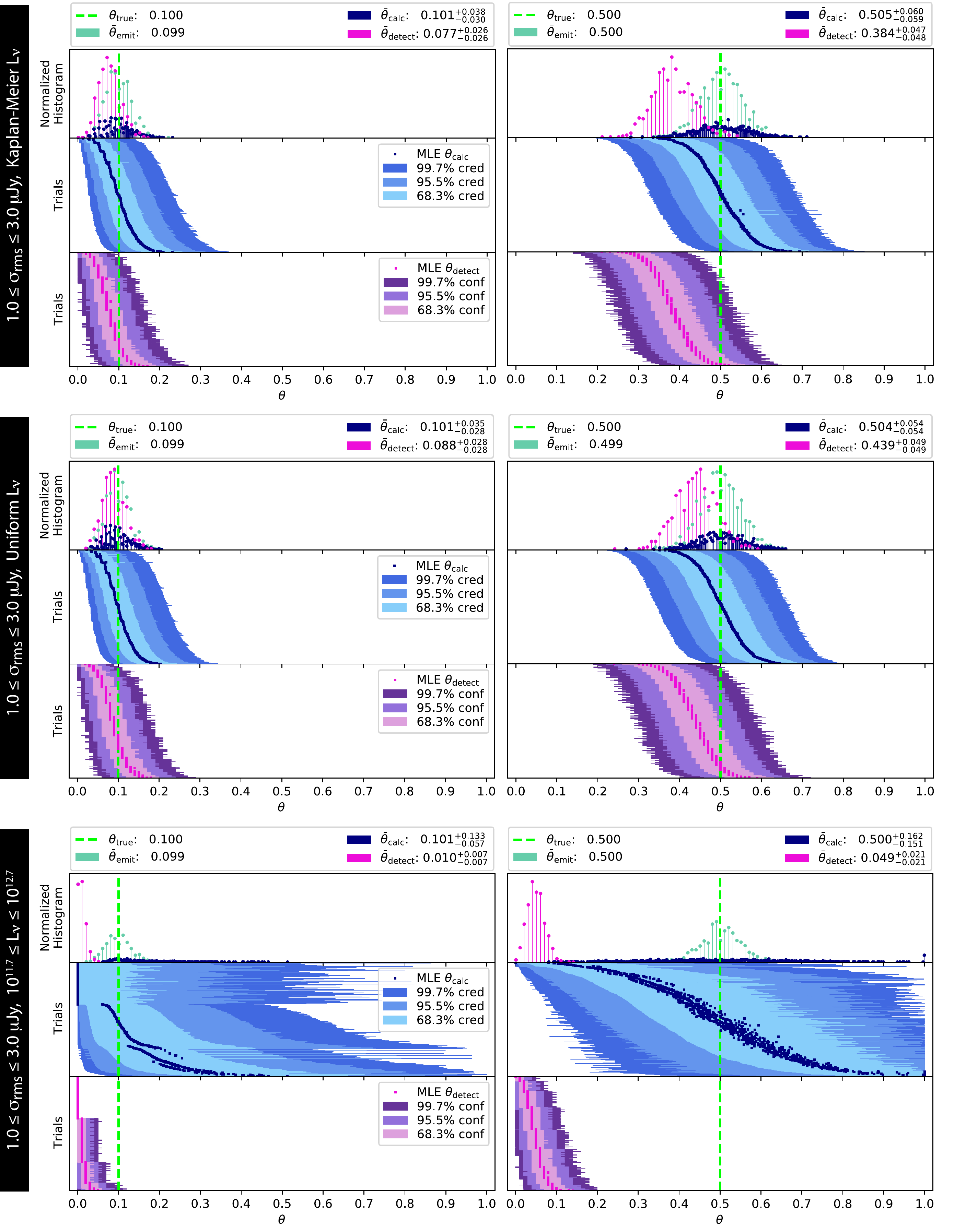}
		\caption{\label{fig:simulation_rms} The same as Figure \ref{fig:simulation_Lum} but for $\rms \in [1.0, 3.0]$ $\mu$Jy.  For an interpretation key, see Figure \ref{fig:simulation_Lum}. }
	\end{centering}
\end{figure*}

\begin{figure*}
	\begin{centering}
		\includegraphics[width=\columnwidth]{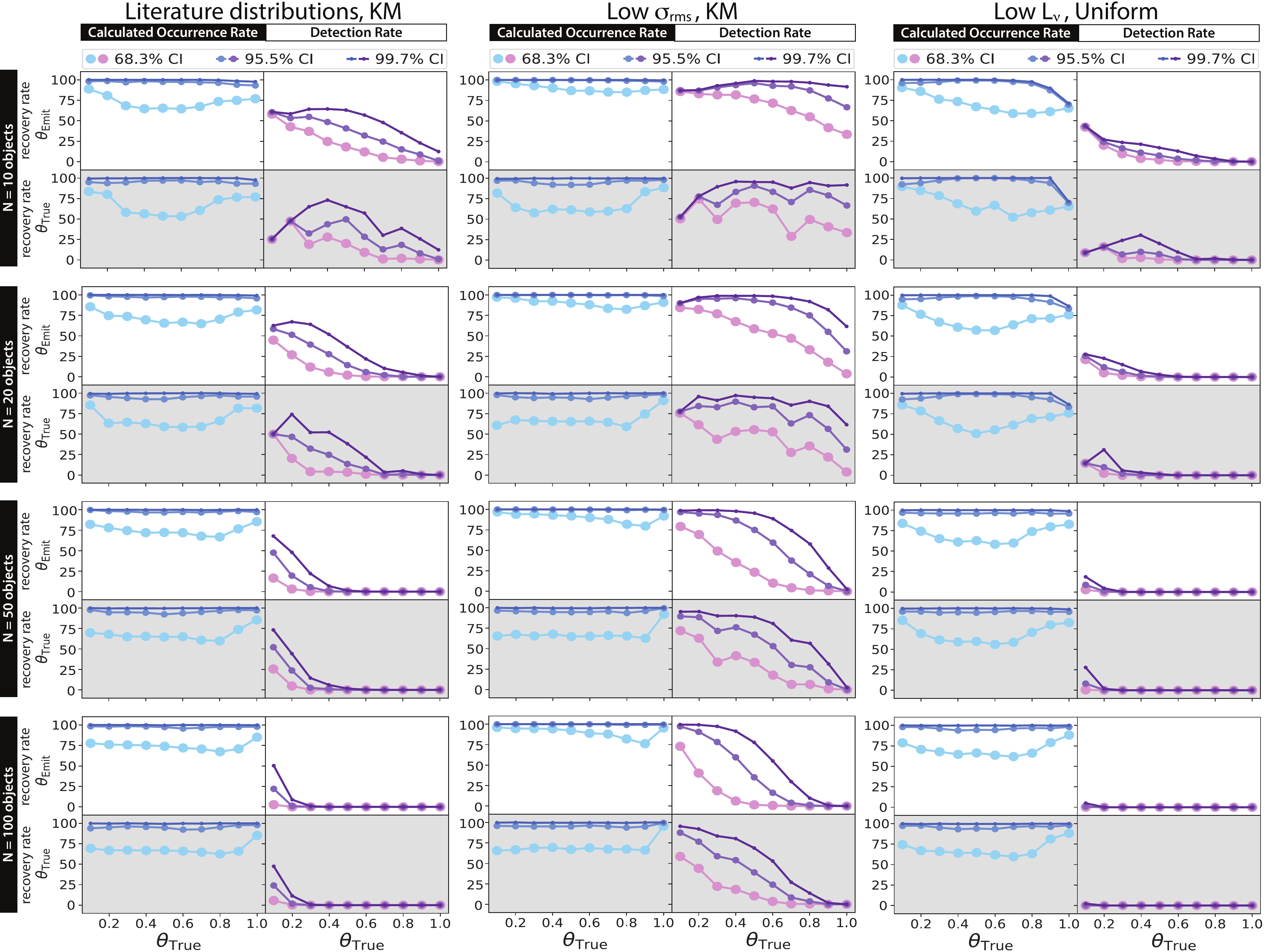}
		\caption{\label{fig:ci_performance} The mean performance of the 68.3\%, 95.5\%, and 99.7\% credible intervals (blue, left panels) and confidence intervals (purple, right panels) of the calculated occurrence and detection rates, respectively, in recovering simulated emission rates ($\theta_{\text{emit}}$, white panels) and the simulated true occurrence rates ($\theta_{\text{true}}$, grey panels) for selected simulation suites.  }
	\end{centering}
\end{figure*}

\setlength{\tabcolsep}{0.075in}
	\begin{table*}\centering 
		\begin{ThreePartTable}
			\caption{Summary of mean errors, correction factors, and recovery rates by credible or confidence intervals (68.3\%, 95.5\%, and 99.7\%), averaged over all $\theta_{\mathrm{true}}$ \label{table:mean_CI}}
			\begin{tabularx}{\textwidth}{llc@{\hspace{3pt}}c@{\hspace{3pt}}c@{\hspace{3pt}}cr@{\hspace{1pt}}r@{\hspace{1pt}}rr@{\hspace{1pt}}r@{\hspace{1pt}}rr@{\hspace{1pt}}r@{\hspace{1pt}}rr@{\hspace{1pt}}r@{\hspace{1pt}}r}
				\toprule \vspace{2pt}

   	& 
   	& 
   	& 
   	& 
   	& 
   	& 
\multicolumn{6}{l}{ $\theta_{\mathrm{emit}}$ }   	& 
\multicolumn{6}{l}{ $\theta_{\mathrm{true}}$  }   	 \\[-3pt]
Simulation suite  					& 
$N$    								&
$\abs{\epsilon_{\mathrm{calc}}}$   	& 
$\abs{\epsilon_{\mathrm{detect}}}$ 	& 
$f_{\mathrm{calc}}$    				& 
$f_{\mathrm{detect}}$  				& 
\multicolumn{3}{l}{ $\theta_{\mathrm{calc}}$ }   	&
\multicolumn{3}{l}{ $\theta_{\mathrm{detect}}$ }   	&
\multicolumn{3}{l}{ $\theta_{\mathrm{calc}}$ }   	&
\multicolumn{3}{l}{ $\theta_{\mathrm{detect}}$ }   	\\
   			& 
   			& 
 (\%)  		& 
 (\%)  		& 
 (\%)  		& 
 (\%)  		& 
\multicolumn{3}{l}{ (\%) } & 
\multicolumn{3}{l}{ (\%) } & 
\multicolumn{3}{l}{ (\%) } & 
\multicolumn{3}{l}{ (\%) }  \\

\midrule 
 Literature KM 		&  10  	&  5 &  65 &  1.0 &  2.9  	  &   72.5,  &  96.8,  & 99.6    &   19.8,  & 34.0, & 48.6     &    67.3, & 95.2, & 99.6    &  15.2, & 26.9, & 44.0   	\\[2pt] 
 					&  20 	&  3 &  65 &  1.0 &  2.9  	  &   73.1,  &  97.5,  & 99.8    &   9.1,  & 20.0, & 32.3     &    68.2, & 95.2, & 99.6    &  8.3, & 17.4, & 29.9   	\\[2pt] 
  					&  50 	&  1 &  65 &  1.0 &  2.9  	  &   74.8,  &  97.6,  & 99.9    &   2.0,  & 7.4, & 14.5     &    67.8, & 95.4, & 99.7    &  3.1, & 7.9, & 14.2   		\\[2pt] 
  					&  100 	&  1 &  66 &  1.0 &  2.9  	  &   74.5,  &  97.6,  & 99.9    &   0.3,  & 2.5, & 5.8     &    68.1, & 95.1, & 99.7    &  0.6, & 2.6, & 5.9   \\[2pt] 
\midrule
 Literature Uniform &  10  	&  3 &  51 &  1.0 &  2.0  	  &   77.7,  &  97.9,  & 99.9    &   30.7,  & 51.5, & 68.7     &    68.4, & 95.3, & 99.6    &  25.0, & 43.0, & 62.9   \\[2pt] 
  					&  20 	&  2 &  51 &  1.0 &  2.1  	  &   78.4,  &  98.3,  & 99.9    &   16.0,  & 34.1, & 50.7     &    67.6, & 95.5, & 99.6    &  15.1, & 31.1, & 46.9   \\[2pt] 
  					&  50 	&  1 &  51 &  1.0 &  2.0  	  &   79.9,  &  98.7,  & 100.0    &   5.0,  & 15.8, & 27.9     &    69.3, & 95.8, & 99.8    &  6.8, & 15.7, & 25.7   \\[2pt] 
  					&  100 	&  1 &  52 &  1.0 &  2.1  	  &   79.2,  &  98.6,  & 100.0    &   1.4,  & 6.9, & 13.7     &    69.0, & 95.6, & 99.8    &  2.0, & 7.3, & 13.0   \\[2pt] 
\midrule
 Low $L_{\nu}$ Uniform &  10  &  34 &  90 &  1.5 &  10.1  	  &   70.1,  &  93.5,  & 95.6    &   8.0,  & 11.1, & 15.9     &    68.9, & 94.4, & 97.0    &  3.2, & 5.1, & 11.3   \\[2pt] 
  					&  20 	 &  12 &  90 &  1.1 &  10.2  	  &   69.1,  &  94.8,  & 98.4    &   2.7,  & 4.8, & 7.8     &    66.3, & 95.1, & 98.5    &  1.9, & 3.1, & 6.2   \\[2pt] 
  					&  50 	 &  2 &  90 &  1.0 &  10.2  	  &   70.0,  &  96.2,  & 99.7    &   0.3,  & 1.1, & 2.5     &    68.6, & 95.7, & 99.6    &  0.1, & 0.7, & 3.1   \\[2pt] 
  					&  100 	 &  2 &  90 &  1.0 &  10.4  	  &   70.7,  &  96.6,  & 99.8    &   0.0,  & 0.1, & 0.4     &    69.2, & 95.8, & 99.8    &  0.0, & 0.0, & 0.2   \\[2pt] 
\midrule
 Literature + Low $\rms$  KM &  10  &  1 &  24 &  1.0 &  1.3  	  &   89.7,  &  99.5,  & 99.9    &   67.4,  & 87.3, & 94.1     &    68.1, & 95.2, & 99.7    &  52.9, & 76.2, & 87.5   \\[2pt] 
  					&  20 	 &  1 &  23 &  1.0 &  1.3  	  &   89.9,  &  99.6,  & 100.0    &   52.3,  & 80.3, & 91.4     &    68.0, & 95.5, & 99.8    &  43.2, & 72.3, & 87.1   \\[2pt] 
  					&  50 	 &  1 &  23 &  1.0 &  1.3  	  &   90.2,  &  99.7,  & 100.0    &   27.7,  & 56.9, & 74.1     &    68.2, & 95.6, & 99.6    &  27.2, & 51.3, & 69.0   \\[2pt] 
  					&  100 	 &  1 &  23 &  1.0 &  1.3  	  &   90.2,  &  99.7,  & 100.0    &   14.0,  & 38.2, & 55.7     &    70.2, & 95.6, & 99.7    &  15.7, & 35.2, & 51.3   \\[2pt] 
\midrule
 Literature + Low $\rms$ Uniform &  10  &  2 &  12 &  1.0 &  1.1  	  &   95.2,  &  99.8,  & 100.0    &   86.4,  & 95.8, & 97.3     &    68.9, & 95.2, & 99.7    &  66.4, & 83.7, & 89.7   \\[2pt] 
  &  20 	 &  1 &  12 &  1.0 &  1.1  	  &   95.6,  &  99.8,  & 100.0    &   79.3,  & 95.7, & 98.5     &    69.4, & 95.4, & 99.6    &  62.7, & 88.7, & 95.5   \\[2pt] 
  &  50 	 &  1 &  11 &  1.0 &  1.1  	  &   96.2,  &  99.9,  & 100.0    &   59.4,  & 84.8, & 93.2     &    69.4, & 95.6, & 99.7    &  47.9, & 76.7, & 90.2   \\[2pt] 
  &  100 	 &  1 &  12 &  1.0 &  1.1  	  &   96.3,  &  99.9,  & 100.0    &   40.4,  & 72.6, & 84.1     &    69.5, & 95.2, & 99.8    &  36.4, & 64.2, & 79.3   \\[2pt] 
\midrule
 Low $L_{\nu}$ + Low $\rms$ Uniform &  10  &  3 &  46 &  1.0 &  1.9  	  &   80.9,  &  98.6,  & 99.9    &   35.5,  & 58.0, & 74.5     &    68.3, & 95.3, & 99.7    &  29.1, & 49.8, & 69.2   \\[2pt] 
  &  20 	 &  2 &  46 &  1.0 &  1.8  	  &   81.8,  &  99.2,  & 100.0    &   20.1,  & 42.4, & 59.8     &    68.4, & 95.6, & 99.7    &  18.5, & 38.2, & 55.5   \\[2pt] 
  &  50 	 &  2 &  45 &  1.0 &  1.8  	  &   81.8,  &  99.1,  & 100.0    &   7.3,  & 20.8, & 34.7     &    69.0, & 95.6, & 99.7    &  9.1, & 20.5, & 32.3   \\[2pt] 
  &  100 	 &  2 &  45 &  1.0 &  1.8  	  &   81.9,  &  99.1,  & 100.0    &   2.4,  & 10.2, & 18.7     &    69.2, & 95.3, & 99.7    &  3.6, & 10.3, & 18.1   \\[2pt] 
\bottomrule
\end{tabularx}
\end{ThreePartTable}  
\end{table*}

To illustrate our simulations, we show normalized histograms for the empirical detection rate $\theta_{\mathrm{detect}}$ and our calculated maximum likelihood occurrence rate $\theta_{\mathrm{calc}}$ when $N=100$ in Figures \ref{fig:simulation_Lum}--\ref{fig:simulation_rms} compared to $\theta_{\mathrm{emit}}$ and $\theta_{\mathrm{true}}$.  We also show the distribution of 68.3\%, 95.5\%, and 99.7\% highest density credible intervals for each trial.  Finally, we show  $\theta_{\mathrm{detect}}$ and the corresponding confidence intervals for each trial, which we calculate with 1000 bootstrap samples per trial.   We summarize each simulation suite as a function of $\theta_{\mathrm{true}}$ and $N$ in  the Appendix Tables \ref{table:simulation_literature_kaplanMeier}--\ref{table:simulation_low-rms_lowLum}. We summarize the mean recovery rate for credible and confidence intervals in Table \ref{table:mean_CI}.

\subsection{Simulation results and implications}

The simulations show that our occurrence rate framework recovers the simulated true quiescent radio occurrence rate $\theta_{\mathrm{true}}$ well. 

In Figure~\ref{fig:ci_performance}, we show the fraction of trials for which our credible intervals recover $\theta_{\mathrm{emit}}$ and $\theta_{\mathrm{true}}$ for a subset of simulation suites, and we summarize the performance of all simulation suites averaged over  $\theta_{\mathrm{true}}$ in Table \ref{table:mean_CI}.  The left-most column in Figure \ref{fig:ci_performance} showcases the performance of our model when we assume a KM luminosity distribution and literature distributions for all other parameters, such as we do in the scientific application of our framework in  \S\ref{section:application}. This provides a baseline comparison to the other two columns and also corresponds to case (i) as defined in \S\ref{sec.sim_setup}. The middle column isolates the effect of low rms noise to demonstrate how our framework would perform for various survey sizes with noise floors that are achievable today.  The right-most column corresponds to case (iii) defined in \S\ref{sec.sim_setup} and shows the accuracy of our framework if a population of objects has systematically low luminosities, as may be the case for T/Y dwarfs. We include this column to help the reader assess our findings and discussion (\S\ref{section:application}, \S\ref{sec.Discussion}).

 In simulations drawn from literature distributions, our 68.3\% credible intervals recover $\theta_{\mathrm{emit}}$ in an average of $\geq$72.5\% of trials for sample sizes as small as $N=10$ .  This increases to a $\geq$99.6\% recovery rate for the  99.7\%  credible intervals. Note that credible intervals are distinct from confidence intervals. The latter correspond to the fraction of trials for which the interval contains $\theta_{\mathrm{emit}}$ but do not intrinsically yield information about $\theta_{\mathrm{true}}$.  In contrast, credible intervals correspond to the location of that probability mass fraction for each individual trial. A good Bayesian model, such as ours, will return a posterior distribution with credible intervals that recover $\theta_{\mathrm{emit}}$ more frequently than the corresponding confidence interval.   Consequently, Figure~\ref{fig:ci_performance} demonstrates that our credible intervals generally closely approximate confidence intervals for $\theta_{\mathrm{true}}$, where the 68.3\% credible intervals recover $\theta_{\mathrm{true}}$ in an average of $\geq$67.3\% of our trials.  It is important to note that the recovery rate varies as a function of $\theta_{\mathrm{true}}$, $N$, and the luminosity prior. We tabulate these values in Tables \ref{table:simulation_literature_kaplanMeier}--\ref{table:simulation_low-lum_uniform}.  In contrast, detection rate confidence intervals on average underpredict $\theta_{\mathrm{emit}}$ and thus poorly recover $\theta_{\mathrm{true}}$ at an average rate of no more $\leq$24.7\%, depending on the luminosity prior and the sample size.  

 In simulations drawn from literature parameter distributions, the distribution of calculated maximum-likelihood radio occurrence rates $\theta_{\mathrm{calc}}$ closely track that of emission rates $\theta_{\mathrm{emit}}$.  Since $\theta_{\mathrm{emit}}$ are distributed around $\theta_{\mathrm{true}}$, the mean error of $\theta_{\mathrm{calc}}$ with respect to $\theta_{\mathrm{true}}$ 

\begin{equation}
\abs{\bar{\epsilon}_{\text{calc}} } = \frac{1}{N} \sum_{n=1}^{N}   \frac{  \abs{  \theta_{\mathrm{true}} - \theta_{\mathrm{calc,n}}} }{\theta_{\mathrm{true}}}  
\end{equation}
is within 12\% for all trial sample sizes for the Kaplan-Meier luminosity distribution (Table \ref{table:simulation_literature_kaplanMeier}) and 6\% for the uniform luminosity distribution (Table \ref{table:simulation_literature_uniform}).
We similarly calculate a correction factor 
\begin{equation}
\bar{f}_{\text{calc}}  = \frac{1}{N} \sum_{n=1}^{N}   \frac{ \theta_{\mathrm{true}}}{\theta_{\mathrm{calc,n}}} 
\end{equation}
as well as the analogous quantities for $\theta_{\mathrm{detect}}$.  We find that the mean correction factor for the calculated rate $\bar{f}_{\text{calc}}$ averaged over all occurrence rates is 1.0 for both the Kaplan-Meier and uniform luminosity distributions for all sample sizes.
 
In contrast, the mean detection rate under-predicts $\theta_{\mathrm{true}}$ by up to 66\% or 52\% for Kaplan-Meier and uniform luminosity distributions respectively, such that  $\bar{f}_{\text{detect}} = 2.9$ or $2.0 \leq \bar{f}_{\text{detect}} \leq 2.1$. As expected, our trials show that calculated occurrence rates for individual experiments correlate with the number of detected objects.  For instance, trials with $\theta_{\mathrm{detect}} = 0\%$ detection rate have a maximum likelihood calculated occurrence rate $\theta_{\mathrm{calc}} = 0\%$. Similarly, trials with nonzero $\theta_{\mathrm{detect}}$ have $\theta_{\mathrm{calc}} \geq \theta_{\mathrm{detect}}$.  Finally, as is also expected, the sizes of the credible and confidence intervals decrease with increasing sample size. 

Thus far, our discussion has focused on an ensemble of trials to characterize $\theta_{\mathrm{true}}$.  However, in reality, we have only one real trial: the dataset consisting all existing observations for a population of objects. With well-performing credible intervals, our framework enables meaningful constraints on the true occurrence rate of any such finitely sampled population.  The caveat is that all objects in the population must have radio luminosities drawn from the same distribution for the results of our simulation to hold, since we have not yet modeled how luminosities may depend on properties such as youth, binarity, or rotation rate.  Characterizing the relevant object properties that impact an individual object's luminosity is therefore essential.  In \citet{Kao2020b} and \citet{Kao2020_binaries}, we explore how youth and binarity affect $\pdf(L_{\nu} \mid e=1)$, respectively.  

We expect the detection rates and calculated occurrence rates for individual trials to converge with increasing sensitivity. This is because higher sensitivities rule out a greater portion of the luminosity distribution, thus decreasing the probability that a non-detection is associated with a faint emitter.  We test this with our high sensitivity simulations and show that $\theta_{\mathrm{detect}}$ improves as an estimator for $\theta_{\mathrm{true}}$ at higher sensitivities, with the mean error decreasing from 52--66\% to $\abs{\bar{\epsilon}_{\text{detect}} } =$ 11--24\%, and the correction factor decreasing to  $\bar{f}_{\text{detect}} = 1.3$ or 1.1 for Kaplan-Meier or uniform luminosity priors, respectively. Nonetheless, $\theta_{\mathrm{calc}}$ remains the better estimator, as the mean error is 2--24$\times$ that for the detection rate and  $\bar{f}_{\text{calc}} = 1.0$  (Tables \ref{table:simulation_low-rms_kaplanMeier}--\ref{table:simulation_low-rms_uniform}). Finally, increasing observation sensitivity slightly decreases the sizes of the credible and confidence intervals. 

We caution against concluding that detection rates can serve as a good estimators for occurrence rates in the high sensitivity limit.  Instead, our simulations show that this is the case \textit{only for our specified $L_{\nu}$ range.}  Higher observational sensitivities will likely lead to detections of ultracool dwarf quiescent radio emission at lower luminosities than existing detections.  If this proves true, then detections rates may diverge from occurrence rates for those lower luminosities.  

The high sensitivity simulations demonstrate the accuracy and precision that we can expect with large ultracool dwarf surveys at modern VLA sensitivities. However, such surveys demand significant telescope allocations with current instrumentation, exceeding 200 hours simply to re-observe the existing dataset.  For comparison, the ngVLA will reach $\sim0.22$ $\mu$Jy sensitivities for 1 hour observing blocks\footnote{https://ngvla.nrao.edu/page/performance} when using the full array of 244 antennas. Making use of sub-array capabilities to simultaneously observe multiple objects at our simulated sensitivities will make it feasible to conduct detailed population studies in several tens of hours.

For a population that is fainter than our assumed luminosity distribution, we expected that our framework would over-predict $\theta_{\mathrm{true}}$. However, the low luminosity simulations show more nuanced behaviors. On the one hand, individual trials can significantly over-predict $\theta_{\mathrm{true}}$ when observation sensitivity and object distance cannot rule out large portions of luminosity space. On the other hand, the simulations also show a wide variance in $\theta_{\mathrm{calc}}$, including ones that strongly underpredict $\theta_{\mathrm{true}}$.  This occurs because the combination of low sensitivity and low luminosity can cause numerous observations to be poorly constrained.  When an observation cannot rule out any portion of the luminosity probability space,  that observation does not update the prior.  As a result, the simulations for low luminosity samples suffer from small number statistics, in which the calculated occurrence rate is dominated by the few objects that are detected or entirely ruled out.  This also leads to very wide credible intervals.

Despite larger variance in $\theta_{\mathrm{calc}}$ and the wider credible intervals for faint populations, our framework handily outperforms $\theta_{\mathrm{detect}}$ as an estimator for $\theta_{\mathrm{true}}$. When $N=10$,  $\bar{\theta}_{\mathrm{calc}}$ underpredicts $\theta_{\mathrm{true}}$ by a mean error of $\abs{\bar{\epsilon}_{\text{calc}} } = 34\%$. The performance of  $\bar{\theta}_{\mathrm{calc}}$ increases with increasing $N$, such that  $\bar{\theta}_{\mathrm{calc}}$ underpredicts $\theta_{\mathrm{true}}$ by only a mean error of $\abs{\bar{\epsilon}_{\text{calc}} } = 2\%$ for $N=100$.  This is expected, since larger sample sizes will reduce the effects of small number statistics that we discussed above.  In contrast,  $\bar{\theta}_{\mathrm{detect}}$ underpredicts $\theta_{\mathrm{true}}$ by $\sim$90\% regardless of sample size.   Thus, depending on the sample size,  $\bar{\theta}_{\mathrm{calc}}$ is  $\sim$2.6--45$\times$ more accurate than $\bar{\theta}_{\mathrm{detect}}$.  

Finally, we show in Table \ref{table:simulation_low-rms_lowLum} that increasing the sensitivity of observations to reflect modern VLA sensitivities dramatically improves the performance of   $\bar{\theta}_{\mathrm{calc}}$, such that  $\abs{\bar{\epsilon}_{\text{calc}} } =$ 2--3\%, while the error on detection rates remain high at $\abs{\bar{\epsilon}_{\text{detect}} } =$ 45-46\%. 

The results of these low luminosity simulations highlight that improving observational sensitivity from the current literature distribution will be one key area of constraining both the radio luminosities and occurrence rates for low luminosity populations.  Additionally, we expect the error between $\theta_{\mathrm{true}}$ and $\bar{\theta}_{\mathrm{calc}}$ to improve with the future inclusion of a more realistic radio luminosity function.

\section{Applying the framework}  \label{section:application} 

\begin{figure*}
	\begin{centering}
		\includegraphics[width=\columnwidth]{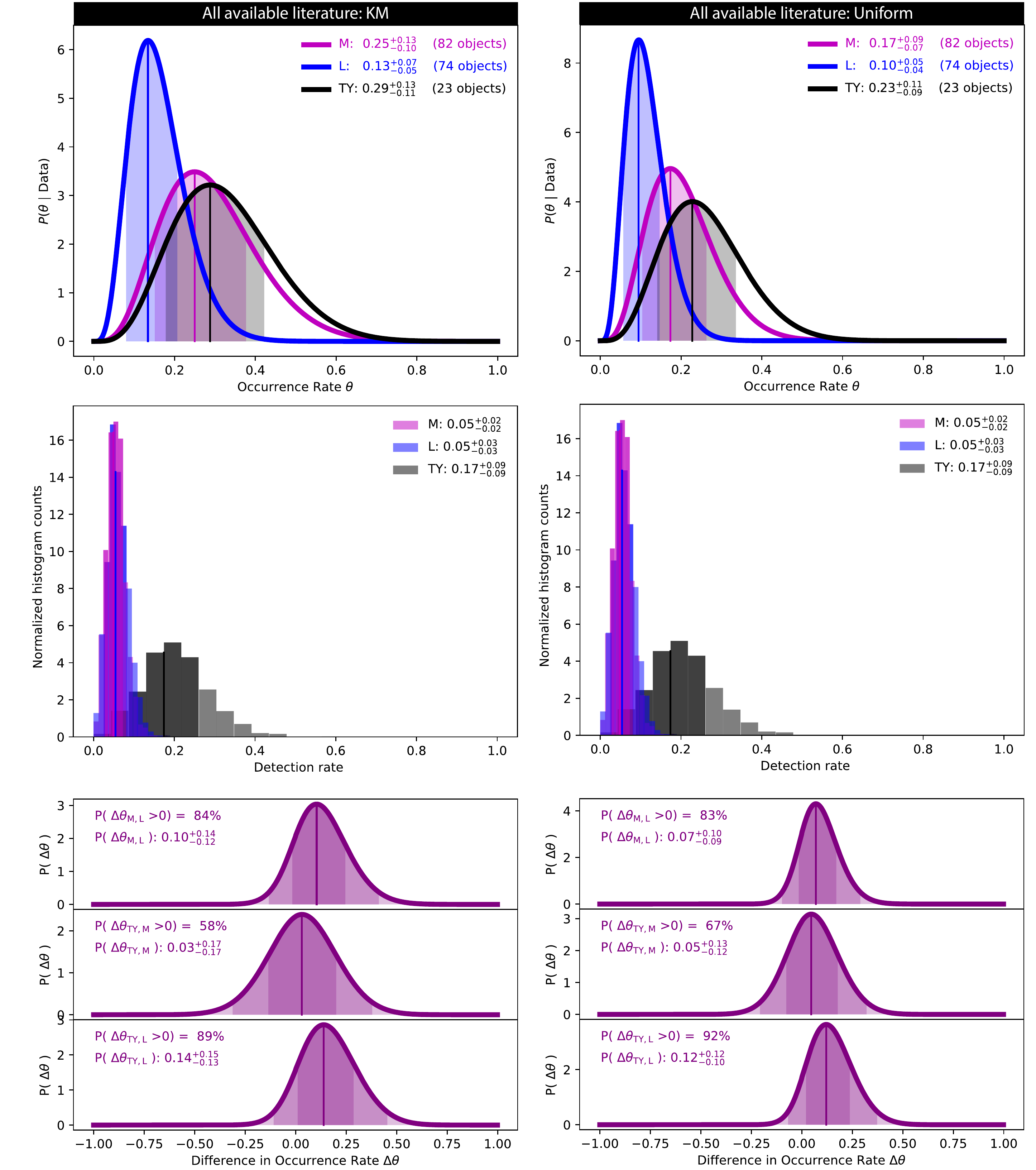}
		\caption{\label{fig:occurrenceRate}	
		Calculations using the full set of data available in the literature for a Kaplan-Meier luminosity prior (left) and a uniform luminosity prior (right).  \\
--- Top: Occurrence rate distributions of quiescent radio emission in M dwarfs (magenta), L dwarfs (blue) and T/Y dwarfs (black). Shaded regions correspond to the 68.3\% credible intervals. The L dwarf radio occurrence rate appears tentatively suppressed compared to M and T/Y dwarfs.  
--- Middle: Detection rate and bootstrapped distributions from 1000 trials.  Shaded regions correspond to the 68.3\% confidence intervals.  Relying on detection rates would not have identified the tentative suppression of the  L dwarf radio occurrence rate.  
--- Bottom: Probability density distributions for differences in occurrence rates $\Delta \theta$ between spectral type samples.  Shaded regions correspond to 68.3\%, 95.5\% and 99.7\% credible intervals.  Existing data cannot rule out the null hypothesis that the occurrence rates are the same for all spectral types but provide tentative evidence that ultracool dwarf magnetism may evolve over the L spectral type. 
  } 
	\end{centering}
\end{figure*}

Radio observations of ultracool dwarfs demonstrated that very strong magnetism can occur in very low mass and cold substellar objects \citep[e.g.]{Berger2001Natur.410..338B, RouteWolszczan2012ApJ...747L..22R, Kao2018ApJS..237...25K} despite predictions from suppressed X-ray emission \citep{Williams2014ApJ...785....9W}.  Such observations have also been instrumental in identifying the onset of a transition in magnetic activity from stellar flaring to more planet-like aurorae in the ultracool dwarf regime \citep{Hallinan2008ApJ...684..644H, Hallinan2015Natur.523..568H, Kao2016ApJ...818...24K, Pineda2017ApJ...846...75P}.  Finally, nonthermal radio emission is currently the only viable direct probe of magnetic fields for the coldest and lowest mass ultracool dwarfs. However, detection rates  from existing surveys remain stubbornly low, between $\sim$5--10\% \citep{Antonova2013AA...549A.131A, RouteWolszczan2013ApJ...773...18R, RouteWolszczan2016ApJ...830...85R, Lynch2016MNRAS.457.1224L, Richey-Yowell2020}. The sole exception is a small sample of 5 objects with which \citet{Kao2016ApJ...818...24K} achieved an 80\% detection rate by selecting targets for possible markers of aurorae. To date, only 20 ultracool dwarf systems, including both single objects and multiples, have been detected at gigahertz radio frequencies.  

Low detection rates severely hamper assessments of ultracool dwarf magnetospheric physics for analyses that rely on detected radio emission. However, our occurrence rate framework provides a path forward both by disentangling systematics from  intrinsic magnetic activity properties as well as by making use of reported non-detections. As a proof-of-concept,  we apply our framework to studying quiescent radio occurrence rates in ultracool dwarfs as a function of spectral type.

\subsection{Data inclusion policy}\label{sec.DataInclusion}
   We performed a literature search to compile a list of all radio observations of ultracool dwarfs ($\geq$M7).  Several objects have been observed multiple times, including at multiple frequencies and/or in both flaring and non-flaring states.  We identify detections of quiescent emission by inspecting all published radio detections. Most studies distinguish between flaring and quiescent emission, and when available, we use published timeseries and/or check for low circular polarization fractions characteristic of quiescent radio emission by comparing Stokes I and V flux densities. We note ambiguous cases in  Table \ref{table:literatureBD} and provide details that include our treatment of each ambiguous case to aid the reader in assessing the work that we present here.
  
    While our assumption that each observation is independent from other observations (Eq. \ref{eqn:p(D|theta)}) allows for individual ultracool dwarfs to be re-drawn for new observations, this treatment is only valid if observations are not follow-up observations for previous detections and instead were done in blind surveys.  However, less distant objects are more likely to be re-observed even in surveys and are not randomly selected for multiple observations.  We therefore elect to consider only one observation for each object using the following selection criteria for including data in our calculations:

\begin{enumerate}
\item We include only data from interferometric radio arrays.  Single dishes like Arecibo are confusion limited by background sources and are insensitive to ultracool dwarf incoherent radio emission.  
\item For objects with detections, we include only the observation with the lowest rms noise in which radio emission has been detected. 
\item For objects without detections, we include only the observation with the lowest rms noise. 
\end{enumerate}

Finally, we mitigate the impact of a sloped emission spectrum on the luminosity prior by including only observations centered within a narrow frequency range, between 4--10 GHz.  The instantaneous frequency-dependent emission spectra of GHz quiescent radio emission from single ultracool dwarfs are not well sampled at present \citep[e.g.][]{Williams2015ApJ...815...64W, Kao2018ApJS..237...25K, hughes2021AJ....162...43H} and will likely vary from object to object. However, recent resolved magnetospheric imaging of the ultracool dwarf LSR~J1835+3259 show that GHz quiescent emission can trace synchrotron radiation belts analogous to the Jovian radiation belts \citep{Kao2023, Climent2023arXiv230306453C}. Radiation belts around Jupiter produce relatively flat emission spectra with sharp frequency cut-offs \citep{Zarka2007PSS...55..598Z}, and similar behavior is observed for quiescent radio emission from massive stars proposed to be radiation belts \citep{Leto2021MNRAS.507.1979L, Owocki2022}.

Despite the unknown quiescent radio emission spectra of objects included in our sample, the non-parametric Kaplan-Meier estimator implicitly accounts for sloped emission spectra as it does not require knowledge about the distribution’s underlying specified parameters. This ensures that our framework remains resilient to evolving knowledge about emission spectra by allowing the user to define a data-driven luminosity distribution that improves with time as more observations are reported.

\subsection{Quiescent radio occurrence rate for ultracool dwarfs}
Using the full data set of 179 objects compiled, we find that the occurrence rate of ultracool dwarfs is between  $20^{+6}_{-5}$\% for a Kaplan-Meier luminosity prior and $15^{+4}_{-4}$\% for a uniform luminosity prior. This decrease occurs because the uniform luminosity distribution allows for a higher proportion of brighter objects and fewer dimmer objects, compared to the Kaplan-Meier distribution.  As a result, more objects meet the ``detectable" threshold, so a higher proportion of non-detections correspond to the ``off" state.  For comparison, the ultracool dwarf detection rate is $7\pm2$\%.  Quoted uncertainties correspond to 68.3\% credible intervals and confidence intervals obtained from 1000 bootstrap trials, respectively.

\subsection{Quiescent radio occurrence rates vs. spectral type}
Do radio occurrence rates change across spectral types?  To answer this question, we compare the quiescent radio occurrence rates for M, L and T ultracool dwarfs.  Table \ref{table:literatureBD} summarizes parallax or distance, rms noise, and measured flux densities for each object that we include in our calculations.

Spectral types available in the literature that we use here are identified with a mix of optical and near infrared spectra, which can result in an uncertainty of up to approximately $\pm2$ in spectral sub-type \citep{Kirkpatrick2005ARA&A..43..195K}. Applying our framework to multi-object systems requires additional theoretical development and consideration of assumptions. This application is outside of the scope of this work, so we exclude binaries from our samples and refer the reader to \citet{Kao2020_binaries} for a detailed treatment. This gives sample sizes of 82 M dwarfs, 74 L dwarfs, and 23 T/Y dwarfs. Less than $\sim$10\% of known late-type objects are young \citep{BardalezGagliuffi2014ApJ...794..143B}, so for this study we do not account for possible effects of age.

In Figure \ref{fig:occurrenceRate}, we show quiescent radio occurrence rates and detection rates for each sample. When we assume a Kaplan-Meier luminosity distribution and use the full set of radio observations that are available in the literature, we find that occurrence (detection) rates are $25^{+13}_{-10}$\% ($5\pm2$\%) for M dwarfs, $13^{+7}_{-5}$\% ($5\pm3$\%) for L dwarfs, and $29^{+13}_{-11}$\%  ($17\pm9$\%) for T/Y dwarfs.   When we assume a uniform luminosity distribution, the occurrence rates slightly decrease to  $17^{+9}_{-7}$\%  for M dwarfs, $10^{+5}_{-4}$\%  for L dwarfs, and $23^{+11}_{-9}$\% for T/Y dwarfs. 

Following \citet{Radigan2014ApJ...793...75R}, we also calculate the probability density distribution for differences in occurrence rates $\Delta \theta$ between spectral type samples
\begin{equation}
P(\Delta \theta) = \int_{-1}^{1} P(\theta_1) P(\theta_2 = \theta_1 + \Delta\theta) d \theta_1  \quad.
\end{equation} 
For a Kaplan-Meier (uniform) luminosity distribution, we find that M dwarfs may have a radio occurrence rate that exceeds that of L dwarfs with probability  
$\prob(\theta_{\text{M}} > \theta_{L}) = 84\%$ (83\%). The T/Y dwarf radio occurrence rate exceeds that of M dwarfs with probability 
$\prob(\theta_{\text{T/Y}} > \theta_{M}) = 58\%$ (67\%).  Finally the T/Y dwarf radio occurrence rate exceeds that of L dwarfs with probability 
$\prob(\theta_{\text{T/Y}} > \theta_{L}) = 89\%$ (92\%). For comparison, distributions that are centered around $\Delta \theta = 0$ will give probabilities of 50\%. 

We cannot rule out the null hypothesis that occurrence rates are the same for all spectral types, and we conclude that M and T/Y dwarfs likely have similar occurrence rates.  However, our calculations tentatively suggest that L dwarfs may have a suppressed quiescent radio occurrence rate compared to M and T/Y dwarfs.  We explore the possibility of a suppressed L dwarf radio occurrence rate further by repeating these calculations for additional datasets that we motivate in \S \ref{sec.Discussion}.  Table \ref{table:calculations} tabulates the results of each calculation.

Finally, we note that  Figure \ref{fig:occurrenceRate} illustrates the limitations of detection rate studies.  In particular, our calculated detection rates do not identify the possible L dwarf suppression, and bootstrapped confidence intervals are too narrow to account for uncertainties.

\setlength{\tabcolsep}{0.135in}
\begin{center}
	\begin{table*}\centering 
		\begin{ThreePartTable}
			\caption{Calculated occurrence rates \label{table:calculations}}
			\begin{tabularx}{\textwidth}{llr@{\hspace{0.01in}}lr@{\hspace{0.01in}}lr@{\hspace{0.01in}}lccc}
			\toprule \vspace{2pt}
Dataset 											&	
$\pdf(L_{\nu})$ 									&	
\multicolumn{2}{c}{$\theta_{\mathrm{calc, M}}$} 	&	
\multicolumn{2}{c}{$\theta_{\mathrm{calc, L}}$} 	&	
\multicolumn{2}{c}{$\theta_{\mathrm{calc, T/Y}}$} 	&	
$\prob(\theta_{\text{M}}   > \theta_{L})$			&	  
$\prob(\theta_{\text{T/Y}} > \theta_{M})$			&	  
$\prob(\theta_{\text{T/Y}} > \theta_{L})$			\\
			\midrule  
Literature						&	KM		&	0.25 & $_{-0.10}^{+0.13}$	&	0.13 & $_{-0.05}^{+0.07}$ &		0.29 &  $_{-0.11}^{+0.13}$	& 84\% & 58\% & 89\%		\\[2pt]
Literature						&	Uniform	&	0.17 & $_{-0.07}^{+0.09}$	&	0.10 & $_{-0.04}^{+0.05}$ &		0.23 &  $_{-0.09}^{+0.11}$	& 83\% & 67\% &	92\%		\\[8pt]
Literature, no K16$^{a}$		&	KM		&	0.25 & $_{-0.10}^{+0.13}$	&	0.10 & $_{-0.05}^{+0.07}$ &		0.17 &	$_{-0.09}^{+0.13}$  & 90\% & 34\% &	77\% 		\\[2pt]
Literature, no K16$^{a}$		&	Uniform	&	0.17 & $_{-0.07}^{+0.09}$  	&	0.07 & $_{-0.03}^{+0.05}$ & 	0.13 & 	$_{-0.07}^{+0.10}$ 	& 89\% & 41\% & 80\%		\\[8pt]
$\geq$L4$^{b}$	 				&	KM		&	0.25 & $_{-0.10}^{+0.13}$ 	&	0.10 & $_{-0.05}^{+0.08}$ &		0.29 &  $_{-0.11}^{+0.13}$ 	& 88\% & 58\% &	91\%		\\[2pt]
$\geq$L4$^{b}$					&	Uniform	&	0.17 & $_{-0.07}^{+0.09}$	&	0.07 & $_{-0.04}^{+0.06}$ & 	0.23 &  $_{-0.09}^{+0.11}$	& 86\% & 67\% &	93\%		\\[8pt]
$\geq$L4$^{b}$, no K16$^{a}$	&	KM		&	0.25 & $_{-0.10}^{+0.13}$	&	0.05 & $_{-0.04}^{+0.07}$ &		0.17 &	$_{-0.09}^{+0.13}$  & 94\% & 34\% &	86\%		\\[2pt]
$\geq$L4$^{b}$, no K16$^{a}$	&	Uniform	&	0.17 & $_{-0.07}^{+0.09}$ 	&	0.04 & $_{-0.03}^{+0.05}$ &		0.13 & 	$_{-0.07}^{+0.10}$	& 93\% & 41\% &	87\%		\\[2pt]
			\bottomrule
			\end{tabularx}	
			\begin{tablenotes}[]\footnotesize
				\item[] \textit{Note} --- For the single sample, single objects were randomly drawn from a samples of 82 ultracool M dwarfs, 74 L dwarfs, and 23 T/Y dwarfs to match the spectral type distribution of the individual components in the binary sample.  Listed values are calculated from the mean distribution of 1000 trials.  Reported uncertainties correspond to 68.3\% credible intervals.
				\item[$a$] Excluding observations from \citet{Kao2016ApJ...818...24K,Kao2018ApJS..237...25K} to reduce possible bias from their selection effects. 
				\item[$b$] Excluding L dwarfs with spectral type earlier than L4 to examine the transition from L to T spectral types. 
			\end{tablenotes}
		\end{ThreePartTable}  
	\end{table*}
\end{center}

\section{Discussion}\label{sec.Discussion}

\begin{figure*}
	\begin{centering}
		\includegraphics[width=\columnwidth]{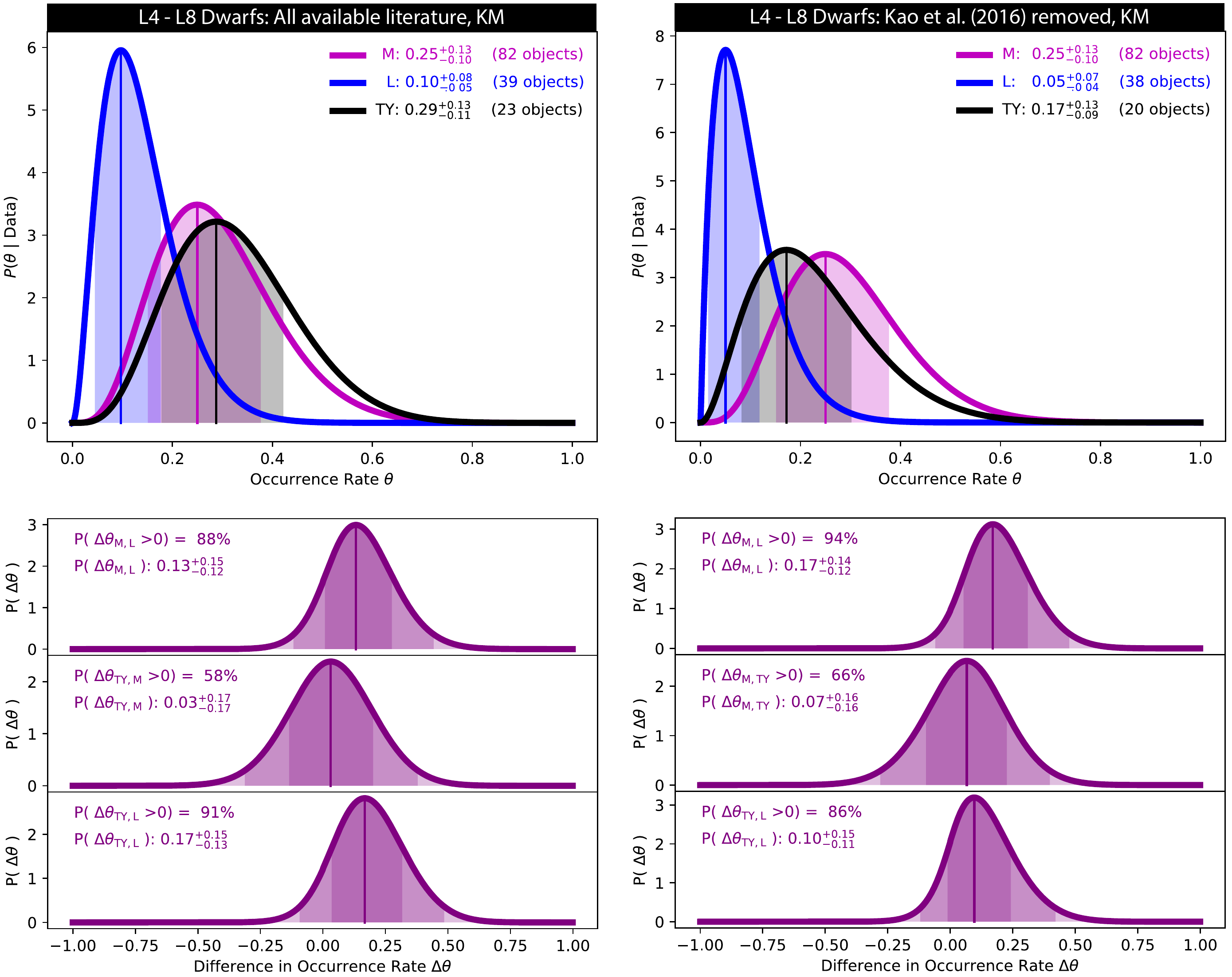}
		\caption{\label{fig:occurrenceRate_L4-L9}
The same as the top and bottom rows of Figure \ref{fig:occurrenceRate}, except L dwarf data include only L4 and later spectral type.  The possible suppression in the L dwarf occurrence rate compared to M and T/Y dwarfs is more pronounced.}
\end{centering}
\end{figure*}

\subsection{Comparison between quiescent and flaring radio occurrence rates and implications}
Applying our framework to the case presented in \S \ref{section:application} demonstrates the subtle but important science questions that can be identified with the careful analyses afforded by occurrence rate studies.

For example, we find that the $15^{+4}_{-4}$\% and ($20^{+6}_{-5}$)\% quiescent radio occurrence rates obtained using the uniform and Kaplan-Meier luminosity priors, respectively, for the full (not subdivided) sample significantly exceeds the  $\sim$4.6\% flaring radio occurrence rate calculated via Monte Carlo simulations  by \citet{Route2017ApJ...845...66R}, even when they account for time dependence. If the processes for quiescent radio emission and flaring auroral emission are linked \citep[e.g.][]{Pineda2017ApJ...846...75P, Kao2019MNRAS.487.1994K}, this discrepancy may suggest that the luminosity distribution chosen by \citet{Route2017ApJ...845...66R} is inconsistent with the processes occurring in these systems.  Indeed, \citet{Route2017ApJ...845...66R} choose a luminosity prior that reflects flare energy distributions from stellar flares \citep[e.g.][]{Crosby1993AdSpR..13i.179C, Guedel2003ApJ...582..423G, Loyd2018ApJ...867...71L,Paudel2018ApJ...858...55P}.  These flare energy distributions predict that lower energy flares occur more frequently than high energy flares.  However, for all completeness calculations, luminosity priors that up-weight low luminosity emission should result in higher rather than lower inferred occurrence rates.  

Thus, one possible explanation for the discrepancy between modeled quiescent and flaring ultracool dwarf radio occurrence rates is that not all quiescent radio emission traces flaring auroral activity in ultracool dwarfs.  If this is the case, how is the source of quiescent radio emission connected to or distinct from the reservoir of electrons that produce auroral emission?  

Another possibility is that the quiescent radio occurrence rate may plausibly exceed the auroral radio occurrence rate at high frequencies even if the two phenomena are physically connected.  This is because auroral electron cyclotron maser emits at near the fundamental electron gyrofrequency for ultracool dwarfs \citep{Melrose1984, Dulk1985, Treumann2006AARv..13..229T, Hallinan2008ApJ...684..644H}, whereas quiescent gyrosynchrotron processes can emit at high harmonics of this same frequency \citep{Dulk1985}.  Indeed, \citet{Williams2015ApJ...815...64W} detect 95 GHz emission from the auroral ultracool dwarf TVLM 513-46546 at frequencies that well exceed plausible electron cyclotron maser emission. We therefore surmise that a significant proportion of ultracool dwarfs may have weaker magnetic fields than the $\gtrsim$1.4 kG fields traced by this population's flaring auroral radio emission. A recent detection of coherent radio emission from at $\sim$T6.5 dwarf at 144 MHz attributed to the electron cyclotron maser instability by \citet{Vedantham2020ApJ...903L..33V} is consistent with this picture.

\subsection{A possible suppressed L dwarf quiescent radio occurrence rate}
\subsubsection{Ruling out framework and selection effects}
Additionally, though not statistically significant, the possibility of a divergence between the radio occurrence rates of L dwarfs versus T/Y dwarfs is noteworthy upon closer examination.  

First, we rule out occurrence rate over-predictions from faint T/Y dwarf radio luminosities as a possible explanation for a divergence between L and T/Y dwarf radio occurrence rates. Our simulations show that our framework can significantly over-predict occurrence rates in populations with low luminosities that are not well constrained by the available data (Table \ref{table:simulation_low-lum_uniform}).   Thus far, detected T/Y dwarf quiescent radio emission has been less luminous than for M and L dwarfs (Table \ref{table:quiescent}). This trend is tentative, but it may point to a T dwarf luminosity distribution that is on average less luminous than our assumed distribution. If this proves to be the case with future observations of T and Y dwarfs, then comparing the low luminosity and literature simulations demonstrates that our calculated T/Y dwarf radio occurrence rate actually gives a more conservative maximum-likelihood occurrence rate than if we had used a luminosity distribution over lower luminosities.  Furthermore, the credible intervals perform similarly in recovering the true occurrence rate.

Second, we rule out selection effects as a possible explanation.  In particular, two of the four radio-bright T dwarfs in our sample were detected in a study that biased their sample using possible multiwavelength tracers of auroral magnetic activity  \citep{Kao2016ApJ...818...24K}. This selection criterion resulted in a significantly higher detection rate ($\sim$80\%) compared to historical detection rates ($\sim$5--10\%) \citep{Antonova2013AA...549A.131A, RouteWolszczan2016ApJ...830...85R}.  

However, our literature sample includes numerous studies with various selection criteria, including but not limited to distance, youth, planet companions, photometric variability, and previously detected markers of possible magnetic activity. As the number of studies and survey targets increase for a given sample, the individual contributions of the selection criteria to the calculated radio occurrence rate decrease.  Thus, for this study, we elected to simply use all available data in the literature.  Given the relatively small size of the T/Y dwarf sample, we expect that selection effects will have the strongest impact for this sample. Future radio studies of T dwarfs that average down the impact of selection criteria on the radio occurrence rate will test this hypothesis.

As a check, we exclude the objects from \citet{Kao2016ApJ...818...24K}, except for the previously detected T6.5 dwarf 2MASS J10475385+2124234 \citep{RouteWolszczan2012ApJ...747L..22R}. With the removal of one radio-quiet and two radio-bright T dwarfs, the maximum-likelihood occurrence rate for T/Y dwarfs decreases to $17^{+13}_{-9}$--$13^{+10}_{-7}$\%, converging with that of M dwarfs. Concurrently, the L dwarf  maximum-likelihood occurrence rate also decreases to $10^{+7}_{-5}$--$7^{+5}_{-3}$\% with the removal of one radio-bright L dwarf, depending on the luminosity prior.  The tentative divergence between the L and T/Y dwarf sample persists but becomes less significant, with the probability  T/Y dwarfs have a higher radio occurrence rate  than L dwarfs   $\prob(\theta_{\text{T/Y}} > \theta_{L})$ decreasing from 89--92\% to 77--80\%.  Simultaneously, the significance of a possible suppressed L dwarf radio occurrence rate relative to the M dwarf sample increases from 83--84\% to 89--90\% probability.  We conclude that selection effects from the \citet{Kao2016ApJ...818...24K} are insufficient to explain the divergence between the L and T/Y dwarf samples.

\subsubsection{Possible physical effects and future work}
Having ruled out observational and calculation effects, we now consider physical interpretations.   In the following discussion, we focus on L and T dwarfs, since brown dwarfs can exhibit behaviors that are distinct from stars.  We exclude M dwarfs from our discussion, since a significant fraction of them may be stars rather than brown dwarfs, depending on their ages.  

The low occurrence rate of radio emission in L dwarfs is at odds with published H$\alpha$ magnetic activity trends in ultracool dwarfs.  \citet{Pineda2016ApJ...826...73P} show that 67/195 ($34^{+3.5}_{-3.2}\%$) of L dwarfs possess  detectable H$\alpha$ emission, whereas the detection fraction for T dwarfs is 3/42 ($8.8^{+7.4}_{-2.8}\%$).    Although \citet{Pineda2016ApJ...826...73P} do not separately treat binaries as we do here, the significant excess in L dwarf H$\alpha$ detection fractions compared to their radio occurrence rates suggests that H$\alpha$ emission in some fraction of L dwarfs may be independent of the sources for their radio emission.  H$\alpha$ emission correlates with both radio aurorae in late L and T dwarfs \citep{Kao2016ApJ...818...24K} as well as quiescent radio luminosities in auroral ultracool dwarfs \citep{Pineda2017ApJ...846...75P, Richey-Yowell2020}.  Comparing the L and T dwarf radio occurrence rates to their H$\alpha$ detection fractions implies that the transition from flare-like to aurorae-like magnetic activity continues through L spectral types, which reports of white light flares on L dwarfs support \citep{Gizis2013ApJ...779..172G, Paudel2018ApJ...858...55P, Jackman2019MNRAS.485L.136J, Paudel2020arXiv200410579P}. 

Evidence for a tentatively suppressed L dwarf occurrence rate becomes more pronounced when we exclude early L dwarfs.  Early L dwarfs are expected to be more chromospherically active than late L dwarfs with more neutral atmospheres, and excluding them significantly lowers the H$\alpha$ detection fraction to 7/75 ($9.3^{+4.5}_{-2.4}\%$) for L4--L8 dwarfs \citep{Pineda2016ApJ...826...73P}.  We find that the L dwarf radio occurrence rate also decreases, albeit less dramatically, when we exclude L3 and earlier objects (Figure \ref{fig:occurrenceRate_L4-L9}), to $4^{+5}_{-3} - 10^{+8}_{-5}\%$, depending on if we exclude the \citet{Kao2016ApJ...818...24K} objects and the luminosity prior that we choose.   At such late spectral types, these objects are most likely brown dwarfs like their T dwarf brethren, and one might naively expect brown dwarfs to exhibit similar radio magnetic activity.  Surprisingly,  the  tentative L dwarf suppression in fact becomes more significant, with $\prob(\theta_{\text{T/Y}} > \theta_{L}) = 86 - 93\%$ and $\prob(\theta_{\text{M}} > \theta_{L}) = 86 - 94\%$.  

For stars, youth correlates with faster rotation, resulting in stronger magnetic activity at X-ray wavelengths \citep{PreibischFeigelson2005ApJS..160..390P, Kiraga2007AcA....57..149K, Vidotto2014MNRAS.441.2361V, Booth2017MNRAS.471.1012B} that in turn correlates tightly with radio activity \citep{Guedel1993ApJ...405L..63G}. While ultracool dwarfs depart strongly from the G\"{u}del-Benz relationship that links stellar coronal heating  (X-ray) to magnetic fields (radio) \citep{Williams2014ApJ...785....9W}, could rotation nevertheless also play an important role in ultracool dwarf magnetic activity?

The role of rotation could provide one possible resolution to existing tension between predictions from dynamo models and ultracool dwarf magnetic fields measurements. For rapidly rotating fully convective objects like ultracool dwarfs, simulations find that the internal thermal energy budget available for convection is the dominating factor for their magnetic energy budget  \citep{Christensen2009Natur.457..167C, Yadav2017ApJ...849L..12Y}. However, convected thermal fluxes cannot predict magnetic fields strengths measured from observations of brown dwarf radio aurorae \citep{Kao2018ApJS..237...25K}, suggesting that other factors --- for example, rotation --- may also need to be considered. Indeed, \citet{Kao2023} recently demonstrate using resolved imaging at 8.4~GHz that quiescent emission from ultracool dwarfs can trace radiation belt analogs. The rotational energy reservoirs of solar system planets are important contributors to their radiation belt acceleration processes \citep{CowleyBunce2001, Krupp2007jupi.book..617K, Kollmann2018JGRA..123.9110K} and models suggest that the same may be true for radiation belts around ultracool dwarfs \citep{Leto2021MNRAS.507.1979L, Owocki2022}.

Intriguingly, the total radio detection fraction for both bursting and/or quiescent ultracool dwarf radio emission increases with increasing $v \sin i$, approximately quadrupling between 30 and 50 km s$^{-1}$ from $\sim$10\% to 40\% \citep{Pineda2017ApJ...846...75P}. Though the impact of rotation on the radio occurrence rates of ultracool dwarfs remains unstudied, theory quantifying the energy available to produce aurorae find that rotation rates are an important factor \citep{Turnpenney2017MNRAS.470.4274T, Saur2021AA...655A..75S}, while a correlation between ultracool dwarf H$\alpha$ luminosities attributed to aurorae \citep{Hallinan2015Natur.523..568H} and their quiescent radio emission \citep{Pineda2017ApJ...846...75P, Richey-Yowell2020} suggests that electron acceleration processes tied to auroral radio emissions -- such as rotation -- may also be important for producing their quiescent radio emission. 

If confirmed, a suppression in the L dwarf quiescent radio occurrence rate would be consistent with the age-rotation evolution of brown dwarfs. Unlike stars, brown dwarfs spin up rather than down with age as they contract throughout their lifetimes. For example, the rotation periods of 16 young brown dwarfs in the nearby Upper Scorpius association (age 5--10 Myr) range from 0.2 to 5 days, with a median of 1.1 days \citep{Scholz2015ApJ...809L..29S}. In contrast, field L and T dwarfs typically rotate on timescales of several hours \citep{Radigan2014ApJ...793...75R}. Since brown dwarfs cool and age along the M/L/T/Y spectral sequence, the L dwarfs in our sample may be younger on average with lower $v \sin i$ values than T/Y dwarfs of similar masses. In this scenario, if rotation is a dominant contributor to ultracool dwarf magnetism, we might expect to observe a suppressed L dwarf radio occurrence rate. While no studies currently exist for examining how youth correlates with radio activity, we apply the occurrence rate framework that we present here for a dedicated radio study of young brown dwarfs in a forthcoming paper.

Our occurrence rate framework  provides a path forward for more rigorously examining  the impact of rotation on ultracool dwarf non-auroral radio activity.  How does their radio occurrence rate rate change as function of rotation period?  Do younger brown dwarfs actually have a suppressed radio occurrence rate relative to their older counterparts?  While current sample sizes may preclude individual studies from obtaining statistically significant results on their own, a suite of results that support a physically consistent picture may illuminate the subtle physics occurring in ultracool dwarf magnetospheres.

\section{Conclusions}\label{sec.Conclusion}

In this work, we develop a generalized analytical Bayesian framework for calculating the occurrence rate of steady emission in astrophysical populations.  Our framework can take input data from telescopes with varying noise properties, which enables the full use of data available in the literature.  This approach complements existing approaches in other fields that rely on characterizing the recovery rates of dedicated data reduction pipelines for observing campaigns that use a single telescope.  Our framework treats completeness by analytically accounting for object distances, distance uncertainties, observational sensitivities, and intrinsic luminosity. As a proof of concept, we apply this framework to non-flaring quiescent radio emission in ultracool dwarfs.  We caution the reader that applying our framework to multi-object systems is a nuanced matter that requires some additional theoretical development and careful consideration of assumptions. This application is outside of the scope of this work, and we refer the reader to \citep{Kao2020_binaries} for the proper treatment.

We present the first detailed statistical study of non-flaring quiescent radio emission in ultracool dwarfs as a means of studying their magnetospheric physics.  Prior to this work, raw detection rates were the standard for ultracool dwarf radio statistics not specifically focused on time-variable radio aurorae, yet these detection rate studies did not account for completeness or binarity.  Furthermore, while quiescent radio emission cannot directly measure magnetic field strengths in the manner enabled by radio aurorae, we focus on the former to minimize the number of physical assumptions and observing effects that must be accounted for.   

 We compile all published radio observations of ultracool dwarfs and simulate synthetic data samples drawn from literature distributions of relevant characteristics to assess the performance of our framework. We show that for literature distributions, our framework recovers the simulated radio occurrence rate  within 1--5\% on average for sample sizes as small as 10 objects if the underlying radio luminosity models that we assume accurately reflect the true luminosity distribution of ultracool dwarf radio emission.  This is a factor of $\sim$13 improvement in accuracy compared to detection rates, which we show underestimate the true quiescent radio emission rate by $\sim$51\% or 66\%, for uniform and Kaplan-Meier luminosity distribution priors, respectively.  Furthermore, we show that our credible intervals serve as good proxies for confidence intervals on the true occurrence rate for all sample sizes.   In contrast, the performance of detection rate confidence intervals suffers with increasing sample size.  Thus, our simulations highlight that developing robust statistical frameworks that properly account for systematics is essential for preparing for dramatically expanded datasets as radio astronomy transitions to an era of big data.
 
Comparing the quiescent radio occurrence rate of ultracool dwarfs to the flaring radio occurrence rate calculated by  \citet{Route2017ApJ...845...66R} suggest that processes in addition to auroral processes may contribute to observed quiescent emission and/or that a significant portion of the radio-emitting population may have $\leq$1.4 kG magnetic fields.  Additionally, comparing quiescent radio occurrence rates by spectral type tentatively suggests that age or rotation effects may contribute to a possible suppression in the L dwarf quiescent radio occurrence rate.  Applying the occurrence rate framework that we present to studies examining how these properties correlate with quiescent radio emission will be one key strategy for assessing ultracool dwarf magnetic activity.

Finally, we emphasize that ultracool dwarf magnetic activity is just one possible use case for the generalized framework that we present.  This same framework may be applied to other examples of steady or quasi-steady emission or absorption.

\section*{Acknowledgements}
MK specially thanks C. Voloshin and J. S. Pineda for valuable discussions that shaped this work.  Additionally, the authors thank M. Bryan, M. Knapp, R. P. Loyd, and T. Richey-Yowell for insightful questions and providing feedback that helped clarify the presentation of this work.  Support was provided by NASA through the NASA Hubble Fellowship grant HST-HF2-51411.001-A awarded by the Space Telescope Science Institute, which is operated by the Association of Universities for Research in Astronomy, Inc., for NASA, under contract NAS5-26555; and by the Heising-Simons Foundation through the 51 Pegasi b Fellowship grant 2021-2943. This work is based on observations made with the NSF's Karl G. Jansky Very Large Array (VLA) and made use of the SIMBAD and VizieR databases, operated at CDS, Strasbourg, France; and the European Space Agency (ESA) mission \textit{Gaia} (\url{https://www. cosmos.esa.int/gaia}), processed by the Gaia Data Processing and Analysis Consortium (DPAC, \url{https://www. cosmos.esa.int/web/gaia/dpac/consortium}). 

Software: CASA \, \citep{CASA}, Astropy \, \citep{astropy:2018},\, Matplotlib \citep{matplotlib},\, Numpy \,\citep{numpy2}, \, Scipy \, \citep{scipy}, \, Scikit-learn \citep{scikit-learn}.

\vspace*{-0.2cm}
\section*{Data Availability}
This work is a meta-analysis of previously published radio observations, which we have compiled in Table \ref{table:literatureBD}.
 


\vspace*{-0.2cm}
\bibliographystyle{mnras}
\bibliography{occurrenceRate} 



\appendix
\newtheorem{theorem}{Theorem}[section]
\newtheorem{definition}{Definition}[theorem]
\newtheorem{lemma}[theorem]{Lemma}
\newtheorem{statement}[theorem]{Statement}
\newtheorem{remark}{Remark}
\newcommand{\E}{\mathbb{E}}
\newcommand{\XIN}{X_{\mathrm{in}}}
\newcommand{\XOUT}{X_{\mathrm{out}}}
\newcommand{\XTOT}{X_{\mathrm{tot}}}

\vspace*{-0.5cm}
\section{Proof for p($d_i$) } \label{appendix:p(d)}

\subsection{Concise proof }
\begin{theorem}
Let $\vtot \subseteq \mathbb{R}^3$ be the total volume of three-dimensional space. Let $\vin,\vout$ such that $\vin \cap \vout = \varnothing$ and $\vin \cup \vout = \vtot$ be two mutually exclusive subsets whose union is the the total volume. Assume the distribution of brown dwarfs over $\vtot$ is Poisson($\lambda$) for some fixed parameter $\lambda \in \mathbb{R}$. Let $d$ be the event such that a brown dwarf is within $\vin$. Then 

\begin{equation}
\prob(d) = \frac{\vin}{\vtot}
\end{equation}
\end{theorem}

\begin{remark}
While  $\vin,\vout, \vtot$ are technically sets of points in space, we will use a slight abuse of notation to mean volume ($|\vin|,|\vout|, |\vtot|$) rather than the set of points.
\end{remark}

\begin{proof}
Let $X_{\mathrm{in}},X_{\mathrm{out}}$ be distributed like $\mathrm{Poisson}(\lambda)$ be the random variables describing the number of brown dwarfs occurring in the volumes $\vin$ and $\vout$ respectively.  Since the number of brown dwarfs in $\vtot$ is Poisson($\lambda$) and the sets $\vin,\vout$ are mutually exclusive, $X_{in},X_{out}$ are two independent and identically distributed random Poisson variables with parameter $\lambda \in \mathbb{R}$, describing the number of brown dwarfs occurring in the volumes $\vin$ and $\vout$ respectively. Furthermore, because $\vin \cup \vout = \vtot$, then $X_{\mathrm{in}}+X_{\mathrm{out}} = X_{\mathrm{tot}}$, the total number of brown dwarfs occurring in $\vtot$.
\begin{equation}
\prob(d)  =\frac{\E[X_{\mathrm{in}}]}{\E[X_{\mathrm{tot}}]}= \frac{\E[X_{\mathrm{in}}]}{\E[X_{\mathrm{in}} + X_{\mathrm{out}}]} = \frac{\lambda \vin}{ \lambda \vtot } =  \frac{\vin}{\vtot} 
\end{equation}
\end{proof}

\subsection{Algebraic proof }

\begin{lemma}
	 $\sum_{k=0}^{N} k {N \choose k} x^k y^{N-k}  = Nx(x+y)^{N-1} $
\end{lemma}

\begin{proof}
We begin with the statement of the Binomial Theorem
	\begin{equation}
		(x+y)^{N} = \sum_{k=0}^{N} {N \choose k} x^k y^{N-k} \,.
	\end{equation}
By differentiating with respect to $x$ and multiplying x yields
	\begin{equation}
		Nx(x+y)^{N-1} = \sum_{k=0}^{N} {N \choose k} k x^{k} y^{N-k} 
	\end{equation}
\end{proof}

\begin{proof}
The probability that a brown dwarf is within $\vin$ given $i$ brown dwarfs in $\vin$ and $N$ brown dwarfs in $\vtot$ is 
\begin{equation} \label{eqn:p_iN}
\prob(d \mid X_{\mathrm{in}} = i, X_{\mathrm{tot}} = N) = i/N \;.
\end{equation}
Since $\XIN,\XOUT \sim \mathrm{Poisson}(\lambda)$ then 
\begin{align}
\prob(X_{\mathrm{in}} = i) = & \frac{(\lambda \vin)^i }{i!} e^{-\lambda \vin}  \label{eqn:poisson_in}, \quad  \\
\prob(X_{\mathrm{out}}= N - i) = & \frac{(\lambda \vout)^{N-i} }{(N-i)!} e^{-\lambda \vout} 
\end{align}
by definition. Using this fact, then we can derive:
\begin{align} 
	\prob(X_{\mathrm{tot}}	= N \mid X_{\mathrm{in}} = i)  & = \prob(X_{\mathrm{in}} + X_{\mathrm{out}} = N \mid X_{\mathrm{in}} = i) \\
							&	= \prob(X_{\mathrm{out}}= N - X_{\mathrm{in}} \mid X_{\mathrm{in}} = i)  \\ 
							&	= \prob(X_{\mathrm{out}}= N - i) \\
							&	= \frac{(\lambda \vout)^{N-i} }{(N-i)!} e^{-\lambda \vout}  \,. \label{pXtot}
\end{align}

By the law of total probability, sum over all possible $X_{\mathrm{in}}, X_{\mathrm{tot}}$:
\begin{align}
	\prob(d)  & = \sum_{N=0}^{\infty}  \sum_{i=0}^{N} \prob(d \cap  X_{\mathrm{tot}} = N \cap X_{\mathrm{in}} = i) \\
			   & = \sum_{N=0}^{\infty}\sum_{i=0}^{N}  \prob(d \mid  X_{\mathrm{tot}} = N, X_{\mathrm{in}} = i) \times \prob(X_{\mathrm{in}} = i) \times \prob( X_{\mathrm{tot}}=N \mid X_{\mathrm{in}} = i) \;, \quad \text{by definition}  \\
	 			& = \sum_{N=0}^{\infty}  \sum_{i=0}^{N} \underbrace{\frac{i}{N}}_{\ref{eqn:p_iN}}  
				\underbrace{\frac{(\lambda \vin)^i }{i!} e^{-\lambda \vin}}_{\ref{eqn:poisson_in}}  
				\underbrace{\frac{(\lambda \vout)^{N-i} }{(N-i)!} e^{-\lambda \vout}}_{\ref{pXtot}} \; .
\end{align}

Multiplying by $N!/N!$ on both sides gives 
\begin{align}
	\prob(d) 	& = \sum_{N=0}^{\infty}  \sum_{i=0}^{N} \frac{i}{N (N!)}   \frac{N!}{i! (N-i)!} \lambda^{N} \vin^i \vout^{N-i} e^{-\lambda \vtot} \,.
\end{align}

By the definition of the binomial coefficient, 
\begin{align}
	\prob(d) 	& = \sum_{N=0}^{\infty}  \sum_{i=0}^{N} \frac{i}{N (N!)}   {N \choose i} \lambda^{N} \, \vin^i \, \vout^{N-i} \, e^{-\lambda \vtot} \,.
\end{align}

By Lemma 1, 
\begin{align}
	\prob(d) 	& =  e^{-\lambda \vtot}  \sum_{N=0}^{\infty}  \frac{\lambda^N}{N!}	 \vin \vtot^{N-1} \frac{\vtot}{\vtot} 
				\,\,\, = \,\,\,  \frac{\vin}{\vtot}  e^{-\lambda \vtot}  \sum_{N=0}^{\infty}  \frac{(\lambda \vtot)^N}{N!}	 \,\,.  \\
\end{align}

By the series expansion definition of the exponential, 
\begin{align}
	\prob(d) 	& =  \frac{\vin}{\vtot}  e^{-\lambda \vtot}  e^{\lambda \vtot} 	   \\
				& = \frac{\vin}{\vtot} \, ,
\end{align}
as desired.
\end{proof}

\setlength{\tabcolsep}{0.15in}
\begin{center}
	\begin{table*}\centering 
		\begin{ThreePartTable}
			\caption{\label{table:literatureBD} Radio Observations of Ultracool Dwarfs from the Literature included in Occurrence Rate Calculation }
			\begin{tabularx}{\textwidth}{llccllcll}
			\toprule \vspace{2pt}
Object Name			&
SpT					&	
ref					&	
$\pi$				&	
$d$$^{a}$			&	
ref					&	
$F_{\nu}$			&	
ref					&	
Note	\\[-3pt]
					&	
					&	
					&	
(mas)				&	
(pc)				&	
					&	
($\mu$Jy)			&	
					&	
					\\%
\midrule  
\midrule  	
M dwarfs		& &	 		&		&	&	&	&	& \\
\midrule
2MASS J00192626+4614078	&	M8	&	74	&	26.0858	$\pm$	0.1696	&	38.3	$\pm$	0.2	&	30	&	$<$33			&	10	&		\\
2MASS J00194579+5213179	&	M9	&	74	&	50.1097	$\pm$	0.18	&	20.0	$\pm$	0.1	&	30	&	$<$42			&	10	&		\\
2MASS J00274197+0503417	&	M8	&	28	&	13.24	$\pm$	1.56	&	75.5	$\pm$	8.9	&	22	&	$<$75			&	8	&		\\
2MASS J01090150-5100494 	&	M8.5	&	51	&	62.8527	$\pm$	0.113	&	15.9	$\pm$	0.0	&	30	&	$<$110			&	17	&		\\
2MASS J01092170+2949255	&	M9.5	&	72	&	62.7825	$\pm$	0.2442	&	15.9	$\pm$	0.1	&	30	&	$<$54			&	64	&		\\
2MASS J01095117-0343264 	&	M9	&	69	&	94.3979	$\pm$	0.318	&	10.6	$\pm$	0.0	&	30	&	$<$22			&	82	&		\\
2MASS J01400263+2701505	&	M8.5	&	24	&	52.6488	$\pm$	0.1771	&	19.0	$\pm$	0.1	&	30	&	$<$20			&	64	&		\\
2MASS J01483864-3024396	&	M7.5	&	74	&	41.0986	$\pm$	0.4302	&	24.3	$\pm$	0.3	&	30	&	$<$45			&	10	&		\\
2MASS J01490895+2956131	&	M9.7	&	5	&	41.6592	$\pm$	0.4501	&	24.0	$\pm$	0.3	&	30	&	$<$140			&	64	&		\\
2MASS J02150802-3040011	&	M8	&	74	&	71.111	$\pm$	0.1502	&	14.1	$\pm$	0.0	&	30	&	$<$75			&	4	&		\\
2MASS J02484100-1651216	&	M8	&	70	&	44.5509	$\pm$	0.1472	&	22.4	$\pm$	0.1	&	30	&	$<$81			&	58	&		\\
2MASS J03140344+1603056	&	M9.4	&	5	&	73.4296	$\pm$	0.2757	&	13.6	$\pm$	0.1	&	30	&	$<$108			&	58	&		\\
2MASS J03205965+1854233	&	M8	&	41	&	68.2785	$\pm$	0.1466	&	14.6	$\pm$	0.0	&	30	&	$<$81			&	58	&		\\
2MASS J03313025-3042383	&	M7.5	&	20	&	79.9181	$\pm$	0.1136	&	12.5	$\pm$	0.0	&	30	&	$<$72			&	10	&		\\
2MASS J03341218-4953322 	&	M9	&	26	&	112.6361	$\pm$	0.0923	&	8.9	$\pm$	0.0	&	30	&	$<$29.4			&	52	&		\\
2MASS J03350208+2342356	&	M8.5	&	78	&	19.5277	$\pm$	0.1543	&	51.2	$\pm$	0.4	&	30	&	$<$45			&	3	&		\\
2MASS J03505737+1818069	&	M9.0	&	24	&	27.3355	$\pm$	0.1528	&	36.6	$\pm$	0.2	&	30	&	$<$105			&	10	&		\\
2MASS J03510004-0052452	&	M8	&	24	&	68.0325	$\pm$	0.1851	&	14.7	$\pm$	0.0	&	30	&	$<$123			&	58	&		\\
2MASS J04173745-0800007	&	M7.5	&	24	&	55.2859	$\pm$	0.1612	&	18.1	$\pm$	0.1	&	30	&	$<$36			&	10	&		\\
2MASS J04341527+2250309	&	M7	&	61	&	5.815	$\pm$	0.5524	&	172.0	$\pm$	16.3	&	30	&	$<$69			&	10	&		\\
2MASS J04351455-1414468	&	M7	&	28	&	7.3148	$\pm$	0.1492	&	136.7	$\pm$	2.8	&	30	&	$<$42			&	4	&		\\
2MASS J04351612-1606574	&	M8	&	57	&	94.4906	$\pm$	0.2523	&	10.6	$\pm$	0.0	&	30	&	$<$48			&	10	&		\\
2MASS J04361038+2259560	&	M8	&	56	&	5.8396	$\pm$	0.4808	&	171.2	$\pm$	14.1	&	30	&	$<$45			&	10	&		\\
2MASS J04363893+2258119	&	M7.75	&	61	&	5.4375	$\pm$	0.3327	&	183.9	$\pm$	11.3	&	30	&	$<$57			&	10	&		\\
2MASS J04402325-0530082	&	M7.5	&	57	&	102.4638	$\pm$	0.1063	&	9.8	$\pm$	0.0	&	30	&	$<$39			&	10	&		\\
2MASS J04433761+0002051	&	M9	&	28	&	47.4115	$\pm$	0.1857	&	21.1	$\pm$	0.1	&	30	&	$<$54			&	58	&		\\
2MASS J05173766-3349027	&	M8.0	&	24	&	59.2738	$\pm$	0.1043	&	16.9	$\pm$	0.0	&	30	&	$<$54			&	10	&		\\
2MASS J05392474+4038437	&	M8.0	&	48	&	87.5414	$\pm$	0.1905	&	11.4	$\pm$	0.0	&	30	&	$<$63			&	4	&		\\
2MASS J05441150-2433018	&	M7.6	&	5	&	48.4871	$\pm$	0.1281	&	20.6	$\pm$	0.1	&	30	&	$<$63			&	58	&		\\
2MASS J07410681+1738459 	&	M7.0	&	24	&	53.5151	$\pm$	0.2233	&	18.7	$\pm$	0.1	&	30	&	$<$75			&	58	&		\\
2MASS J08105865+1420390	&	M8	&	80	&	41.9878	$\pm$	0.1692	&	23.8	$\pm$	0.1	&	30	&	$<$39			&	64	&		\\
2MASS J08185804+2333522	&	M7	&	74	&	44.3886	$\pm$	0.1868	&	22.5	$\pm$	0.1	&	30	&	$<$78			&	58	&		\\
2MASS J08533619-0329321	&	M9	&	71	&	115.3036	$\pm$	0.1132	&	8.7	$\pm$	0.0	&	30	&	$<$42.9			&	52	&		\\
2MASS J10063197-1653266	&	M7.5	&	74	&	49.8335	$\pm$	0.1435	&	20.1	$\pm$	0.1	&	30	&	$<$87			&	58	&		\\
2MASS J10163470+2751497	&	M8	&	41	&	49.4994	$\pm$	0.139	&	20.2	$\pm$	0.1	&	30	&	$<$47			&	8	&		\\
2MASS J10240997+1815533	&	M8	&	74	&	28.958	$\pm$	0.4403	&	34.5	$\pm$	0.5	&	30	&	$<$87			&	58	&		\\
2MASS J10481463-3956062	&	M8.5	&	35	&	247.2157	$\pm$	0.1236	&	4.0	$\pm$	0.0	&	30	&	--			&	67	&		\\
2MASS J11020983-3430355	&	M9	&	28	&	16.731	$\pm$	0.2121	&	59.8	$\pm$	0.8	&	30	&	$<$42			&	63	&		\\
2MASS J11240487+3808054	&	M8.5	&	20	&	54.1348	$\pm$	0.2015	&	18.5	$\pm$	0.1	&	30	&	$<$66			&	58	&		\\
2MASS J11395113-3159214	&	M9	&	2	&	20.1343	$\pm$	0.246	&	49.7	$\pm$	0.6	&	30	&	$<$0.10			&	17	&		\\
2MASS J11414406-2232156	&	M7.5	&	20	&	52.6168	$\pm$	0.2017	&	19.0	$\pm$	0.1	&	30	&	$<$108			&	58	&		\\
2MASS J11554286-2224586	&	M7.5	&	25	&	91.5952	$\pm$	0.1544	&	10.9	$\pm$	0.0	&	30	&	$<$19			&	82	&		\\
2MASS J12245222-1238352	&	M8	&	32	&	58.1118	$\pm$	0.1929	&	17.2	$\pm$	0.1	&	30	&	$<$102			&	58	&		\\
2MASS J12505265-2121136	&	M7.5	&	51	&	56.0805	$\pm$	0.2522	&	17.8	$\pm$	0.1	&	30	&	$<$66			&	4	&		\\
2MASS J12531240+4034038	&	M7	&	71	&	47.1353	$\pm$	0.0953	&	21.2	$\pm$	0.0	&	30	&	$<$78			&	58	&		\\
2MASS J13092185-2330350	&	M7	&	32	&	66.6683	$\pm$	0.1913	&	15.0	$\pm$	0.0	&	30	&	$<$93			&	58	&		\\
2MASS J13322442-0441126	&	M7.5	&	74	&	--			&	19.0	$\pm$	2.0	&	26	&	$<$60			&	58	&		\\
2MASS J13540876+0846083	&	M7.5	&	60	&	37.9724	$\pm$	0.4169	&	26.3	$\pm$	0.3	&	30	&	$<$105			&	58	&		\\
2MASS J13564148+4342587	&	M8	&	81	&	50.0407	$\pm$	0.6125	&	20.0	$\pm$	0.2	&	30	&	$<$99			&	58	&		\\
2MASS J14032232+3007547	&	M9	&	80	&	50.484	$\pm$	0.2048	&	19.8	$\pm$	0.1	&	30	&	$<$60			&	58	&		\\
2MASS J14112131-2119503	&	M8	&	2	&	22.1552	$\pm$	0.3232	&	45.1	$\pm$	0.7	&	30	&	$<$93			&	58	&		\\
2MASS J14213145+1827407	&	M8.9	&	5	&	52.6729	$\pm$	0.2626	&	19.0	$\pm$	0.1	&	30	&	$<$42			&	64	&		\\
2MASS J14280419+1356137	&	M7	&	62	&	75.4469	$\pm$	0.1329	&	13.3	$\pm$	0.0	&	30	&	$<$90			&	58	&		\\
2MASS J14284323+3310391	&	M9	&	71	&	90.9962	$\pm$	0.1271	&	11.0	$\pm$	0.0	&	30	&	$<$63			&	4	&		\\
\bottomrule
\end{tabularx}
\begin{tablenotes}[]\footnotesize
\item[]\textit{Note} --- We report upper limits for $F_{\nu}$ or flux densities as presented in the radio observation references and define detected = 1 if $F_{\nu} \geq 4.0$ to remain consistent with our occurrence rate calculation. This table does not include 2MASS J05233822-1403022, which \citet{Antonova2007AA...472..257A} report can vary by a factor of $\sim$5 from $\leq$45 to 230$\pm$17 $\mu$Jy with no evidence of short-duration flares during 2-hr observing blocks. Although the low circular polarization of this object rules out coherent aurorae, non-auroral  flares at both radio and optical frequencies can persist for at least several hours \citep[e.g.][]{Villadsen2019ApJ...871..214V, Paudel2018ApJ...861...76P}. We exclude this outlier object from on the basis of the uncertain nature of its radio emission. 
\item[$a$] Distances calculated from parallaxes are provided for the reader's convenience but are truncated in precision to the first decimal place in the interest of space. As such, uncertainties reported as 0.0 indicate that actual distance uncertainties are less than the reported significant figures but not truly zero.
	\end{tablenotes}
	\end{ThreePartTable}  
	\end{table*}
\end{center}

\setlength{\tabcolsep}{0.15in}
\begin{center}
	\begin{table*}\centering 
		\begin{ThreePartTable}
			\contcaption{Radio Observations of Ultracool Dwarfs from the Literature included in Occurrence Rate Calculation}
			\begin{tabularx}{\textwidth}{llccllcll}
			\toprule \vspace{2pt}
Object Name			&
SpT					&	
ref					&	
$\pi$				&	
$d$					&	
ref					&	
$F_{\nu}$			&	
ref					&	
Note	\\[-3pt]
					&	
					&	
					&	
(mas)				&	
(pc)				&	
					&	
($\mu$Jy)			&	
					&	
					\\%
\midrule  
\midrule  
2MASS J14402293+1339230	&	M7	&	80	&	38.9597	$\pm$	0.1403	&	25.7	$\pm$	0.1	&	30	&	$<$75			&	58	&		\\
2MASS J14441717+3002145	&	M8.5	&	60	&	67.2904	$\pm$	0.1144	&	14.9	$\pm$	0.0	&	30	&	$<$81			&	58	&		\\
2MASS J14563831-2809473	&	M7.0	&	23	&	141.6865	$\pm$	0.1063	&	7.1	$\pm$	0.0	&	30	&	$<$120			&	17	&		\\
2MASS J15072779-2000431	&	M7.5	&	74	&	41.7961	$\pm$	0.2207	&	23.9	$\pm$	0.1	&	30	&	$<$96			&	58	&		\\
2MASS J15164073+3910486	&	M7	&	62	&	56.7689	$\pm$	0.0719	&	17.6	$\pm$	0.0	&	30	&	$<$81			&	58	&		\\
2MASS J15210103+5053230 	&	M7.5	&	74	&	61.9989	$\pm$	0.1	&	16.1	$\pm$	0.0	&	30	&	$<$16			&	82	&		\\
2MASS J15345704-1418486	&	M8.6	&	5	&	91.6785	$\pm$	0.1735	&	10.9	$\pm$	0.0	&	30	&	$<$87			&	58	&		\\	
2MASS J15460540+3749458	&	M7.5	&	24	&	28.3105	$\pm$	0.1394	&	35.3	$\pm$	0.2	&	30	&	$<$84			&	58	&		\\
2MASS J16272794+8105075 	&	M9	&	26	&	49.2537	$\pm$	0.167	&	20.3	$\pm$	0.1	&	30	&	$<$60			&	64	&		\\
2MASS J16553529-0823401	&	M7	&	35	&	153.8139	$\pm$	0.1148	&	6.5	$\pm$	0.0	&	30	&	$<$22.5			&	43	&		\\
2MASS J17071830+6439331	&	M8.5	&	60	&	55.363	$\pm$	0.1155	&	18.1	$\pm$	0.0	&	30	&	$<$60			&	64	&		\\
2MASS J17571539+7042011	&	M7.5	&	24	&	52.5837	$\pm$	0.2462	&	19.0	$\pm$	0.1	&	30	&	$<$81			&	4	&		\\
2MASS J18261131+3014201	&	M8.5	&	47	&	90.0024	$\pm$	0.1086	&	11.1	$\pm$	0.0	&	30	&	$<$87			&	4	&		\\
2MASS J18353790+3259545	&	M8.5	&	24	&	175.8244	$\pm$	0.0905	&	5.7	$\pm$	0.0	&	30	&	464	$\pm$	10	&	11	&		\\
2MASS J18432213+4040209	&	M7.5	&	69	&	69.4569	$\pm$	0.0906	&	14.4	$\pm$	0.0	&	30	&	$<$16			&	82	&		\\
2MASS J18544597+8429470	&	M9	&	72	&	37.163	$\pm$	0.2005	&	26.9	$\pm$	0.1	&	30	&	$<$87			&	58	&		\\
2MASS J19165762+0509021	&	M8	&	39	&	168.962	$\pm$	0.1299	&	5.9	$\pm$	0.0	&	30	&	$<$79.8			&	43	&	a	\\
2MASS J20370715-1137569	&	M8	&	74	&	46.6796	$\pm$	0.1355	&	21.4	$\pm$	0.1	&	30	&	$<$33			&	10	&		\\
2MASS J22373255+3922398 	&	M9.5	&	74	&	47.6127	$\pm$	0.1972	&	21.0	$\pm$	0.1	&	30	&	$<$81			&	58	&		\\
2MASS J23062928-0502285 	&	M8	&	74	&	80.4512	$\pm$	0.1211	&	12.4	$\pm$	0.0	&	30	&	$<$8.1			&	65	&		\\
2MASS J23312174-2749500	&	M7.0	&	24	&	73.2134	$\pm$	0.1575	&	13.7	$\pm$	0.0	&	30	&	$<$72			&	58	&		\\
2MASS J23494899+1224386	&	M8.0	&	24	&	44.5694	$\pm$	0.1801	&	22.4	$\pm$	0.1	&	30	&	$<$60			&	58	&		\\
2MASS J23515044-2537367	&	M9	&	51	&	49.1253	$\pm$	0.4479	&	20.4	$\pm$	0.2	&	30	&	$<$69			&	58	&	b	\\
2MASS J23535946-0833311	&	M8.6	&	5	&	45.8744	$\pm$	0.2745	&	21.8	$\pm$	0.1	&	30	&	$<$69			&	58	&		\\
BRI 0021-0214	&	M9.5	&	44	&	79.9653	$\pm$	0.2212	&	12.5	$\pm$	0.0	&	30	&	83	$\pm$	18	&	8	&		\\
S Ori 55	&	M9	&	84	&	--			&	300.0	--	450.0	&	18	&	$<$66			&	10	&		\\
TVLM 513-46546	&	M8.5	&	41	&	93.4497	$\pm$	0.1945	&	10.7	$\pm$	0.0	&	30	&	308	$\pm$	16	&	8	&		\\
TWA 5B	&	M9	&	2	&	20.2517	$\pm$	0.0586	&	49.4	$\pm$	0.1	&	30	&	$<$84			&	63	&	c	\\
\midrule
L dwarfs	& 	&		&		&	&	&	& & 		\\
\midrule
2MASS J00303013-1450333 	&	L7	&	40	&	37.42	$\pm$	4.5	&	26.7	$\pm$	3.2	&	27	&	$<$57			&	10	&		\\
2MASS J00361617+1821104	&	L3.5	&	68	&	114.4167	$\pm$	0.2088	&	8.7	$\pm$	0.0	&	30	&	152	$\pm$	9	&	9	&		\\
2MASS J00452143+1634446	&	L0	&	28	&	65.0151	$\pm$	0.2274	&	15.4	$\pm$	0.1	&	30	&	$<$39			&	10	&		\\
2MASS J01075242+0041563	&	L8	&	76	&	64.13	$\pm$	4.51	&	15.6	$\pm$	1.1	&	27	&	$<$10.2			&	73	&		\\
2MASS J01410321+1804502	&	L0.8	&	5	&	42.2974	$\pm$	0.2923	&	23.6	$\pm$	0.2	&	30	&	$<$30			&	10	&		\\
2MASS J01443536-0716142	&	L6.5	&	76	&	79.0319	$\pm$	0.624	&	12.7	$\pm$	0.1	&	30	&	$<$33			&	10	&		\\
2MASS J02050344+1251422	&	L5	&	40	&	--			&	22.0	$\pm$	2.0	&	26	&	$<$48			&	10	&		\\
2MASS J02132880+4444453	&	L1.5	&	74	&	51.6812	$\pm$	0.3832	&	19.3	$\pm$	0.1	&	30	&	$<$30			&	10	&		\\
2MASS J02355993-2331205 	&	L1.5	&	76	&	46.5815	$\pm$	0.2863	&	21.5	$\pm$	0.1	&	30	&	$<$99			&	58	&		\\
2MASS J02511490-0352459	&	L3	&	74	&	90.62	$\pm$	3.02	&	11.0	$\pm$	0.4	&	7	&	$<$36			&	10	&		\\
2MASS J02550357-4700509	&	L9	&	76	&	205.3266	$\pm$	0.2545	&	4.9	$\pm$	0.0	&	30	&	$<$30.9			&	52	&		\\
2MASS J02572581-3105523 &	L8.5	&	76	&	102.3651	$\pm$	0.6073	&	9.8	$\pm$	0.1	&	30	&	$<$63.0			&	52	&		\\
2MASS J03400942-6724051	&	L7	&	26	&	107.1165	$\pm$	0.6174	&	9.3	$\pm$	0.1	&	30	&	$<$27.0			&	52	&		\\
2MASS J03454316+2540233	&	L0	&	76	&	37.4595	$\pm$	0.4188	&	26.7	$\pm$	0.3	&	30	&	$<$36			&	3	&		\\
2MASS J03552337+1133437	&	L3-L6	&	29	&	109.6451	$\pm$	0.7368	&	9.1	$\pm$	0.1	&	30	&	$<$45			&	4	&		\\
2MASS J04390101-2353083	&	L4.5	&	76	&	80.7917	$\pm$	0.5139	&	12.4	$\pm$	0.1	&	30	&	$<$42			&	10	&		\\
2MASS J04455387-3048204	&	L2	&	74	&	61.9685	$\pm$	0.1843	&	16.1	$\pm$	0.0	&	30	&	$<$66			&	10	&		\\
2MASS J05002100+0330501	&	L4	&	29	&	76.2093	$\pm$	0.3565	&	13.1	$\pm$	0.1	&	30	&	$<$51			&	4	&		\\
2MASS J05395200-0059019	&	L5	&	76	&	78.5318	$\pm$	0.5707	&	12.7	$\pm$	0.1	&	30	&	$<$48			&	4	&		\\
2MASS J06023045+3910592	&	L2	&	2	&	85.614	$\pm$	0.1663	&	11.7	$\pm$	0.0	&	30	&	$<$30			&	10	&		\\
2MASS J06523073+4710348	&	L4.5	&	6	&	109.6858	$\pm$	0.438	&	9.1	$\pm$	0.0	&	30	&	$<$33			&	10	&		\\
2MASS J08251968+2115521	&	L7	&	29	&	92.2242	$\pm$	1.184	&	10.8	$\pm$	0.1	&	30	&	$<$45			&	10	&		\\
2MASS J08283419-1309198	&	L2	&	77	&	85.5438	$\pm$	0.172	&	11.7	$\pm$	0.0	&	30	&	$<$63			&	4	&		\\
2MASS J08300825+4828482	&	L9.5	&	76	&	76.42	$\pm$	3.43	&	13.1	$\pm$	0.6	&	27	&	$<$87			&	4	&		\\
2MASS J08354256-0819237	&	L6.5	&	76	&	138.6098	$\pm$	0.2781	&	7.2	$\pm$	0.0	&	30	&	$<$14.7			&	73	&		\\
2MASS J08575849+5708514	&	L8	&	29	&	71.2343	$\pm$	1.0255	&	14.0	$\pm$	0.2	&	30	&	$<$51			&	4	&		\\
\bottomrule
\end{tabularx}
\begin{tablenotes}[]\footnotesize
\item[$a$]{	2MASS J19165762+0509021 is the widely separated companion to an  M3 primary at a distance of $\sim$400 AU \citep{vanBiesbroeck1944AJ.....51...61V}.  Its wide separation is easily resolvable and precludes magnetic interactions between components, so we include it in the single-object sample.}
\item[$b$]{	\citet{McLean2012ApJ...746...23M} list radio limits for primary and secondary components for 2MASS J23515044-2537367.   However, \citet{BardalezGagliuffi2019ApJ...883..205B} confirm that this object is not a spectroscopic binary.}	
\item[$c$]{ TWA 5B is part of a triple system and is separated from the binary TWA 5Aab by $127^{+4}_{-4}$ AU, which has an M1.5 primary \citep{Konopacky2007AJ....133.2008K, Kohler2013AA...558A..80K}.  Its wide separation is easily resolvable and precludes magnetic interactions between components, so we include TWA 5B in the single-object sample.}
	\end{tablenotes}
	\end{ThreePartTable}  
	\end{table*}
\end{center}

\setlength{\tabcolsep}{0.15in}
\begin{center}
	\begin{table*}\centering 
		\begin{ThreePartTable}
			\contcaption{Radio Observations of Ultracool Dwarfs from the Literature included in Occurrence Rate Calculation}
			\begin{tabularx}{\textwidth}{llccllcll}
			\toprule \vspace{2pt}
Object Name			&
SpT					&	
ref					&	
$\pi$				&	
$d$					&	
ref					&	
$F_{\nu}$			&	
ref					&	
Note	\\[-3pt]
					&	
					&	
					&	
(mas)				&	
(pc)				&	
					&	
($\mu$Jy)			&	
					&	
					\\%
\midrule  
\midrule  	
2MASS J09083803+5032088	&	L8	&	76	&	95.8202	$\pm$	0.6983	&	10.4	$\pm$	0.1	&	30	&	$<$111			&	4	&		\\
2MASS J09130320+1841501	&	L3	&	26	&	--			&	46.0	$\pm$	4.0	&	26	&	$<$102			&	58	&		\\
2MASS J09211410-2104446	&	L1	&	76	&	79.3128	$\pm$	0.2253	&	12.6	$\pm$	0.0	&	30	&	$<$72			&	4	&		\\
2MASS J09293364+3429527	&	L8	&	40	&	--			&	25.1	$\pm$	5.2	&	66	&	$<$42			&	10	&		\\
2MASS J10101480-0406499	&	L6	&	76	&	59.8	$\pm$	8.1	&	16.7	$\pm$	2.3	&	27	&	$<$47.4			&	73	&		\\
2MASS J10292165+1626526	&	L2.5	&	40	&	52.3361	$\pm$	0.7414	&	19.1	$\pm$	0.3	&	30	&	$<$33			&	58	&		\\
2MASS J10430758+2225236	&	L8.5	&	76	&	--			&	16.4	$\pm$	3.2	&	75	&	9.5	$\pm$	1.0	&	37	&		\\
2MASS J10433508+1213149	&	L9	&	76	&	68.5	$\pm$	10.6	&	14.6	$\pm$	2.3	&	27	&	$<$12.6			&	73	&		\\
2MASS J10452400-0149576 	&	L2	&	53	&	35.3233	$\pm$	0.3796	&	28.3	$\pm$	0.3	&	30	&	$<$57			&	58	&		\\
2MASS J10484281+0111580	&	L1	&	76	&	66.4589	$\pm$	0.2143	&	15.0	$\pm$	0.0	&	30	&	$<$21			&	58	&		\\
2MASS J10584787-1548172	&	L2.5	&	76	&	54.6468	$\pm$	0.5213	&	18.3	$\pm$	0.2	&	30	&	$<$10.5			&	73	&		\\
2MASS J11455714+2317297 	&	L1.5	&	1	&	24.0341	$\pm$	1.2986	&	41.6	$\pm$	2.2	&	30	&	$<$90			&	58	&		\\
2MASS J11593850+0057268	&	L0.4	&	5	&	31.0564	$\pm$	0.3643	&	32.2	$\pm$	0.4	&	30	&	$<$54			&	58	&		\\
2MASS J12035812+0015500	&	L5.0	&	5	&	67.2362	$\pm$	0.5553	&	14.9	$\pm$	0.1	&	30	&	$<$63			&	58	&		\\
2MASS J12195156+3128497	&	L9	&	76	&	--			&	18.1	$\pm$	3.7	&	75	&	$<$14.1			&	73	&		\\
2MASS J12212770+0257198 	&	L0.5	&	76	&	53.9501	$\pm$	0.2528	&	18.5	$\pm$	0.1	&	30	&	$<$78			&	58	&		\\
2MASS J12560183-1257276 B	&	L7.0	&	31	&	78.8	$\pm$	6.4	&	12.7	$\pm$	1.0	&	31	&	$<$9			&	34	&	d	\\
2MASS J13004255+1912354	&	L1.7	&	5	&	71.6755	$\pm$	0.2012	&	14.0	$\pm$	0.0	&	30	&	$<$87			&	10	&		\\
2MASS J13340623+1940351	&	L1	&	75	&	23.277	$\pm$	0.9101	&	43.0	$\pm$	1.7	&	30	&	$<$60			&	58	&		\\
2MASS J14122449+1633115	&	L0.5	&	40	&	31.9983	$\pm$	0.3124	&	31.3	$\pm$	0.3	&	30	&	$<$69			&	58	&		\\
2MASS J14243909+0917104	&	L4	&	44	&	31.7	$\pm$	2.5	&	31.5	$\pm$	2.5	&	27	&	$<$97			&	8	&		\\
2MASS J14252798-3650229	&	L4	&	29	&	84.5181	$\pm$	0.3435	&	11.8	$\pm$	0.0	&	30	&	$<$12.9			&	73	&		\\
2MASS J14380829+6408363	&	L0.2	&	5	&	58.7	$\pm$	0.143	&	17.0	$\pm$	0.0	&	30	&	$<$105			&	58	&		\\
2MASS J14392836+1929149	&	L1	&	76	&	69.6	$\pm$	0.5	&	14.4	$\pm$	0.1	&	27	&	$<$78			&	58	&		\\
2MASS J15065441+1321060	&	L3	&	76	&	85.581	$\pm$	0.2883	&	11.7	$\pm$	0.0	&	30	&	$<$78			&	58	&		\\
2MASS J15074769-1627386	&	L5	&	40	&	135.2332	$\pm$	0.3274	&	7.4	$\pm$	0.0	&	30	&	$<$36.6			&	52	&		\\
2MASS J15150083+4847416	&	L6.5	&	20	&	102.59	$\pm$	0.63	&	9.8	$\pm$	0.1	&	22	&	$<$27			&	10	&		\\
2MASS J15232263+3014562 	&	L8	&	40	&	53.7	$\pm$	1.24	&	18.6	$\pm$	0.4	&	27	&	$<$45			&	10	&		\\
2MASS J15551573-0956055	&	L1.6	&	5	&	73.6519	$\pm$	0.187	&	13.6	$\pm$	0.0	&	30	&	$<$84			&	58	&		\\
2MASS J16154255+4953211	&	L4	&	21	&	32	$\pm$	1	&	31.3	$\pm$	1.0	&	50	&	$<$9.0			&	73	&		\\
2MASS J16154416+3559005	&	L3	&	40	&	50.0611	$\pm$	0.3713	&	20.0	$\pm$	0.1	&	30	&	$<$75			&	58	&		\\
2MASS J16322911+1904407	&	L8	&	76	&	65.6	$\pm$	2.1	&	15.2	$\pm$	0.5	&	27	&	$<$10.8			&	73	&		\\
2MASS J16452211-1319516	&	L1	&	76	&	88.822	$\pm$	0.1444	&	11.3	$\pm$	0.0	&	30	&	$<$108			&	58	&		\\
2MASS J17054834-0516462	&	L1	&	76	&	52.6734	$\pm$	0.3516	&	19.0	$\pm$	0.1	&	30	&	$<$45			&	10	&		\\
2MASS J17210390+3344160	&	L5.3	&	5	&	61.3203	$\pm$	0.205	&	16.3	$\pm$	0.1	&	30	&	$<$11.4			&	73	&		\\
2MASS J17312974+2721233	&	L0	&	72	&	83.7364	$\pm$	0.1182	&	11.9	$\pm$	0.0	&	30	&	$<$57			&	4	&		\\
2MASS J17502484-0016151	&	L5	&	42	&	108.2676	$\pm$	0.2552	&	9.2	$\pm$	0.0	&	30	&	185	$\pm$	18	&	73	&		\\
2MASS J18071593+5015316	&	L1	&	76	&	68.3317	$\pm$	0.128	&	14.6	$\pm$	0.0	&	30	&	$<$39			&	10	&		\\
2MASS J18212815+1414010	&	L5	&	76	&	106.874	$\pm$	0.2518	&	9.4	$\pm$	0.0	&	30	&	$<$12.9			&	73	&		\\
2MASS J20575409-0252302	& L1	&	76	&	64.471	$\pm$	0.2365	&	15.5	$\pm$	0.1	&	30	&	$<$36			&	10	&		\\
2MASS J21041491-1037369	&	L2	&	76	&	58.1658	$\pm$	0.4051	&	17.2	$\pm$	0.1	&	30	&	$<$24			&	10	&		\\
2MASS J21481628+4003593	&	L7	&	76	&	123.2758	$\pm$	0.4557	&	8.1	$\pm$	0.0	&	30	&	$<$9.6			&	73	&		\\
2MASS J22244381-0158521	&	L4.5	&	40	&	86.6169	$\pm$	0.708	&	11.5	$\pm$	0.1	&	30	&	$<$33			&	10	&		\\
2MASSI J0030300-145033	&	L7	&	40	&	37.42	$\pm$	4.5	&	26.7	$\pm$	3.2	&	27	&	$<$17.4			&	73	&		\\
2MASSI J0103320+193536	&	L6	&	40	&	46.9	$\pm$	7.6	&	21.3	$\pm$	3.5	&	27	&	$<$11.4			&	73	&		\\
LP 944-20	&	L0	&	2	&	155.759	$\pm$	0.0991	&	6.4	$\pm$	0.0	&	30	&	137.6	$\pm$	10.5	&	52	&		\\
WISEP J060738.65+242953.4	&	L9	&	19	&	136.9449	$\pm$	0.6553	&	7.3	$\pm$	0.0	&	30	&	15.6	$\pm$	6.3	&	33	&	e	\\
\bottomrule
\end{tabularx}
\begin{tablenotes}[]\footnotesize
\item[$d$]{ 2MASS J12560183-1257276B is a low-mass companion to the equal-mass binary 2MASS J12560183-1257276Aab \citep{Stone2016ApJ...818L..12S}, implying that this system may be as far as $17.1 \pm 2.5$ pc rather than the $12.7\pm1.0$ pc reported by \citet{Gauza2015ApJ...804...96G}.  The wide $\geq 102 \pm 9$ AU separation between the primary and secondary components \citep{Gauza2015ApJ...804...96G} is easily resolved in all observations and precludes magnetic interactions between components, so we include it in the single-object sample. }
\item[$e$]{\cite{Gizis2016AJ....152..123G} report a 3.5$\sigma$ detection for WISEP J060738.65+242953.4.  However, this significance is inconsistent with their reported total flux densities of $15.6 \pm 6.3$ $\mu$Jy at a mean frequency of 6.05 GHz and $16.5 \pm 7.6$ and $15.6 \pm 10.7$ $\mu$Jy at 5.0 and 7.1 GHz, respectively.  The reported significance may use the image rms noise rather than the uncertainties returned from the fit.  They report an rms noise of 4.5 $\mu$Jy for the Stokes V image, but they do report a Stokes I rms noise. Both their reported significance and the implied significance of their reported flux densities do not meet our 4$\rms$ threshold for detection, so we classify this object as a non-detection.}		
\end{tablenotes}
	\end{ThreePartTable}  
	\end{table*}
\end{center}

\setlength{\tabcolsep}{0.15in}
\begin{center}
	\begin{table*}\centering 
		\begin{ThreePartTable}
			\contcaption{Radio Observations of Ultracool Dwarfs from the Literature included in Occurrence Rate Calculation }
			\begin{tabularx}{\textwidth}{llccllcll}
			\toprule \vspace{2pt}
Object Name			&
SpT					&	
ref					&	
$\pi$				&	
$d$					&	
ref					&	
$F_{\nu}$			&	
ref					&	
Note	\\[-3pt]
					&	
					&	
					&	
(mas)				&	
(pc)				&	
					&	
($\mu$Jy)			&	
					&	
					\\%
\midrule  
\midrule  
T/Y dwarfs	&		&	 		&		&	&	&	& & 	\\
\midrule	
2MASS J01514155+1244300	&	T0.0	&	54	&	46.73	$\pm$	3.37	&	21.4	$\pm$	1.5	&	27	&	$<$51			&	10	&		\\
2MASS J02074284+0000564	&	T4.5	&	49	&	34.85	$\pm$	9.87	&	28.7	$\pm$	8.1	&	27	&	$<$39			&	10	&		\\
2MASS J04151954-0935066	&	T8.0	&	13	&	174.34	$\pm$	2.76	&	5.7	$\pm$	0.1	&	27	&	$<$45			&	10	&		\\
2MASS J05591914-1404488	&	T4.5	&	49	&	94.4208	$\pm$	1.2021	&	10.6	$\pm$	0.1	&	30	&	$<$27			&	10	&		\\
2MASS J07271824+1710012	&	T7.0	&	13	&	110.14	$\pm$	2.34	&	9.1	$\pm$	0.2	&	27	&	$<$54			&	4	&		\\
2MASS J09373487+2931409	&	T6.0	&	13	&	163.39	$\pm$	1.76	&	6.1	$\pm$	0.1	&	27	&	$<$66			&	4	&		\\
2MASS J09393548-2448279	&	T8	&	26	&	196	$\pm$	10.4	&	5.1	$\pm$	0.3	&	27	&	$<$54			&	4	&		\\
2MASS J10475385+2124234	&	T6.5	&	26	&	94.73	$\pm$	3.81	&	10.6	$\pm$	0.4	&	27	&	21.9	$\pm$	1.3	&	37	&		\\
2MASS J11145133-2618235	&	T7.5	&	26	&	176.8	$\pm$	7	&	5.7	$\pm$	0.0	&	27	&	$<$60			&	4	&		\\
2MASS J12171110-0311131	&	T7.5	&	59	&	90.8	$\pm$	2.2	&	11.0	$\pm$	0.3	&	27	&	$<$111			&	10	&		\\
2MASS J12373919+6526148	&	T6.5	&	12	&	96.07	$\pm$	4.78	&	10.4	$\pm$	0.5	&	27	&	35	$\pm$	1	&	37	&		\\
2MASS J12545393-0122474	&	T2	&	15	&	74.1838	$\pm$	2.3057	&	13.5	$\pm$	0.4	&	30	&	4.0	$\pm$	3.3	&	36	&		\\
2MASS J13464634-0031501	&	T6.0	&	16	&	68.3	$\pm$	2.3	&	14.6	$\pm$	0.5	&	27	&	$<$68			&	8	&		\\	
2MASS J14571496-2121477	&	T8	&	12	&	169.3	$\pm$	1.7	&	5.9	$\pm$	0.1	&	27	&	$<$54			&	4	&		\\
2MASS J15031961+2525196	&	T5.5	&	14	&	154.9208	$\pm$	1.1025	&	6.5	$\pm$	0.0	&	30	&	$<$90			&	4	&		\\
2MASS J16241436+0029158	&	T6	&	13	&	90.9	$\pm$	1.2	&	11.0	$\pm$	0.1	&	27	&	$<$81			&	3	&		\\
2MASS J21543318+5942187	&	T6	&	26	&	--			&	10.0	$\pm$	1.0	&	26	&	$<$60			&	4	&		\\
Gl 229B	&	T6.5	&	12	&	173.19	$\pm$	1.12	&	5.8	$\pm$	0.0	&	27	&	$<$70.2			&	43	&		\\
SIMP J013656.5+093347.3	&	T2.5	&	26	&	163.6824	$\pm$	0.7223	&	6.1	$\pm$	0.0	&	30	&	85.7	$\pm$	1.3	&	37	&		\\
WISEP J112254.73+255021.5	&	T6	&	41	&	--			&	15.0	--	20.0	&	41	&	64	$\pm$	3	&	83	&		\\
WISE J085510.83-071442.5	&	Y2	&	46	&	448	$\pm$	33	&	2.2	$\pm$	0.2	&	79	&	$<$21.6			&	38	&		\\
WISE J140518.39+553421.3	&	Y0.5	&	45	&	133	$\pm$	81	&	7.5	$\pm$	4.6	&	55	&	$<$6.6			&	38	&		\\
WISEP J173835.53+273258.9	&	Y0	&	41	&	66	$\pm$	50	&	15.2	$\pm$	11.5	&	55	&	$<$9.6			&	38	&		\\
\bottomrule
\end{tabularx}
\begin{tablenotes}[]\footnotesize
\item[] \textit{References} ---					
(1)	\citet{Allen2007AJ....133..971A}	;
(2)	\citet{Allers2013ApJ...772...79A}	;
(3)	\citet{Antonova2008AA...487..317A}	;
(4)	\citet{Antonova2013AA...549A.131A}	;
(5)	\citet{BardalezGagliuffi2014ApJ...794..143B}	;
(6)	\citet{BardalezGagliuffi2019ApJ...883..205B}	;
(7)	\citet{Bartlett2017AJ....154..151B}	;
(8)	\citet{Berger2002ApJ...572..503B}	;
(9)	\citet{Berger2005ApJ...627..960B}	;
(10)	\citet{Berger2006ApJ...648..629B}	;
(11)	\citet{Berger2008ApJ...676.1307B}	;
(12)	\citet{Burgasser2000ApJ...531L..57B}	;
(13)	\citet{Burgasser2002ApJ...564..421B}	;
(14)	\citet{Burgasser2003AJ....125..850B}	;
(15)	\citet{Burgasser2003ApJ...594..510B}	;
(16)	\citet{Burgasser2010ApJ...710.1142B}	;
(17)	\citet{BurgasserPutnam2005ApJ...626..486B}	;
(18)	\citet{Caballero2008MNRAS.383..750C}	;
(19)	\citet{Castro2013ApJ...776..126C}	;
(20)	\citet{Cruz2003AJ....126.2421C}	;
(21)	\citet{Cruz2018AJ....155...34C}	;
(22)	\citet{Dahn2017AJ....154..147D}	;
(23)	\citet{Davison2015AJ....149..106D}	;
(24)	\citet{Deshpande2012AJ....144...99D}	;
(25)	\citet{Dieterich2014AJ....147...94D}	;
(26)	\citet{Faherty2009AJ....137....1F}	;
(27)	\citet{Faherty2012ApJ...752...56F}	;
(28)	\citet{Faherty2016ApJS..225...10F}	;
(29)	\citet{Gagne2015ApJS..219...33G}	;
(30)	\citet{Gaia2018yCat.1345....0G}	;
(31)	\citet{Gauza2015ApJ...804...96G}	;
(32)	\citet{Gizis2002ApJ...575..484G}	;
(33)	\citet{Gizis2016AJ....152..123G}	;
(34)	\citet{Guirado2018AA...610A..23G}	;
(35)	\citet{Henry2004AJ....128.2460H}	;
(36)	\citet{Kao2016ApJ...818...24K}	;
(37)	\citet{Kao2018ApJS..237...25K}	;
(38)	\citet{Kao2019MNRAS.487.1994K}	;
(39)	\citet{Kirkpatrick1991ApJS...77..417K}	;
(40)	\citet{Kirkpatrick2000AJ....120..447K}	;
(41)	\citet{Kirkpatrick2011ApJS..197...19K}	;
(42)	\citet{Koen2017MNRAS.465.4723K}	;
(43)	\citet{Krishnamurthi1999AJ....118.1369K}	;
(44)	\citet{Leggett2001ApJ...548..908L}	;
(45)	\citet{Leggett2013ApJ...763..130L}	;
(46)	\citet{Leggett2015ApJ...799...37L}	;
(47)	\citet{Lepine2002ApJ...581L..47L}	;
(48)	\citet{Lepine2003AJ....125.1598L}	;
(49)	\citet{Liu2006ApJ...647.1393L}	;
(50)	\citet{Liu2016ApJ...833...96L}	;
(51)	\citet{Lodieu2005AA...440.1061L}	;
(52)	\citet{Lynch2016MNRAS.457.1224L}	;
(53)	\citet{Marocco2013AJ....146..161M}	;
(54)	\citet{Marocco2015MNRAS.449.3651M}	;
(55)	\citet{Marsh2013ApJ...762..119M}	;
(56)	\citet{Martin2001ApJ...561L.195M}	;
(57)	\citet{McCaughrean2002AA...390L..27M}	;
(58)	\citet{McLean2012ApJ...746...23M}	;
(59)	\citet{Metchev2008ApJ...676.1281M}	;
(60)	\citet{Metodieva2015MNRAS.446.3878M}	;
(61)	\citet{Monin2010AA...515A..91M}	;
(62)	\citet{Newton2014AJ....147...20N}	;
(63)	\citet{OstenJayawardhana2006ApJ...644L..67O}	;
(64)	\citet{Phan-Bao2007ApJ...658..553P}	;
(65)	\citet{Pineda2018ApJ...866..155P}	;
(66)	\citet{Radigan2013ApJ...778...36R }	;
(67)	\citet{Ravi2011ApJ...735L...2R}	;
(68)	\citet{Reid2000AJ....119..369R}	;
(69)	\citet{Reid2003AJ....126.3007R}	;
(70)	\citet{Reid2004AJ....128..463R}	;
(71)	\citet{Reid2005PASP..117..676R}	;
(72)	\citet{Reid2008AJ....136.1290R}	;
(73)	\citet{Richey-Yowell2020}	;
(74)	\citet{Schmidt2007AJ....133.2258S}	;
(75)	\citet{Schmidt2010AJ....139.1808S}	;
(76)	\citet{Schneider2014AJ....147...34S}	;
(77)	\citet{Scholz2002MNRAS.336L..49S}	;
(78)	\citet{Shkolnik2009ApJ...699..649S}	;
(79)	\citet{Tinney2014ApJ...796...39T}	;
(80)	\citet{West2011AJ....141...97W}	;
(81)	\citet{West2015ApJ...812....3W}	;
(82)	\citet{Williams2014ApJ...785....9W}	;
(83)	\citet{Williams2017ApJ...834..117W}	;
(84)	\citet{ZapateroOsorio2000Sci...290..103Z}		
\end{tablenotes}
	\end{ThreePartTable}  
	\end{table*}
\end{center}

\setlength{\tabcolsep}{0.04in}
\begin{table*}\centering 
		\begin{ThreePartTable}
			\caption{Summary of mean simulation results: Literature distributions, KM  \label{table:simulation_literature_kaplanMeier}} 
			\begin{tabularx}{\textwidth}{rrr@{\hspace{0.1pt}}lr@{\hspace{0.1pt}}l@{\hspace{0.16in}}llr@{\hspace{0.1pt}}lr@{\hspace{0.1pt}}l@{\hspace{0.16in}}ll@{\hspace{0.16in}}llc}
   		&
		&
   		&
 		&
 		&
 		&
   		&
 		&
   		&
 		&
   		&
 		&
\multicolumn{2}{l}{Recovery rate for $\theta_{\mathrm{emit}}$ }  &
\multicolumn{2}{l}{Recovery rate for $\theta_{\mathrm{true}}$ } \\ 
$\bar{\theta}_{\mathrm{true}}$  							&
$\bar{\theta}_{\mathrm{emit}}$		    					&
\multicolumn{2}{c}{$\bar{\theta}_{\mathrm{calc}}$}   		&
\multicolumn{2}{c}{$\bar{\theta}_{\mathrm{detect}}$} 		&
$\abs{\bar{\epsilon}_{\mathrm{calc}}   }$   				&
$\abs{\bar{\epsilon}_{\mathrm{detect}} }$ 					&
\multicolumn{2}{c}{$\bar{f}_{\mathrm{calc}}$}   			&
\multicolumn{2}{c}{$\bar{f}_{\mathrm{detect}}$} 			&
$\theta_{\mathrm{calc}}$ CI 								&
$\theta_{\mathrm{detect}}$ CI 								&
$\theta_{\mathrm{calc}}$ CI 								&
$\theta_{\mathrm{detect}}$ CI \\
(\%)  						& 
(\%)  						& 
\multicolumn{2}{c}{(\%)}  	& 
\multicolumn{2}{c}{(\%)}  	& 
(\%)	  					& 
(\%)  						& 
	  						& 
  							& 
	  						& 
	  						& 
(\%)  						& 
(\%)  						& 
(\%)  						& 
(\%)  						\\ 
   		&
		&
   		&
   		&
   		&
   		&
 		&
 		&
 		&
   		&
 		&
 		&
68.3, 95.5, 99.7	&
68.3, 95.5, 99.7	&
68.3, 95.5, 99.7	&
68.3, 95.5, 99.7  	\\  [-2pt] 
\midrule  
\multicolumn{2}{l}{ N = 10}	& &  & &	&	&	&	&	& & & &	  &   &  \\ [-2pt] 
\midrule 
 10.0  &  9.8 	& 9.5 & $_{-5.3}^{+27.3}$ 	& 3.3 & $_{-3.0}^{+3.0}$ &    5 & 67   &  1.06 & $_{-0.59}^{+3.06}$   & 	3.05 & $_{-2.79}^{+2.41}$   &   88.8, 98.2, 100.0    &   58.6, 61.8, 62.3  &   83.6, 95.1, 99.4   &    26.6, 26.9, 26.9  \\[3pt] 
 20.0  &  20.1 	& 19.6 & $_{-10.0}^{+26.9}$ 	& 6.8 & $_{-5.5}^{+5.5}$ &    2 & 66   &  1.02 & $_{-0.52}^{+1.40}$   & 	2.94 & $_{-2.38}^{+2.04}$   &   80.8, 98.6, 100.0    &   41.8, 52.7, 58.0  &   80.2, 93.8, 99.7   &    47.1, 47.2, 47.2  \\[3pt] 
 30.0  &  29.8 	& 31.4 & $_{-14.2}^{+25.7}$ 	& 10.8 & $_{-7.5}^{+7.5}$ &    5 & 64   &  0.96 & $_{-0.43}^{+0.78}$   & 	2.78 & $_{-1.93}^{+1.89}$   &   68.4, 97.0, 100.0    &   36.6, 55.5, 64.4  &   58.1, 94.7, 99.6   &    18.6, 33.0, 65.0  \\[3pt] 
 40.0  &  40.3 	& 39.5 & $_{-16.5}^{+24.7}$ 	& 13.8 & $_{-8.7}^{+8.7}$ &    1 & 65   &  1.01 & $_{-0.42}^{+0.63}$   & 	2.89 & $_{-1.82}^{+1.79}$   &   64.7, 98.4, 100.0    &   24.0, 47.0, 63.7  &   56.4, 96.6, 99.9   &    26.8, 42.8, 72.7  \\[3pt] 
 50.0  &  50.2 	& 50.2 & $_{-19.4}^{+22.1}$ 	& 17.5 & $_{-10.1}^{+10.1}$ &    0 & 65   &  1.00 & $_{-0.38}^{+0.44}$   & 	2.85 & $_{-1.65}^{+1.67}$   &   65.4, 97.6, 100.0    &   17.2, 40.5, 62.9  &   53.6, 97.0, 100.0   &    19.1, 48.2, 64.8  \\[3pt] 
 60.0  &  59.8 	& 58.5 & $_{-21.2}^{+19.8}$ 	& 20.9 & $_{-10.9}^{+10.9}$ &    3 & 65   &  1.03 & $_{-0.37}^{+0.35}$   & 	2.87 & $_{-1.50}^{+1.55}$   &   64.6, 97.5, 100.0    &   11.0, 33.0, 55.6  &   53.2, 97.1, 100.0   &    10.1, 28.6, 56.9  \\[3pt] 
 70.0  &  70.0 	& 67.4 & $_{-22.7}^{+16.6}$ 	& 24.2 & $_{-11.5}^{+11.5}$ &    4 & 65   &  1.04 & $_{-0.35}^{+0.25}$   & 	2.89 & $_{-1.38}^{+1.46}$   &   67.4, 97.0, 100.0    &   5.3, 24.6, 46.3  &   60.6, 95.1, 100.0   &    0.9, 12.4, 30.4  \\[3pt] 
 80.0  &  80.0 	& 75.9 & $_{-24.1}^{+13.7}$ 	& 28.0 & $_{-12.2}^{+12.2}$ &    5 & 65   &  1.05 & $_{-0.34}^{+0.19}$   & 	2.85 & $_{-1.24}^{+1.39}$   &   73.5, 96.4, 99.6    &   2.8, 15.2, 38.7  &   73.6, 96.2, 99.9   &    2.0, 19.2, 39.3  \\[3pt] 
 90.0  &  90.2 	& 81.5 & $_{-24.2}^{+10.9}$ 	& 31.5 & $_{-12.7}^{+12.7}$ &    9 & 65   &  1.10 & $_{-0.33}^{+0.15}$   & 	2.86 & $_{-1.16}^{+1.25}$   &   74.9, 94.0, 98.5    &   1.0, 8.6, 24.0  &   76.5, 93.2, 99.9   &    0.9, 9.5, 26.0  \\[3pt] 
 100.0  &  100.0 	& 88.3 & $_{-24.8}^{+8.0}$ 	& 34.8 & $_{-13.2}^{+13.2}$ &    12 & 65   &  1.13 & $_{-0.32}^{+0.10}$   & 	2.87 & $_{-1.09}^{+1.20}$   &   76.9, 93.1, 97.7    &   0.1, 1.3, 10.4  &   76.9, 93.1, 97.7   &    0.1, 1.3, 10.4  \\ 
\midrule  
\multicolumn{2}{l}{ N = 20} & &  & &	&	&	&	&	& & & &	  &   &  \\ [-2pt] 
\midrule 
 10.0  &  10.0 	& 10.1 & $_{-5.9}^{+17.0}$ 	& 3.5 & $_{-2.7}^{+2.7}$ &    1 & 65   &  0.99 & $_{-0.58}^{+1.65}$   & 	2.87 & $_{-2.24}^{+2.13}$   &   85.5, 99.4, 100.0    &   44.6, 58.2, 62.3  &   85.4, 97.3, 99.3   &    49.6, 49.6, 49.6  \\[3pt] 
 20.0  &  19.8 	& 20.2 & $_{-10.1}^{+18.5}$ 	& 7.0 & $_{-4.4}^{+4.4}$ &    1 & 65   &  0.99 & $_{-0.49}^{+0.90}$   & 	2.86 & $_{-1.82}^{+1.81}$   &   74.7, 98.2, 100.0    &   26.2, 51.5, 67.3  &   63.2, 95.4, 99.1   &    21.0, 45.0, 74.3  \\[3pt] 
 30.0  &  30.5 	& 31.2 & $_{-13.4}^{+19.2}$ 	& 10.7 & $_{-5.9}^{+5.9}$ &    4 & 64   &  0.96 & $_{-0.41}^{+0.59}$   & 	2.79 & $_{-1.53}^{+1.57}$   &   73.7, 97.6, 100.0    &   12.3, 39.2, 65.1  &   64.4, 93.9, 99.6   &    4.4, 31.2, 52.4  \\[3pt] 
 40.0  &  40.0 	& 41.0 & $_{-15.6}^{+19.3}$ 	& 14.0 & $_{-7.0}^{+7.0}$ &    3 & 65   &  0.98 & $_{-0.37}^{+0.46}$   & 	2.86 & $_{-1.44}^{+1.46}$   &   69.5, 96.8, 100.0    &   5.3, 27.3, 51.6  &   62.8, 92.6, 99.7   &    4.1, 25.2, 52.2  \\[3pt] 
 50.0  &  50.5 	& 50.9 & $_{-16.8}^{+18.7}$ 	& 17.7 & $_{-8.1}^{+8.1}$ &    2 & 65   &  0.98 & $_{-0.32}^{+0.36}$   & 	2.83 & $_{-1.30}^{+1.25}$   &   65.5, 97.3, 99.9    &   2.3, 15.2, 35.9  &   59.0, 92.5, 99.8   &    3.3, 14.9, 39.8  \\[3pt] 
 60.0  &  60.1 	& 59.7 & $_{-17.8}^{+17.4}$ 	& 20.9 & $_{-9.2}^{+9.2}$ &    0 & 65   &  1.00 & $_{-0.30}^{+0.29}$   & 	2.87 & $_{-1.26}^{+1.17}$   &   66.7, 97.9, 99.8    &   0.2, 6.0, 22.3  &   58.4, 94.9, 99.9   &    0.8, 7.0, 21.8  \\[3pt] 
 70.0  &  69.8 	& 68.7 & $_{-18.7}^{+15.2}$ 	& 23.8 & $_{-9.9}^{+9.9}$ &    2 & 66   &  1.02 & $_{-0.28}^{+0.23}$   & 	2.94 & $_{-1.22}^{+1.08}$   &   64.8, 97.7, 99.9    &   0.2, 2.6, 11.2  &   59.1, 96.4, 100.0   &    0.0, 0.3, 3.7  \\[3pt] 
 80.0  &  80.0 	& 78.1 & $_{-19.0}^{+12.2}$ 	& 27.7 & $_{-10.5}^{+10.5}$ &    2 & 65   &  1.02 & $_{-0.25}^{+0.16}$   & 	2.88 & $_{-1.09}^{+0.98}$   &   70.3, 97.0, 99.8    &   0.0, 0.3, 5.3  &   66.2, 97.3, 99.8   &    0.0, 0.8, 4.3  \\[3pt] 
 90.0  &  90.3 	& 86.2 & $_{-19.1}^{+8.7}$ 	& 31.3 & $_{-10.9}^{+10.9}$ &    4 & 65   &  1.04 & $_{-0.23}^{+0.10}$   & 	2.88 & $_{-1.00}^{+0.93}$   &   78.9, 97.0, 99.6    &   0.0, 0.1, 1.7  &   81.5, 95.8, 99.3   &    0.0, 0.2, 1.2  \\[3pt] 
 100.0  &  100.0 	& 93.0 & $_{-18.9}^{+5.1}$ 	& 34.8 & $_{-11.3}^{+11.3}$ &    7 & 65   &  1.08 & $_{-0.22}^{+0.06}$   & 	2.87 & $_{-0.94}^{+0.87}$   &   81.6, 95.8, 99.2    &   0.0, 0.0, 0.0  &   81.6, 95.8, 99.2   &    0.0, 0.0, 0.0  \\ 
\midrule 
\multicolumn{2}{l}{ N = 50}	& &  & &	&	&	&	&	& & & &	  &   &  \\ [-2pt] 
\midrule 
 10.0  &  10.3 	& 10.3 & $_{-5.2}^{+9.1}$ 	& 3.5 & $_{-2.2}^{+2.2}$ &    3 & 65   &  0.97 & $_{-0.49}^{+0.86}$   & 	2.82 & $_{-1.79}^{+1.75}$   &   82.2, 99.7, 100.0    &   16.6, 47.7, 69.2  &   69.8, 97.8, 99.9   &    25.6, 51.8, 74.9  \\[3pt] 
 20.0  &  19.9 	& 20.0 & $_{-8.0}^{+11.0}$ 	& 6.9 & $_{-3.4}^{+3.4}$ &    0 & 65   &  1.00 & $_{-0.40}^{+0.55}$   & 	2.88 & $_{-1.42}^{+1.39}$   &   78.0, 98.3, 100.0    &   3.0, 20.1, 47.5  &   67.9, 94.6, 99.5   &    4.9, 23.1, 44.5  \\[3pt] 
 30.0  &  30.1 	& 30.3 & $_{-10.1}^{+12.5}$ 	& 10.4 & $_{-4.2}^{+4.2}$ &    1 & 65   &  0.99 & $_{-0.33}^{+0.41}$   & 	2.87 & $_{-1.15}^{+1.13}$   &   74.6, 97.8, 100.0    &   0.5, 5.7, 20.0  &   64.9, 94.8, 99.6   &    0.0, 2.6, 14.0  \\[3pt] 
 40.0  &  40.2 	& 40.2 & $_{-11.5}^{+13.3}$ 	& 13.9 & $_{-4.7}^{+4.7}$ &    1 & 65   &  0.99 & $_{-0.29}^{+0.33}$   & 	2.89 & $_{-0.97}^{+1.00}$   &   72.0, 96.7, 99.9    &   0.0, 0.7, 6.5  &   65.5, 94.3, 99.6   &    0.1, 0.9, 6.2  \\[3pt] 
 50.0  &  49.8 	& 50.1 & $_{-12.5}^{+13.7}$ 	& 17.2 & $_{-5.3}^{+5.3}$ &    0 & 66   &  1.00 & $_{-0.25}^{+0.27}$   & 	2.91 & $_{-0.89}^{+0.88}$   &   72.3, 96.7, 99.9    &   0.0, 0.1, 1.6  &   64.6, 92.3, 99.5   &    0.0, 0.2, 2.0  \\[3pt] 
 60.0  &  60.1 	& 60.7 & $_{-13.1}^{+13.4}$ 	& 21.1 & $_{-5.8}^{+5.8}$ &    1 & 65   &  0.99 & $_{-0.21}^{+0.22}$   & 	2.85 & $_{-0.79}^{+0.76}$   &   72.0, 97.3, 100.0    &   0.0, 0.0, 0.1  &   65.0, 93.7, 99.8   &    0.0, 0.0, 0.3  \\[3pt] 
 70.0  &  70.2 	& 69.6 & $_{-13.2}^{+12.5}$ 	& 23.7 & $_{-6.1}^{+6.1}$ &    1 & 66   &  1.01 & $_{-0.19}^{+0.18}$   & 	2.95 & $_{-0.76}^{+0.73}$   &   67.8, 96.7, 99.6    &   0.0, 0.0, 0.0  &   60.7, 95.2, 99.7   &    0.0, 0.0, 0.0  \\[3pt] 
 80.0  &  79.9 	& 80.5 & $_{-13.0}^{+9.8}$ 	& 28.0 & $_{-6.4}^{+6.4}$ &    1 & 65   &  0.99 & $_{-0.16}^{+0.12}$   & 	2.86 & $_{-0.65}^{+0.64}$   &   66.8, 97.6, 99.7    &   0.0, 0.0, 0.0  &   59.9, 96.6, 99.9   &    0.0, 0.0, 0.0  \\[3pt] 
 90.0  &  89.9 	& 88.5 & $_{-12.6}^{+7.1}$ 	& 30.8 & $_{-6.6}^{+6.6}$ &    2 & 66   &  1.02 & $_{-0.15}^{+0.08}$   & 	2.92 & $_{-0.62}^{+0.61}$   &   76.7, 98.5, 99.9    &   0.0, 0.0, 0.0  &   73.7, 97.6, 99.9   &    0.0, 0.0, 0.0  \\[3pt] 
 100.0  &  100.0 	& 96.3 & $_{-11.7}^{+2.9}$ 	& 34.9 & $_{-6.7}^{+6.7}$ &    4 & 65   &  1.04 & $_{-0.13}^{+0.03}$   & 	2.87 & $_{-0.55}^{+0.56}$   &   85.5, 97.2, 99.9    &   0.0, 0.0, 0.0  &   85.5, 97.2, 99.9   &    0.0, 0.0, 0.0  \\
\midrule  
\multicolumn{2}{l}{ N = 100} & &  & &	&	&	&	&	& & & &	  &   &  \\ [-2pt] 
\midrule 
 10.0  &  10.0 	& 10.1 & $_{-4.2}^{+6.0}$ 	& 3.5 & $_{-1.7}^{+1.7}$ &    1 & 65   &  0.99 & $_{-0.41}^{+0.58}$   & 	2.86 & $_{-1.41}^{+1.40}$   &   77.7, 98.4, 100.0    &   2.6, 23.3, 49.8  &   69.3, 93.8, 99.8   &    6.0, 23.9, 46.2  \\[3pt] 
 20.0  &  20.2 	& 20.0 & $_{-6.2}^{+7.7}$ 	& 6.9 & $_{-2.5}^{+2.5}$ &    0 & 66   &  1.00 & $_{-0.31}^{+0.38}$   & 	2.90 & $_{-1.04}^{+1.04}$   &   76.1, 98.1, 100.0    &   0.0, 1.1, 8.0  &   66.9, 95.0, 99.6   &    0.3, 1.6, 12.1  \\[3pt] 
 30.0  &  29.9 	& 30.0 & $_{-7.6}^{+8.8}$ 	& 10.4 & $_{-3.0}^{+3.0}$ &    0 & 65   &  1.00 & $_{-0.25}^{+0.29}$   & 	2.89 & $_{-0.84}^{+0.83}$   &   75.7, 98.5, 100.0    &   0.0, 0.2, 0.7  &   67.1, 96.2, 99.9   &    0.0, 0.0, 0.3  \\[3pt] 
 40.0  &  39.6 	& 39.6 & $_{-8.6}^{+9.5}$ 	& 13.6 & $_{-3.4}^{+3.4}$ &    1 & 66   &  1.01 & $_{-0.22}^{+0.24}$   & 	2.94 & $_{-0.73}^{+0.73}$   &   75.2, 98.2, 100.0    &   0.0, 0.0, 0.0  &   66.8, 95.5, 99.6   &    0.0, 0.0, 0.0  \\[3pt] 
 50.0  &  49.9 	& 50.3 & $_{-9.3}^{+9.9}$ 	& 17.3 & $_{-3.8}^{+3.8}$ &    1 & 65   &  0.99 & $_{-0.18}^{+0.20}$   & 	2.88 & $_{-0.63}^{+0.62}$   &   74.1, 97.4, 99.6    &   0.0, 0.0, 0.0  &   67.0, 94.6, 99.2   &    0.0, 0.0, 0.0  \\[3pt] 
 60.0  &  59.9 	& 59.7 & $_{-9.7}^{+10.0}$ 	& 20.6 & $_{-4.1}^{+4.1}$ &    0 & 66   &  1.00 & $_{-0.16}^{+0.17}$   & 	2.92 & $_{-0.57}^{+0.57}$   &   72.0, 95.7, 100.0    &   0.0, 0.0, 0.0  &   66.0, 92.2, 99.7   &    0.0, 0.0, 0.0  \\[3pt] 
 70.0  &  70.1 	& 70.5 & $_{-9.9}^{+9.7}$ 	& 24.2 & $_{-4.2}^{+4.2}$ &    1 & 65   &  0.99 & $_{-0.14}^{+0.14}$   & 	2.90 & $_{-0.51}^{+0.51}$   &   70.7, 96.7, 99.9    &   0.0, 0.0, 0.0  &   64.6, 92.6, 99.6   &    0.0, 0.0, 0.0  \\[3pt] 
 80.0  &  79.9 	& 80.2 & $_{-9.6}^{+8.4}$ 	& 27.5 & $_{-4.4}^{+4.4}$ &    0 & 66   &  1.00 & $_{-0.12}^{+0.10}$   & 	2.91 & $_{-0.46}^{+0.47}$   &   67.4, 97.7, 100.0    &   0.0, 0.0, 0.0  &   62.4, 95.3, 99.9   &    0.0, 0.0, 0.0  \\[3pt] 
 90.0  &  90.0 	& 89.9 & $_{-8.9}^{+5.7}$ 	& 31.1 & $_{-4.5}^{+4.5}$ &    0 & 65   &  1.00 & $_{-0.10}^{+0.06}$   & 	2.89 & $_{-0.42}^{+0.43}$   &   70.9, 97.7, 99.9    &   0.0, 0.0, 0.0  &   65.9, 97.6, 99.8   &    0.0, 0.0, 0.0  \\[3pt] 
 100.0  &  100.0 	& 97.6 & $_{-8.0}^{+1.9}$ 	& 34.5 & $_{-4.7}^{+4.7}$ &    2 & 66   &  1.03 & $_{-0.08}^{+0.02}$   & 	2.90 & $_{-0.39}^{+0.39}$   &   85.1, 97.9, 99.8    &   0.0, 0.0, 0.0  &   85.1, 97.9, 99.8   &    0.0, 0.0, 0.0  \\ 
\bottomrule
\end{tabularx}
	\begin{tablenotes}[]\footnotesize
		\item[]\textit{Note} --- Uncertainties for $\bar{\theta}_{\mathrm{calc}}$ and $\bar{\theta}_{\mathrm{emit}}$ correspond to the mean values of the  68.3\% credible and confidence intervals, respectively.   These uncertainties are propagated for correction factors $\bar{f}_{\mathrm{calc}}$ and $\bar{f}_{\mathrm{calc}}$, which are calculated by dividing the true rate by $\bar{\theta}_{\mathrm{calc}}$ and $\bar{\theta}_{\mathrm{detect}}$, respectively. $\epsilon_{\text{calc}}$  and  $\epsilon_{\text{calc}}$  are the percent errors of the calculated mean occurrence and detection rates with respect to the true rate.
	\end{tablenotes}
\end{ThreePartTable}  
\end{table*}
 
\setlength{\tabcolsep}{0.04in}
\begin{table*}\centering 
		\begin{ThreePartTable}
			\caption{Summary of mean simulation results: Literature, Uniform  \label{table:simulation_literature_uniform}}
			\begin{tabularx}{\textwidth}{rrr@{\hspace{0.1pt}}lr@{\hspace{0.1pt}}l@{\hspace{0.16in}}llr@{\hspace{0.1pt}}lr@{\hspace{0.1pt}}l@{\hspace{0.16in}}ll@{\hspace{0.16in}}llc}
   		&
		&
   		&
 		&
 		&
 		&
   		&
 		&
   		&
 		&
   		&
 		&
\multicolumn{2}{l}{Recovery rate for $\theta_{\mathrm{emit}}$ }  &
\multicolumn{2}{l}{Recovery rate for $\theta_{\mathrm{true}}$ } \\ 
$\bar{\theta}_{\mathrm{true}}$  							&
$\bar{\theta}_{\mathrm{emit}}$		    					&
\multicolumn{2}{c}{$\bar{\theta}_{\mathrm{calc}}$}   		&
\multicolumn{2}{c}{$\bar{\theta}_{\mathrm{detect}}$} 		&
$\abs{\bar{\epsilon}_{\mathrm{calc}}   }$   				&
$\abs{\bar{\epsilon}_{\mathrm{detect}} }$ 					&
\multicolumn{2}{c}{$\bar{f}_{\mathrm{calc}}$}   			&
\multicolumn{2}{c}{$\bar{f}_{\mathrm{detect}}$} 			&
$\theta_{\mathrm{calc}}$ CI 								&
$\theta_{\mathrm{detect}}$ CI 								&
$\theta_{\mathrm{calc}}$ CI 								&
$\theta_{\mathrm{detect}}$ CI \\
(\%)  						& 
(\%)  						& 
\multicolumn{2}{c}{(\%)}  	& 
\multicolumn{2}{c}{(\%)}  	& 
(\%)	  					& 
(\%)  						& 
	  						& 
  							& 
	  						& 
	  						& 
(\%)  						& 
(\%)  						& 
(\%)  						& 
(\%)  						\\ 
   		&
		&
   		&
   		&
   		&
   		&
 		&
 		&
 		&
   		&
 		&
 		&
68.3, 95.5, 99.7	&
68.3, 95.5, 99.7	&
68.3, 95.5, 99.7	&
68.3, 95.5, 99.7  	\\  [-2pt] 
\midrule  
\multicolumn{2}{l}{ N = 10}	& &  & &	&	&	&	&	& & & &	  &   &  \\ [-2pt] 
\midrule 
 10.0  &  10.1 	& 10.3 & $_{-5.9}^{+20.8}$ 	& 5.0 & $_{-4.2}^{+4.2}$ &    3 & 50   &  0.97 & $_{-0.55}^{+1.95}$   & 	1.99 & $_{-1.66}^{+1.59}$   &   92.8, 100.0, 100.0    &   69.6, 73.2, 73.6  &   87.7, 95.8, 99.5   &    37.8, 38.7, 38.7  \\[3pt] 
 20.0  &  20.0 	& 20.3 & $_{-10.1}^{+20.9}$ 	& 9.8 & $_{-7.0}^{+7.0}$ &    2 & 51   &  0.98 & $_{-0.49}^{+1.01}$   & 	2.04 & $_{-1.46}^{+1.34}$   &   82.4, 98.9, 100.0    &   55.9, 67.3, 70.7  &   63.1, 94.7, 98.8   &    58.3, 59.4, 59.4  \\[3pt] 
 30.0  &  30.2 	& 30.8 & $_{-13.8}^{+21.1}$ 	& 14.8 & $_{-9.0}^{+9.0}$ &    3 & 51   &  0.97 & $_{-0.44}^{+0.67}$   & 	2.02 & $_{-1.23}^{+1.28}$   &   75.1, 98.7, 100.0    &   50.8, 70.5, 77.3  &   61.6, 94.8, 99.4   &    29.7, 44.7, 76.6  \\[3pt] 
 40.0  &  39.9 	& 40.9 & $_{-16.1}^{+20.4}$ 	& 19.7 & $_{-10.8}^{+10.8}$ &    2 & 51   &  0.98 & $_{-0.39}^{+0.49}$   & 	2.03 & $_{-1.11}^{+1.09}$   &   72.3, 98.3, 100.0    &   40.5, 68.0, 81.4  &   59.5, 94.6, 100.0   &    42.5, 60.2, 85.7  \\[3pt] 
 50.0  &  49.2 	& 48.4 & $_{-17.6}^{+19.8}$ 	& 23.6 & $_{-11.4}^{+11.4}$ &    3 & 53   &  1.03 & $_{-0.38}^{+0.42}$   & 	2.12 & $_{-1.02}^{+1.09}$   &   70.9, 98.4, 100.0    &   33.2, 60.4, 79.3  &   59.1, 94.5, 100.0   &    36.5, 65.4, 77.2  \\[3pt] 
 60.0  &  59.9 	& 60.2 & $_{-19.2}^{+17.4}$ 	& 29.6 & $_{-12.4}^{+12.4}$ &    0 & 51   &  1.00 & $_{-0.32}^{+0.29}$   & 	2.03 & $_{-0.85}^{+0.93}$   &   70.6, 97.5, 100.0    &   26.8, 55.3, 78.6  &   56.1, 94.8, 100.0   &    26.1, 52.5, 79.2  \\[3pt] 
 70.0  &  70.0 	& 69.6 & $_{-20.3}^{+15.1}$ 	& 34.1 & $_{-12.9}^{+12.9}$ &    1 & 51   &  1.01 & $_{-0.29}^{+0.22}$   & 	2.06 & $_{-0.78}^{+0.87}$   &   72.3, 98.0, 99.9    &   16.7, 48.3, 72.4  &   61.5, 96.1, 99.6   &    5.3, 31.4, 56.2  \\[3pt] 
 80.0  &  79.7 	& 77.0 & $_{-20.9}^{+12.5}$ 	& 38.1 & $_{-13.6}^{+13.6}$ &    4 & 52   &  1.04 & $_{-0.28}^{+0.17}$   & 	2.10 & $_{-0.75}^{+0.80}$   &   74.9, 97.3, 99.8    &   8.3, 34.7, 61.9  &   69.1, 95.3, 99.9   &    8.3, 39.7, 62.9  \\[3pt] 
 90.0  &  89.5 	& 86.8 & $_{-21.0}^{+8.3}$ 	& 44.2 & $_{-14.3}^{+14.3}$ &    4 & 51   &  1.04 & $_{-0.25}^{+0.10}$   & 	2.04 & $_{-0.66}^{+0.68}$   &   81.2, 97.0, 99.9    &   3.9, 27.7, 54.3  &   81.3, 96.9, 99.9   &    4.0, 28.9, 55.7  \\[3pt] 
 100.0  &  100.0 	& 93.7 & $_{-20.9}^{+4.4}$ 	& 48.9 & $_{-14.3}^{+14.3}$ &    6 & 51   &  1.07 & $_{-0.24}^{+0.05}$   & 	2.05 & $_{-0.60}^{+0.62}$   &   84.7, 95.2, 98.9    &   1.7, 9.2, 37.9  &   84.7, 95.2, 98.9   &    1.7, 9.2, 37.9  \\
\midrule 
\multicolumn{2}{l}{ N = 20}	& &  & &	&	&	&	&	& & & &	  &   &  \\ [-2pt] 
\midrule 
 10.0  &  10.4 	& 10.6 & $_{-5.8}^{+12.9}$ 	& 5.2 & $_{-3.6}^{+3.6}$ &    6 & 48   &  0.94 & $_{-0.52}^{+1.14}$   & 	1.93 & $_{-1.35}^{+1.31}$   &   88.8, 99.8, 100.0    &   59.2, 70.6, 74.3  &   73.8, 96.8, 99.6   &    62.9, 63.3, 63.3  \\[3pt] 
 20.0  &  19.8 	& 19.4 & $_{-9.1}^{+14.3}$ 	& 9.4 & $_{-5.4}^{+5.4}$ &    3 & 53   &  1.03 & $_{-0.48}^{+0.76}$   & 	2.12 & $_{-1.21}^{+1.24}$   &   83.1, 99.2, 100.0    &   40.8, 70.7, 82.1  &   67.8, 98.0, 99.6   &    35.2, 61.6, 83.9  \\[3pt] 
 30.0  &  30.6 	& 30.4 & $_{-12.0}^{+15.6}$ 	& 14.7 & $_{-7.2}^{+7.2}$ &    1 & 51   &  0.99 & $_{-0.39}^{+0.51}$   & 	2.04 & $_{-1.00}^{+1.01}$   &   77.6, 98.6, 100.0    &   28.1, 61.4, 80.8  &   65.5, 93.9, 99.7   &    13.7, 51.4, 70.8  \\[3pt] 
 40.0  &  40.4 	& 40.4 & $_{-13.7}^{+16.1}$ 	& 19.7 & $_{-8.8}^{+8.8}$ &    1 & 51   &  0.99 & $_{-0.34}^{+0.39}$   & 	2.04 & $_{-0.91}^{+0.85}$   &   75.7, 98.7, 100.0    &   15.8, 51.2, 76.7  &   64.2, 94.0, 99.8   &    15.1, 51.2, 75.4  \\[3pt] 
 50.0  &  50.1 	& 50.0 & $_{-14.9}^{+15.9}$ 	& 24.0 & $_{-9.7}^{+9.7}$ &    0 & 52   &  1.00 & $_{-0.30}^{+0.32}$   & 	2.08 & $_{-0.84}^{+0.78}$   &   73.8, 97.4, 100.0    &   8.7, 37.4, 66.4  &   58.7, 92.2, 99.8   &    13.6, 36.1, 62.6  \\[3pt] 
 60.0  &  60.8 	& 61.2 & $_{-15.5}^{+14.9}$ 	& 29.6 & $_{-10.8}^{+10.8}$ &    2 & 51   &  0.98 & $_{-0.25}^{+0.24}$   & 	2.03 & $_{-0.74}^{+0.66}$   &   74.3, 97.7, 100.0    &   4.2, 23.3, 52.9  &   61.9, 92.7, 99.7   &    7.9, 30.5, 54.7  \\[3pt] 
 70.0  &  69.5 	& 69.6 & $_{-15.7}^{+13.4}$ 	& 33.9 & $_{-11.2}^{+11.2}$ &    1 & 52   &  1.01 & $_{-0.23}^{+0.19}$   & 	2.06 & $_{-0.68}^{+0.62}$   &   72.2, 98.4, 100.0    &   2.2, 17.5, 41.1  &   60.1, 95.2, 99.6   &    1.1, 8.3, 24.9  \\[3pt] 
 80.0  &  80.2 	& 80.4 & $_{-15.8}^{+10.4}$ 	& 38.6 & $_{-11.7}^{+11.7}$ &    0 & 52   &  1.00 & $_{-0.20}^{+0.13}$   & 	2.07 & $_{-0.62}^{+0.59}$   &   75.0, 98.3, 99.9    &   0.8, 6.6, 21.2  &   63.6, 97.6, 99.6   &    1.4, 7.0, 22.2  \\[3pt] 
 90.0  &  89.7 	& 88.1 & $_{-15.1}^{+7.3}$ 	& 43.5 & $_{-11.8}^{+11.8}$ &    2 & 52   &  1.02 & $_{-0.18}^{+0.09}$   & 	2.07 & $_{-0.56}^{+0.54}$   &   81.5, 98.2, 99.7    &   0.2, 2.5, 10.4  &   79.0, 97.2, 99.7   &    0.3, 2.0, 10.2  \\[3pt] 
 100.0  &  100.0 	& 95.3 & $_{-14.5}^{+3.5}$ 	& 47.7 & $_{-11.9}^{+11.9}$ &    5 & 52   &  1.05 & $_{-0.16}^{+0.04}$   & 	2.10 & $_{-0.52}^{+0.51}$   &   81.8, 96.9, 99.4    &   0.0, 0.0, 1.0  &   81.8, 96.9, 99.4   &    0.0, 0.0, 1.0  \\
\midrule 
\multicolumn{2}{l}{ N = 50}	& &  & &	&	&	&	&	& & & &	  &   &  \\ [-2pt] 
\midrule 
 10.0  &  10.2 	& 10.3 & $_{-4.7}^{+7.3}$ 	& 5.0 & $_{-2.8}^{+2.8}$ &    3 & 50   &  0.97 & $_{-0.44}^{+0.68}$   & 	2.01 & $_{-1.14}^{+1.12}$   &   86.9, 99.8, 100.0    &   35.8, 72.3, 87.2  &   69.1, 97.9, 100.0   &    45.3, 73.0, 85.9  \\[3pt] 
 20.0  &  20.1 	& 20.1 & $_{-7.0}^{+8.9}$ 	& 9.8 & $_{-4.1}^{+4.1}$ &    1 & 51   &  0.99 & $_{-0.35}^{+0.44}$   & 	2.04 & $_{-0.84}^{+0.84}$   &   81.5, 99.6, 100.0    &   10.7, 47.9, 77.3  &   64.7, 94.8, 100.0   &    19.3, 49.3, 70.0  \\[3pt] 
 30.0  &  29.9 	& 30.2 & $_{-8.6}^{+10.0}$ 	& 14.6 & $_{-4.8}^{+4.8}$ &    1 & 51   &  0.99 & $_{-0.28}^{+0.33}$   & 	2.06 & $_{-0.68}^{+0.68}$   &   82.4, 98.9, 100.0    &   2.5, 24.9, 56.8  &   68.7, 95.5, 99.9   &    0.9, 15.0, 43.5  \\[3pt] 
 40.0  &  40.2 	& 40.8 & $_{-9.7}^{+10.6}$ 	& 19.7 & $_{-5.7}^{+5.7}$ &    2 & 51   &  0.98 & $_{-0.23}^{+0.26}$   & 	2.03 & $_{-0.59}^{+0.57}$   &   78.1, 99.0, 100.0    &   1.0, 9.4, 34.0  &   66.0, 93.6, 99.7   &    2.4, 13.8, 34.6  \\[3pt] 
 50.0  &  49.7 	& 50.2 & $_{-10.4}^{+10.9}$ 	& 24.3 & $_{-6.2}^{+6.2}$ &    0 & 51   &  1.00 & $_{-0.21}^{+0.22}$   & 	2.06 & $_{-0.52}^{+0.51}$   &   79.9, 98.2, 100.0    &   0.0, 2.7, 16.4  &   69.2, 95.3, 99.7   &    0.1, 4.1, 16.9  \\[3pt] 
 60.0  &  60.0 	& 60.1 & $_{-10.8}^{+10.8}$ 	& 29.0 & $_{-6.5}^{+6.5}$ &    0 & 52   &  1.00 & $_{-0.18}^{+0.18}$   & 	2.07 & $_{-0.46}^{+0.45}$   &   77.9, 98.3, 99.9    &   0.0, 0.4, 5.1  &   68.2, 94.8, 99.7   &    0.0, 1.3, 5.1  \\[3pt] 
 70.0  &  70.0 	& 70.6 & $_{-10.7}^{+10.1}$ 	& 34.1 & $_{-6.6}^{+6.6}$ &    1 & 51   &  0.99 & $_{-0.15}^{+0.14}$   & 	2.05 & $_{-0.40}^{+0.41}$   &   76.3, 98.3, 99.9    &   0.0, 0.2, 1.4  &   68.0, 94.1, 99.5   &    0.0, 0.1, 0.6  \\[3pt] 
 80.0  &  80.1 	& 80.5 & $_{-10.2}^{+8.4}$ 	& 38.9 & $_{-6.8}^{+6.8}$ &    1 & 51   &  0.99 & $_{-0.13}^{+0.10}$   & 	2.06 & $_{-0.36}^{+0.37}$   &   71.0, 97.8, 100.0    &   0.0, 0.0, 0.5  &   61.7, 95.7, 99.7   &    0.0, 0.0, 0.2  \\[3pt] 
 90.0  &  90.0 	& 90.1 & $_{-9.4}^{+5.8}$ 	& 43.7 & $_{-6.9}^{+6.9}$ &    0 & 51   &  1.00 & $_{-0.10}^{+0.06}$   & 	2.06 & $_{-0.32}^{+0.33}$   &   78.8, 98.8, 100.0    &   0.0, 0.0, 0.0  &   71.2, 98.5, 99.8   &    0.0, 0.0, 0.0  \\[3pt] 
 100.0  &  100.0 	& 97.8 & $_{-8.1}^{+1.7}$ 	& 48.5 & $_{-7.0}^{+7.0}$ &    2 & 51   &  1.02 & $_{-0.09}^{+0.02}$   & 	2.06 & $_{-0.30}^{+0.30}$   &   86.5, 98.3, 100.0    &   0.0, 0.0, 0.0  &   86.5, 98.3, 100.0   &    0.0, 0.0, 0.0  \\
\midrule 
\multicolumn{2}{l}{ N = 100}	& &  & &	&	&	&	&	& & & &	  &   &  \\ [-2pt] 
\midrule 
 10.0  &  10.0 	& 10.0 & $_{-3.6}^{+4.8}$ 	& 4.8 & $_{-2.0}^{+2.0}$ &    0 & 52   &  1.00 & $_{-0.36}^{+0.48}$   & 	2.07 & $_{-0.88}^{+0.88}$   &   82.9, 99.1, 100.0    &   13.2, 51.6, 78.2  &   68.6, 94.5, 100.0   &    17.4, 53.0, 71.6  \\[3pt] 
 20.0  &  20.0 	& 20.0 & $_{-5.3}^{+6.3}$ 	& 9.6 & $_{-2.9}^{+2.9}$ &    0 & 52   &  1.00 & $_{-0.27}^{+0.31}$   & 	2.08 & $_{-0.62}^{+0.63}$   &   82.2, 99.3, 99.9    &   0.7, 14.9, 43.3  &   68.0, 95.6, 99.6   &    2.6, 16.4, 40.8  \\[3pt] 
 30.0  &  30.2 	& 30.4 & $_{-6.4}^{+7.1}$ 	& 14.7 & $_{-3.5}^{+3.5}$ &    1 & 51   &  0.99 & $_{-0.21}^{+0.23}$   & 	2.04 & $_{-0.48}^{+0.49}$   &   82.8, 99.4, 100.0    &   0.1, 2.0, 12.7  &   67.2, 95.5, 99.8   &    0.3, 2.1, 13.0  \\[3pt] 
 40.0  &  40.0 	& 40.5 & $_{-7.1}^{+7.6}$ 	& 19.5 & $_{-4.0}^{+4.0}$ &    1 & 51   &  0.99 & $_{-0.17}^{+0.19}$   & 	2.05 & $_{-0.42}^{+0.41}$   &   76.8, 98.6, 100.0    &   0.0, 0.1, 2.7  &   65.5, 94.5, 99.7   &    0.0, 1.0, 4.5  \\[3pt] 
 50.0  &  50.3 	& 50.4 & $_{-7.6}^{+7.8}$ 	& 24.4 & $_{-4.2}^{+4.2}$ &    1 & 51   &  0.99 & $_{-0.15}^{+0.15}$   & 	2.05 & $_{-0.35}^{+0.36}$   &   81.8, 99.1, 100.0    &   0.0, 0.0, 0.2  &   69.1, 95.6, 99.7   &    0.0, 0.0, 0.4  \\[3pt] 
 60.0  &  59.6 	& 59.7 & $_{-7.8}^{+7.9}$ 	& 28.8 & $_{-4.4}^{+4.4}$ &    1 & 52   &  1.01 & $_{-0.13}^{+0.13}$   & 	2.08 & $_{-0.32}^{+0.33}$   &   80.9, 98.6, 100.0    &   0.0, 0.0, 0.0  &   69.8, 95.9, 99.8   &    0.0, 0.0, 0.0  \\[3pt] 
 70.0  &  69.9 	& 70.2 & $_{-7.8}^{+7.5}$ 	& 33.9 & $_{-4.6}^{+4.6}$ &    0 & 52   &  1.00 & $_{-0.11}^{+0.11}$   & 	2.07 & $_{-0.28}^{+0.28}$   &   74.7, 98.2, 100.0    &   0.0, 0.0, 0.0  &   69.1, 94.8, 99.8   &    0.0, 0.0, 0.0  \\[3pt] 
 80.0  &  80.1 	& 80.1 & $_{-7.5}^{+6.8}$ 	& 38.6 & $_{-4.8}^{+4.8}$ &    0 & 52   &  1.00 & $_{-0.09}^{+0.08}$   & 	2.07 & $_{-0.26}^{+0.26}$   &   71.5, 96.8, 100.0    &   0.0, 0.0, 0.0  &   63.0, 94.0, 99.8   &    0.0, 0.0, 0.0  \\[3pt] 
 90.0  &  89.9 	& 90.3 & $_{-6.5}^{+4.9}$ 	& 43.7 & $_{-5.0}^{+5.0}$ &    0 & 51   &  1.00 & $_{-0.07}^{+0.05}$   & 	2.06 & $_{-0.24}^{+0.23}$   &   70.6, 98.6, 99.8    &   0.0, 0.0, 0.0  &   62.2, 97.3, 99.8   &    0.0, 0.0, 0.0  \\[3pt] 
 100.0  &  100.0 	& 98.6 & $_{-5.2}^{+1.1}$ 	& 48.7 & $_{-5.0}^{+5.0}$ &    1 & 51   &  1.01 & $_{-0.05}^{+0.01}$   & 	2.05 & $_{-0.21}^{+0.21}$   &   87.4, 98.5, 100.0    &   0.0, 0.0, 0.0  &   87.4, 98.5, 100.0   &    0.0, 0.0, 0.0  \\
\bottomrule
\end{tabularx}
	\begin{tablenotes}[]\footnotesize
		\item[]\textit{Note} --- Uncertainties for $\bar{\theta}_{\mathrm{calc}}$ and $\bar{\theta}_{\mathrm{emit}}$ correspond to the mean values of the  68.3\% credible and confidence intervals, respectively.   These uncertainties are propagated for correction factors $\bar{f}_{\mathrm{calc}}$ and $\bar{f}_{\mathrm{calc}}$, which are calculated by dividing the true rate by $\bar{\theta}_{\mathrm{calc}}$ and $\bar{\theta}_{\mathrm{detect}}$, respectively. $\epsilon_{\text{calc}}$  and  $\epsilon_{\text{calc}}$  are the percent errors of the calculated mean occurrence and detection rates with respect to the true rate.
	\end{tablenotes}
\end{ThreePartTable}  
\end{table*}
 
\setlength{\tabcolsep}{0.04in}
\begin{table*}\centering 
		\begin{ThreePartTable}
			\caption{Summary of mean simulation results:  Low $\rms$, KM  \label{table:simulation_low-rms_kaplanMeier}} 
			\begin{tabularx}{\textwidth}{rrr@{\hspace{0.1pt}}lr@{\hspace{0.1pt}}l@{\hspace{0.16in}}llr@{\hspace{0.1pt}}lr@{\hspace{0.1pt}}l@{\hspace{0.16in}}ll@{\hspace{0.16in}}llc}
   		&
		&
   		&
 		&
 		&
 		&
   		&
 		&
   		&
 		&
   		&
 		&
\multicolumn{2}{l}{Recovery rate for $\theta_{\mathrm{emit}}$ }  &
\multicolumn{2}{l}{Recovery rate for $\theta_{\mathrm{true}}$ } \\ 
$\bar{\theta}_{\mathrm{true}}$  							&
$\bar{\theta}_{\mathrm{emit}}$		    					&
\multicolumn{2}{c}{$\bar{\theta}_{\mathrm{calc}}$}   		&
\multicolumn{2}{c}{$\bar{\theta}_{\mathrm{detect}}$} 		&
$\abs{\bar{\epsilon}_{\mathrm{calc}}   }$   				&
$\abs{\bar{\epsilon}_{\mathrm{detect}} }$ 					&
\multicolumn{2}{c}{$\bar{f}_{\mathrm{calc}}$}   			&
\multicolumn{2}{c}{$\bar{f}_{\mathrm{detect}}$} 			&
$\theta_{\mathrm{calc}}$ CI 								&
$\theta_{\mathrm{detect}}$ CI 								&
$\theta_{\mathrm{calc}}$ CI 								&
$\theta_{\mathrm{detect}}$ CI \\
(\%)  						& 
(\%)  						& 
\multicolumn{2}{c}{(\%)}  	& 
\multicolumn{2}{c}{(\%)}  	& 
(\%)	  					& 
(\%)  						& 
	  						& 
  							& 
	  						& 
	  						& 
(\%)  						& 
(\%)  						& 
(\%)  						& 
(\%)  						\\ 
   		&
		&
   		&
   		&
   		&
   		&
 		&
 		&
 		&
   		&
 		&
 		&
68.3, 95.5, 99.7	&
68.3, 95.5, 99.7	&
68.3, 95.5, 99.7	&
68.3, 95.5, 99.7  	\\  [-2pt] 
\midrule  
\multicolumn{2}{l}{ N = 10}	& &  & &	&	&	&	&	& & & &	  &   &  \\ [-2pt] 
\midrule 
 10.0  &  10.2 	& 10.2 & $_{-5.7}^{+14.5}$ 	& 7.8 & $_{-5.9}^{+5.9}$ &    2 & 23   &  0.98 & $_{-0.55}^{+1.40}$   & 	1.29 & $_{-0.98}^{+0.98}$   &   98.3, 100.0, 100.0    &   86.1, 87.2, 87.2  &   81.9, 97.0, 99.4   &    51.0, 52.8, 52.9  \\[3pt] 
 20.0  &  19.5 	& 19.9 & $_{-9.7}^{+15.8}$ 	& 15.0 & $_{-9.3}^{+9.3}$ &    1 & 25   &  1.01 & $_{-0.49}^{+0.80}$   & 	1.33 & $_{-0.82}^{+0.84}$   &   95.4, 100.0, 100.0    &   83.4, 87.3, 88.1  &   64.2, 97.3, 99.4   &    75.4, 78.0, 78.3  \\[3pt] 
 30.0  &  29.5 	& 30.0 & $_{-12.4}^{+16.6}$ 	& 22.7 & $_{-11.1}^{+11.1}$ &    0 & 24   &  1.00 & $_{-0.41}^{+0.55}$   & 	1.32 & $_{-0.65}^{+0.72}$   &   93.0, 99.8, 100.0    &   82.4, 91.8, 93.6  &   57.7, 94.1, 99.8   &    49.4, 68.6, 90.3  \\[3pt] 
 40.0  &  39.8 	& 41.1 & $_{-14.5}^{+16.8}$ 	& 31.1 & $_{-12.7}^{+12.7}$ &    3 & 22   &  0.97 & $_{-0.34}^{+0.40}$   & 	1.29 & $_{-0.53}^{+0.57}$   &   90.2, 100.0, 100.0    &   82.1, 93.7, 95.8  &   62.2, 92.1, 99.6   &    69.6, 82.3, 95.9  \\[3pt] 
 50.0  &  49.7 	& 50.8 & $_{-15.8}^{+16.2}$ 	& 38.4 & $_{-14.0}^{+14.0}$ &    2 & 23   &  0.98 & $_{-0.31}^{+0.31}$   & 	1.30 & $_{-0.47}^{+0.48}$   &   87.0, 99.8, 100.0    &   77.1, 94.5, 97.8  &   61.3, 91.8, 99.9   &    69.4, 90.7, 94.9  \\[3pt] 
 60.0  &  59.8 	& 60.6 & $_{-16.4}^{+15.2}$ 	& 45.8 & $_{-14.1}^{+14.1}$ &    1 & 24   &  0.99 & $_{-0.27}^{+0.25}$   & 	1.31 & $_{-0.40}^{+0.42}$   &   86.8, 99.7, 100.0    &   71.3, 92.9, 97.8  &   58.9, 92.2, 99.8   &    62.6, 83.5, 95.9  \\[3pt] 
 70.0  &  69.2 	& 69.1 & $_{-16.6}^{+13.7}$ 	& 52.7 & $_{-14.4}^{+14.4}$ &    1 & 25   &  1.01 & $_{-0.24}^{+0.20}$   & 	1.33 & $_{-0.36}^{+0.37}$   &   85.2, 99.6, 100.0    &   63.0, 92.8, 98.7  &   59.8, 95.4, 99.7   &    28.0, 72.3, 89.0  \\[3pt] 
 80.0  &  79.7 	& 79.8 & $_{-16.6}^{+10.5}$ 	& 61.2 & $_{-15.0}^{+15.0}$ &    0 & 24   &  1.00 & $_{-0.21}^{+0.13}$   & 	1.31 & $_{-0.32}^{+0.28}$   &   85.1, 99.1, 100.0    &   52.1, 86.4, 96.7  &   62.9, 97.2, 99.9   &    48.5, 85.0, 94.6  \\[3pt] 
 90.0  &  90.4 	& 89.3 & $_{-16.1}^{+6.7}$ 	& 68.4 & $_{-13.8}^{+13.8}$ &    1 & 24   &  1.01 & $_{-0.18}^{+0.08}$   & 	1.32 & $_{-0.26}^{+0.25}$   &   87.3, 98.7, 99.7    &   42.1, 78.1, 94.0  &   83.5, 96.9, 99.8   &    40.4, 79.7, 91.4  \\[3pt] 
 100.0  &  100.0 	& 96.7 & $_{-15.1}^{+2.5}$ 	& 76.5 & $_{-12.2}^{+12.2}$ &    3 & 23   &  1.03 & $_{-0.16}^{+0.03}$   & 	1.31 & $_{-0.21}^{+0.20}$   &   88.4, 98.0, 99.6    &   34.3, 68.7, 91.4  &   88.4, 98.0, 99.6   &    34.3, 68.7, 91.4  \\
\midrule 
\multicolumn{2}{l}{ N = 20}	& &  & &	&	&	&	&	& & & &	  &   &  \\ [-2pt] 
\midrule 
 10.0  &  9.9 	& 10.1 & $_{-5.1}^{+9.3}$ 	& 7.6 & $_{-4.7}^{+4.7}$ &    1 & 24   &  0.99 & $_{-0.51}^{+0.91}$   & 	1.31 & $_{-0.80}^{+0.82}$   &   97.2, 99.9, 100.0    &   84.8, 89.1, 89.7  &   60.7, 97.7, 99.6   &    74.8, 77.0, 77.2  \\[3pt] 
 20.0  &  19.8 	& 19.7 & $_{-8.1}^{+11.0}$ 	& 15.0 & $_{-7.3}^{+7.3}$ &    1 & 25   &  1.01 & $_{-0.42}^{+0.57}$   & 	1.33 & $_{-0.65}^{+0.64}$   &   95.9, 99.9, 100.0    &   81.0, 95.2, 96.9  &   67.3, 95.0, 100.0   &    60.2, 84.1, 95.8  \\[3pt] 
 30.0  &  29.8 	& 30.1 & $_{-10.1}^{+12.1}$ 	& 22.8 & $_{-9.6}^{+9.6}$ &    0 & 24   &  1.00 & $_{-0.34}^{+0.40}$   & 	1.32 & $_{-0.55}^{+0.49}$   &   92.2, 99.7, 100.0    &   76.1, 94.2, 98.1  &   66.2, 94.1, 99.8   &    43.7, 81.9, 90.3  \\[3pt] 
 40.0  &  40.5 	& 41.3 & $_{-11.6}^{+12.7}$ 	& 31.3 & $_{-11.0}^{+11.0}$ &    3 & 22   &  0.97 & $_{-0.27}^{+0.30}$   & 	1.28 & $_{-0.45}^{+0.41}$   &   92.3, 100.0, 100.0    &   69.5, 96.3, 99.4  &   65.6, 94.8, 99.2   &    55.2, 89.5, 96.8  \\[3pt] 
 50.0  &  50.1 	& 50.6 & $_{-12.3}^{+12.7}$ 	& 38.4 & $_{-11.7}^{+11.7}$ &    1 & 23   &  0.99 & $_{-0.24}^{+0.25}$   & 	1.30 & $_{-0.40}^{+0.36}$   &   90.1, 99.9, 100.0    &   59.8, 92.3, 98.4  &   65.3, 94.3, 99.5   &    53.4, 82.8, 95.9  \\[3pt] 
 60.0  &  60.1 	& 60.7 & $_{-12.6}^{+12.2}$ 	& 46.1 & $_{-11.9}^{+11.9}$ &    1 & 23   &  0.99 & $_{-0.20}^{+0.20}$   & 	1.30 & $_{-0.34}^{+0.33}$   &   88.0, 99.7, 100.0    &   51.3, 89.7, 98.5  &   65.8, 92.7, 99.8   &    53.5, 83.7, 93.2  \\[3pt] 
 70.0  &  69.7 	& 70.8 & $_{-12.5}^{+11.1}$ 	& 53.9 & $_{-11.6}^{+11.6}$ &    1 & 23   &  0.99 & $_{-0.17}^{+0.15}$   & 	1.30 & $_{-0.28}^{+0.29}$   &   83.5, 99.5, 100.0    &   45.0, 84.3, 96.3  &   64.4, 94.6, 100.0   &    28.2, 61.9, 86.1  \\[3pt] 
 80.0  &  80.0 	& 81.4 & $_{-11.9}^{+8.8}$ 	& 61.8 & $_{-10.8}^{+10.8}$ &    2 & 23   &  0.98 & $_{-0.14}^{+0.11}$   & 	1.30 & $_{-0.23}^{+0.24}$   &   82.4, 99.8, 100.0    &   32.9, 74.9, 91.9  &   59.2, 96.3, 99.8   &    35.3, 73.2, 90.3  \\[3pt] 
 90.0  &  90.0 	& 90.3 & $_{-11.1}^{+5.7}$ 	& 68.8 & $_{-9.8}^{+9.8}$ &    0 & 24   &  1.00 & $_{-0.12}^{+0.06}$   & 	1.31 & $_{-0.19}^{+0.21}$   &   86.9, 99.6, 100.0    &   18.4, 55.3, 82.7  &   74.5, 97.5, 99.8   &    23.4, 57.3, 84.0  \\[3pt] 
 100.0  &  100.0 	& 98.3 & $_{-9.7}^{+1.4}$ 	& 76.6 & $_{-8.9}^{+8.9}$ &    2 & 23   &  1.02 & $_{-0.10}^{+0.01}$   & 	1.30 & $_{-0.15}^{+0.17}$   &   90.9, 98.5, 100.0    &   3.9, 31.8, 61.8  &   90.9, 98.5, 100.0   &    3.9, 31.8, 61.8  \\
\midrule 
\multicolumn{2}{l}{ N = 50}	& &  & &	&	&	&	&	& & & &	  &   &  \\ [-2pt] 
\midrule 
 10.0  &  10.1 	& 10.1 & $_{-4.0}^{+5.5}$ 	& 7.6 & $_{-3.6}^{+3.6}$ &    1 & 24   &  0.99 & $_{-0.39}^{+0.54}$   & 	1.31 & $_{-0.62}^{+0.60}$   &   96.9, 100.0, 100.0    &   80.6, 97.0, 98.6  &   65.4, 96.5, 99.8   &    71.1, 89.3, 95.4  \\[3pt] 
 20.0  &  20.2 	& 20.5 & $_{-5.9}^{+7.0}$ 	& 15.5 & $_{-5.0}^{+5.0}$ &    2 & 22   &  0.98 & $_{-0.28}^{+0.33}$   & 	1.29 & $_{-0.41}^{+0.42}$   &   94.0, 100.0, 100.0    &   69.1, 95.5, 99.2  &   67.4, 95.7, 99.5   &    61.2, 88.0, 95.6  \\[3pt] 
 30.0  &  30.0 	& 30.4 & $_{-7.0}^{+7.8}$ 	& 23.0 & $_{-6.0}^{+6.0}$ &    1 & 23   &  0.99 & $_{-0.23}^{+0.25}$   & 	1.30 & $_{-0.34}^{+0.33}$   &   94.5, 100.0, 100.0    &   51.4, 92.2, 99.2  &   65.5, 95.1, 99.8   &    34.2, 72.0, 90.2  \\[3pt] 
 40.0  &  40.1 	& 40.4 & $_{-7.8}^{+8.3}$ 	& 30.8 & $_{-6.5}^{+6.5}$ &    1 & 23   &  0.99 & $_{-0.19}^{+0.20}$   & 	1.30 & $_{-0.28}^{+0.28}$   &   92.9, 99.9, 100.0    &   36.2, 85.8, 98.0  &   67.8, 94.7, 99.4   &    40.9, 76.6, 91.1  \\[3pt] 
 50.0  &  50.3 	& 51.0 & $_{-8.2}^{+8.4}$ 	& 38.7 & $_{-6.8}^{+6.8}$ &    2 & 23   &  0.98 & $_{-0.16}^{+0.16}$   & 	1.29 & $_{-0.23}^{+0.23}$   &   92.0, 99.9, 100.0    &   22.6, 75.7, 95.6  &   64.7, 94.4, 99.4   &    33.6, 68.3, 86.5  \\[3pt] 
 60.0  &  59.7 	& 60.1 & $_{-8.4}^{+8.3}$ 	& 45.6 & $_{-6.9}^{+6.9}$ &    0 & 24   &  1.00 & $_{-0.14}^{+0.14}$   & 	1.32 & $_{-0.20}^{+0.20}$   &   90.0, 99.7, 100.0    &   10.9, 57.4, 88.6  &   65.2, 94.7, 99.3   &    18.1, 52.5, 80.0  \\[3pt] 
 70.0  &  70.2 	& 70.7 & $_{-8.3}^{+7.8}$ 	& 53.7 & $_{-7.0}^{+7.0}$ &    1 & 23   &  0.99 & $_{-0.12}^{+0.11}$   & 	1.30 & $_{-0.17}^{+0.17}$   &   87.8, 99.6, 100.0    &   4.4, 38.3, 75.7  &   66.1, 95.6, 99.7   &    6.7, 30.2, 60.9  \\[3pt] 
 80.0  &  80.2 	& 80.9 & $_{-7.7}^{+6.7}$ 	& 61.5 & $_{-7.0}^{+7.0}$ &    1 & 23   &  0.99 & $_{-0.09}^{+0.08}$   & 	1.30 & $_{-0.15}^{+0.14}$   &   81.9, 99.1, 100.0    &   1.6, 20.1, 57.1  &   65.4, 93.5, 99.7   &    5.9, 27.5, 56.0  \\[3pt] 
 90.0  &  90.0 	& 90.6 & $_{-6.6}^{+4.6}$ 	& 69.1 & $_{-6.5}^{+6.5}$ &    1 & 23   &  0.99 & $_{-0.07}^{+0.05}$   & 	1.30 & $_{-0.12}^{+0.12}$   &   79.7, 99.6, 100.0    &   0.6, 7.1, 27.3  &   62.7, 96.0, 99.8   &    0.8, 8.2, 32.3  \\[3pt] 
 100.0  &  100.0 	& 99.1 & $_{-5.1}^{+0.7}$ 	& 76.6 & $_{-5.8}^{+5.8}$ &    1 & 23   &  1.01 & $_{-0.05}^{+0.01}$   & 	1.31 & $_{-0.10}^{+0.10}$   &   92.0, 99.3, 99.8    &   0.0, 0.1, 2.1  &   92.0, 99.3, 99.8   &    0.0, 0.1, 2.1  \\
\midrule 
\multicolumn{2}{l}{ N = 100}	& &  & &	&	&	&	&	& & & &	  &   &  \\ [-2pt] 
\midrule 
 10.0  &  9.9 	& 10.1 & $_{-3.0}^{+3.8}$ 	& 7.7 & $_{-2.6}^{+2.6}$ &    1 & 23   &  0.99 & $_{-0.29}^{+0.37}$   & 	1.30 & $_{-0.44}^{+0.44}$   &   96.0, 99.9, 100.0    &   70.9, 97.6, 99.5  &   65.5, 95.8, 99.7   &    57.8, 86.7, 95.7  \\[3pt] 
 20.0  &  20.0 	& 20.1 & $_{-4.3}^{+4.9}$ 	& 15.3 & $_{-3.5}^{+3.5}$ &    1 & 24   &  0.99 & $_{-0.21}^{+0.24}$   & 	1.31 & $_{-0.30}^{+0.31}$   &   94.6, 100.0, 100.0    &   40.2, 90.7, 99.0  &   66.6, 95.4, 100.0   &    44.2, 76.8, 91.7  \\[3pt] 
 30.0  &  30.1 	& 30.4 & $_{-5.1}^{+5.5}$ 	& 23.1 & $_{-4.2}^{+4.2}$ &    1 & 23   &  0.99 & $_{-0.17}^{+0.18}$   & 	1.30 & $_{-0.24}^{+0.24}$   &   94.6, 100.0, 100.0    &   19.1, 78.4, 97.2  &   68.7, 95.1, 99.5   &    22.3, 60.6, 83.2  \\[3pt] 
 40.0  &  39.9 	& 40.3 & $_{-5.6}^{+5.9}$ 	& 30.6 & $_{-4.5}^{+4.5}$ &    1 & 23   &  0.99 & $_{-0.14}^{+0.14}$   & 	1.31 & $_{-0.19}^{+0.20}$   &   94.2, 99.9, 100.0    &   7.2, 59.0, 91.9  &   69.4, 95.1, 99.7   &    18.1, 53.4, 80.3  \\[3pt] 
 50.0  &  50.0 	& 50.5 & $_{-5.9}^{+6.0}$ 	& 38.4 & $_{-4.8}^{+4.8}$ &    1 & 23   &  0.99 & $_{-0.12}^{+0.12}$   & 	1.30 & $_{-0.16}^{+0.16}$   &   92.3, 99.8, 100.0    &   1.8, 34.7, 77.8  &   66.9, 95.1, 99.5   &    9.9, 38.8, 68.0  \\[3pt] 
 60.0  &  60.1 	& 60.5 & $_{-6.0}^{+6.0}$ 	& 45.9 & $_{-5.0}^{+5.0}$ &    1 & 24   &  0.99 & $_{-0.10}^{+0.10}$   & 	1.31 & $_{-0.14}^{+0.14}$   &   88.9, 99.9, 100.0    &   0.5, 16.6, 53.2  &   68.9, 96.4, 99.7   &    3.5, 24.5, 51.2  \\[3pt] 
 70.0  &  70.1 	& 70.6 & $_{-5.9}^{+5.6}$ 	& 53.7 & $_{-4.9}^{+4.9}$ &    1 & 23   &  0.99 & $_{-0.08}^{+0.08}$   & 	1.30 & $_{-0.12}^{+0.12}$   &   87.9, 99.7, 100.0    &   0.0, 4.4, 28.3  &   67.2, 95.1, 99.6   &    0.6, 8.3, 28.0  \\[3pt] 
 80.0  &  80.2 	& 80.9 & $_{-5.5}^{+5.0}$ 	& 61.4 & $_{-4.8}^{+4.8}$ &    1 & 23   &  0.99 & $_{-0.07}^{+0.06}$   & 	1.30 & $_{-0.10}^{+0.10}$   &   82.0, 99.1, 100.0    &   0.0, 0.6, 9.5  &   67.3, 93.7, 99.4   &    0.5, 3.4, 13.0  \\[3pt] 
 90.0  &  90.0 	& 90.8 & $_{-4.6}^{+3.7}$ 	& 69.1 & $_{-4.6}^{+4.6}$ &    1 & 23   &  0.99 & $_{-0.05}^{+0.04}$   & 	1.30 & $_{-0.09}^{+0.08}$   &   76.2, 99.2, 100.0    &   0.0, 0.1, 0.7  &   66.1, 94.9, 99.9   &    0.0, 0.0, 1.6  \\[3pt] 
 100.0  &  100.0 	& 99.6 & $_{-3.1}^{+0.3}$ 	& 76.6 & $_{-4.2}^{+4.2}$ &    0 & 23   &  1.00 & $_{-0.03}^{+0.00}$   & 	1.31 & $_{-0.07}^{+0.07}$   &   95.4, 99.4, 100.0    &   0.0, 0.0, 0.0  &   95.4, 99.4, 100.0   &    0.0, 0.0, 0.0  \\
\bottomrule
\end{tabularx}
	\begin{tablenotes}[]\footnotesize
		\item[]\textit{Note} --- Uncertainties for $\bar{\theta}_{\mathrm{calc}}$ and $\bar{\theta}_{\mathrm{emit}}$ correspond to the mean values of the  68.3\% credible and confidence intervals, respectively.   These uncertainties are propagated for correction factors $\bar{f}_{\mathrm{calc}}$ and $\bar{f}_{\mathrm{calc}}$, which are calculated by dividing the true rate by $\bar{\theta}_{\mathrm{calc}}$ and $\bar{\theta}_{\mathrm{detect}}$, respectively. $\epsilon_{\text{calc}}$  and  $\epsilon_{\text{calc}}$  are the percent errors of the calculated mean occurrence and detection rates with respect to the true rate.
	\end{tablenotes}
\end{ThreePartTable}  
\end{table*}
 
\setlength{\tabcolsep}{0.04in}
\begin{table*}\centering 
		\begin{ThreePartTable}
			\caption{Summary of mean simulation results:  Low $\rms$, Uniform  \label{table:simulation_low-rms_uniform}} 
			\begin{tabularx}{\textwidth}{rrr@{\hspace{0.1pt}}lr@{\hspace{0.1pt}}l@{\hspace{0.16in}}llr@{\hspace{0.1pt}}lr@{\hspace{0.1pt}}l@{\hspace{0.16in}}ll@{\hspace{0.16in}}llc}
   		&
		&
   		&
 		&
 		&
 		&
   		&
 		&
   		&
 		&
   		&
 		&
\multicolumn{2}{l}{Recovery rate for $\theta_{\mathrm{emit}}$ }  &
\multicolumn{2}{l}{Recovery rate for $\theta_{\mathrm{true}}$ } \\ 
$\bar{\theta}_{\mathrm{true}}$  							&
$\bar{\theta}_{\mathrm{emit}}$		    					&
\multicolumn{2}{c}{$\bar{\theta}_{\mathrm{calc}}$}   		&
\multicolumn{2}{c}{$\bar{\theta}_{\mathrm{detect}}$} 		&
$\abs{\bar{\epsilon}_{\mathrm{calc}}   }$   				&
$\abs{\bar{\epsilon}_{\mathrm{detect}} }$ 					&
\multicolumn{2}{c}{$\bar{f}_{\mathrm{calc}}$}   			&
\multicolumn{2}{c}{$\bar{f}_{\mathrm{detect}}$} 			&
$\theta_{\mathrm{calc}}$ CI 								&
$\theta_{\mathrm{detect}}$ CI 								&
$\theta_{\mathrm{calc}}$ CI 								&
$\theta_{\mathrm{detect}}$ CI \\
(\%)  						& 
(\%)  						& 
\multicolumn{2}{c}{(\%)}  	& 
\multicolumn{2}{c}{(\%)}  	& 
(\%)	  					& 
(\%)  						& 
	  						& 
  							& 
	  						& 
	  						& 
(\%)  						& 
(\%)  						& 
(\%)  						& 
(\%)  						\\ 
   					&
					&
   					&
   					&
   					&
   					&
 					&
 					&
 					&
   					&
 					&
 					&
68.3, 95.5, 99.7	&
68.3, 95.5, 99.7	&
68.3, 95.5, 99.7	&
68.3, 95.5, 99.7  	\\  [-2pt] 
\midrule  
\multicolumn{2}{l}{ N = 10}	& &  & &	&	&	&	&	& & & &	  &   &  \\ [-2pt] 
\midrule 
 10.0  &  10.6 	& 10.7 & $_{-5.8}^{+13.1}$ 	& 9.3 & $_{-6.7}^{+6.7}$ &    7 & 7   &  0.93 & $_{-0.51}^{+1.14}$   & 	1.08 & $_{-0.78}^{+0.73}$   &   99.7, 100.0, 100.0    &   88.5, 89.4, 89.4  &   75.6, 96.0, 99.7   &    53.7, 57.7, 57.8  \\[3pt] 
 20.0  &  19.3 	& 19.5 & $_{-9.1}^{+14.2}$ 	& 16.9 & $_{-9.7}^{+9.7}$ &    3 & 15   &  1.03 & $_{-0.48}^{+0.75}$   & 	1.18 & $_{-0.68}^{+0.72}$   &   98.3, 100.0, 100.0    &   91.2, 93.5, 93.5  &   69.2, 97.1, 99.4   &    76.9, 80.4, 81.2  \\[3pt] 
 30.0  &  29.8 	& 30.1 & $_{-11.8}^{+15.1}$ 	& 26.1 & $_{-11.9}^{+11.9}$ &    0 & 13   &  1.00 & $_{-0.39}^{+0.50}$   & 	1.15 & $_{-0.52}^{+0.57}$   &   97.6, 99.9, 100.0    &   92.0, 96.2, 96.6  &   61.6, 91.8, 99.4   &    57.4, 74.7, 92.5  \\[3pt] 
 40.0  &  39.8 	& 40.1 & $_{-13.6}^{+15.4}$ 	& 34.9 & $_{-13.3}^{+13.3}$ &    0 & 13   &  1.00 & $_{-0.34}^{+0.38}$   & 	1.15 & $_{-0.44}^{+0.47}$   &   96.2, 99.9, 100.0    &   92.0, 97.6, 98.2  &   63.1, 94.4, 99.7   &    74.9, 88.5, 97.5  \\[3pt] 
 50.0  &  50.1 	& 50.5 & $_{-14.7}^{+15.0}$ 	& 44.1 & $_{-14.2}^{+14.2}$ &    1 & 12   &  0.99 & $_{-0.29}^{+0.29}$   & 	1.13 & $_{-0.37}^{+0.38}$   &   94.9, 100.0, 100.0    &   92.6, 99.0, 99.6  &   63.5, 92.1, 99.5   &    80.3, 94.4, 96.3  \\[3pt] 
 60.0  &  60.6 	& 61.0 & $_{-15.2}^{+14.0}$ 	& 53.2 & $_{-14.7}^{+14.7}$ &    2 & 11   &  0.98 & $_{-0.25}^{+0.23}$   & 	1.13 & $_{-0.31}^{+0.30}$   &   95.0, 100.0, 100.0    &   89.7, 98.4, 99.9  &   64.8, 93.2, 99.8   &    75.2, 90.0, 97.8  \\[3pt] 
 70.0  &  69.9 	& 70.8 & $_{-15.2}^{+12.3}$ 	& 61.7 & $_{-14.4}^{+14.4}$ &    1 & 12   &  0.99 & $_{-0.21}^{+0.17}$   & 	1.13 & $_{-0.26}^{+0.25}$   &   94.9, 99.8, 100.0    &   88.6, 98.3, 99.3  &   61.9, 94.5, 100.0   &    43.7, 83.8, 93.6  \\[3pt] 
 80.0  &  80.3 	& 80.6 & $_{-14.6}^{+9.8}$ 	& 70.7 & $_{-13.6}^{+13.6}$ &    1 & 12   &  0.99 & $_{-0.18}^{+0.12}$   & 	1.13 & $_{-0.22}^{+0.20}$   &   92.4, 99.8, 100.0    &   81.4, 96.5, 98.4  &   63.4, 97.3, 99.8   &    72.2, 91.6, 95.2  \\[3pt] 
 90.0  &  89.8 	& 89.8 & $_{-13.7}^{+6.1}$ 	& 79.3 & $_{-11.6}^{+11.6}$ &    0 & 12   &  1.00 & $_{-0.15}^{+0.07}$   & 	1.14 & $_{-0.17}^{+0.15}$   &   92.0, 99.5, 99.9    &   75.0, 95.8, 98.4  &   74.8, 96.9, 99.6   &    56.9, 82.0, 86.2  \\[3pt] 
 100.0 &  100.0 & 98.2 & $_{-12.3}^{+1.5}$ 	& 88.3 & $_{-7.3}^{+7.3}$ &    2 & 12   &  1.02 & $_{-0.13}^{+0.02}$   & 	1.13 & $_{-0.09}^{+0.10}$   &   90.9, 98.8, 99.8    &   73.1, 93.6, 99.2  &   90.9, 98.8, 99.8   &    73.1, 93.6, 99.2  \\ 
\midrule 
\multicolumn{2}{l}{ N = 20}	& &  & &	&	&	&	&	& & & &	  &   &  \\ [-2pt] 
\midrule 
 10.0  &  10.0 	& 10.1 & $_{-5.0}^{+8.4}$ 	& 8.8 & $_{-5.1}^{+5.1}$ &    1 & 13   &  0.99 & $_{-0.49}^{+0.83}$   & 	1.14 & $_{-0.67}^{+0.68}$   &   98.3, 100.0, 100.0    &   91.8, 94.8, 94.9  &   72.3, 97.4, 99.9   &    80.4, 82.9, 83.1  \\[3pt] 
 20.0  &  20.7 	& 21.0 & $_{-7.9}^{+10.2}$ 	& 18.3 & $_{-8.4}^{+8.4}$ &    5 & 9   &  0.95 & $_{-0.36}^{+0.46}$   & 	1.09 & $_{-0.50}^{+0.48}$   &   97.6, 100.0, 100.0    &   93.6, 98.6, 98.9  &   68.4, 95.1, 99.7   &    72.9, 90.0, 97.7  \\[3pt] 
 30.0  &  29.7 	& 30.2 & $_{-9.5}^{+11.0}$ 	& 26.3 & $_{-10.3}^{+10.3}$ &    1 & 12   &  0.99 & $_{-0.31}^{+0.36}$   & 	1.14 & $_{-0.44}^{+0.40}$   &   98.4, 100.0, 100.0    &   93.3, 99.6, 99.7  &   67.1, 94.7, 99.6   &    54.0, 88.6, 94.6  \\[3pt] 
 40.0  &  39.7 	& 40.2 & $_{-10.6}^{+11.5}$ 	& 35.0 & $_{-11.6}^{+11.6}$ &    1 & 13   &  0.99 & $_{-0.26}^{+0.28}$   & 	1.14 & $_{-0.38}^{+0.34}$   &   98.3, 100.0, 100.0    &   89.9, 100.0, 100.0  &   67.6, 93.9, 99.6   &    67.0, 93.8, 97.9  \\[3pt] 
 50.0  &  49.6 	& 50.2 & $_{-11.3}^{+11.5}$ 	& 43.8 & $_{-11.9}^{+11.9}$ &    0 & 12   &  1.00 & $_{-0.22}^{+0.23}$   & 	1.14 & $_{-0.31}^{+0.30}$   &   97.7, 100.0, 100.0    &   90.1, 99.3, 100.0  &   67.0, 93.4, 99.5   &    69.1, 90.3, 97.9  \\[3pt] 
 60.0  &  60.1 	& 60.4 & $_{-11.5}^{+11.0}$ 	& 52.6 & $_{-11.6}^{+11.6}$ &    1 & 12   &  0.99 & $_{-0.19}^{+0.18}$   & 	1.14 & $_{-0.25}^{+0.26}$   &   96.6, 99.9, 100.0    &   85.7, 99.1, 100.0  &   65.5, 94.5, 99.3   &    71.3, 94.1, 98.2  \\[3pt] 
 70.0  &  70.3 	& 71.3 & $_{-11.2}^{+9.8}$ 	& 62.3 & $_{-10.8}^{+10.8}$ &    2 & 11   &  0.98 & $_{-0.15}^{+0.14}$   & 	1.12 & $_{-0.19}^{+0.21}$   &   94.3, 100.0, 100.0    &   84.3, 98.6, 99.9  &   65.3, 93.8, 99.0   &    53.5, 85.9, 95.4  \\[3pt] 
 80.0  &  80.4 	& 81.0 & $_{-10.4}^{+8.1}$ 	& 71.0 & $_{-9.6}^{+9.6}$ &    1 & 11   &  0.99 & $_{-0.13}^{+0.10}$   & 	1.13 & $_{-0.15}^{+0.17}$   &   91.4, 99.9, 100.0    &   75.1, 95.5, 99.4  &   63.9, 94.7, 99.9   &    64.3, 91.1, 97.9  \\[3pt] 
 90.0  &  90.3 	& 90.7 & $_{-9.1}^{+5.2}$ 	& 79.5 & $_{-8.4}^{+8.4}$ &    1 & 12   &  0.99 & $_{-0.10}^{+0.06}$   & 	1.13 & $_{-0.12}^{+0.13}$   &   91.8, 99.9, 100.0    &   57.6, 90.5, 98.2  &   65.2, 98.0, 99.9   &    63.1, 89.5, 97.8  \\[3pt] 
 100.0 &  100.0 & 98.8 & $_{-7.5}^{+1.0}$ 	& 87.9 & $_{-6.5}^{+6.5}$ &    1 & 12   &  1.01 & $_{-0.08}^{+0.01}$   & 	1.14 & $_{-0.08}^{+0.08}$   &   91.3, 98.5, 99.8    &   31.6, 80.7, 94.1  &   91.3, 98.5, 99.8   &    31.6, 80.7, 94.1  \\ 
\midrule 
\multicolumn{2}{l}{ N = 50}	& &  & &	&	&	&	&	& & & &	  &   &  \\ [-2pt] 
\midrule 
 10.0  &  10.2 	& 10.4 & $_{-3.8}^{+5.1}$ 	& 9.1 & $_{-3.9}^{+3.9}$  &    4 & 9    &  0.96 & $_{-0.35}^{+0.47}$   & 	1.10 & $_{-0.47}^{+0.47}$   &   99.3, 100.0, 100.0    &   95.9, 99.3, 99.4  &   66.8, 96.4, 99.9   &    76.1, 93.8, 97.7  \\[3pt] 
 20.0  &  20.4 	& 20.6 & $_{-5.5}^{+6.4}$ 	& 18.0 & $_{-5.3}^{+5.3}$ &    3 & 10   &  0.97 & $_{-0.26}^{+0.30}$   & 	1.11 & $_{-0.33}^{+0.33}$   &   98.6, 100.0, 100.0    &   91.6, 99.9, 99.9  &   65.8, 94.8, 99.4   &    73.5, 95.2, 98.4  \\[3pt] 
 30.0  &  30.1 	& 30.4 & $_{-6.5}^{+7.2}$ 	& 26.5 & $_{-6.4}^{+6.4}$ &    1 & 12   &  0.99 & $_{-0.21}^{+0.23}$   & 	1.13 & $_{-0.27}^{+0.26}$   &   99.4, 100.0, 100.0    &   86.3, 99.9, 100.0  &   68.2, 95.2, 99.2   &    53.7, 87.1, 95.7  \\[3pt] 
 40.0  &  39.9 	& 40.4 & $_{-7.2}^{+7.5}$ 	& 35.1 & $_{-6.7}^{+6.7}$ &    1 & 12   &  0.99 & $_{-0.18}^{+0.18}$   & 	1.14 & $_{-0.22}^{+0.22}$   &   98.5, 100.0, 100.0    &   82.5, 99.5, 100.0  &   68.0, 94.6, 99.6   &    65.4, 91.4, 98.2  \\[3pt] 
 50.0  &  50.0 	& 50.6 & $_{-7.5}^{+7.6}$ 	& 44.1 & $_{-6.8}^{+6.8}$ &    1 & 12   &  0.99 & $_{-0.15}^{+0.15}$   & 	1.13 & $_{-0.18}^{+0.18}$   &   99.3, 100.0, 100.0    &   75.3, 99.2, 100.0  &   68.1, 94.7, 99.4   &    61.1, 89.6, 97.4  \\[3pt] 
 60.0  &  60.0 	& 60.5 & $_{-7.6}^{+7.3}$ 	& 52.9 & $_{-7.0}^{+7.0}$ &    1 & 12   &  0.99 & $_{-0.12}^{+0.12}$   & 	1.13 & $_{-0.15}^{+0.15}$   &   97.2, 100.0, 100.0    &   60.5, 97.1, 100.0  &   65.6, 96.1, 99.9   &    54.3, 86.5, 96.6  \\[3pt] 
 70.0  &  70.1 	& 70.7 & $_{-7.4}^{+6.8}$ 	& 61.8 & $_{-7.0}^{+7.0}$ &    1 & 12   &  0.99 & $_{-0.10}^{+0.10}$   & 	1.13 & $_{-0.13}^{+0.12}$   &   95.9, 100.0, 100.0    &   49.1, 92.0, 99.2  &   68.0, 93.9, 99.5   &    35.3, 74.9, 92.0  \\[3pt] 
 80.0  &  79.8 	& 80.5 & $_{-6.8}^{+5.9}$ 	& 70.3 & $_{-6.4}^{+6.4}$ &    1 & 12   &  0.99 & $_{-0.08}^{+0.07}$   & 	1.14 & $_{-0.10}^{+0.11}$   &   94.8, 100.0, 100.0    &   34.3, 84.0, 97.3  &   67.6, 95.2, 99.7   &    36.7, 75.2, 93.7  \\[3pt] 
 90.0  &  90.0 	& 90.6 & $_{-5.5}^{+4.1}$ 	& 79.4 & $_{-5.6}^{+5.6}$ &    1 & 12   &  0.99 & $_{-0.06}^{+0.05}$   & 	1.13 & $_{-0.08}^{+0.08}$   &   87.4, 99.8, 100.0    &   17.0, 63.4, 89.7  &   65.0, 95.7, 100.0   &    21.4, 59.2, 85.5  \\[3pt] 
 100.0 &  100.0 & 99.4 & $_{-3.7}^{+0.5}$ 	& 88.1 & $_{-4.4}^{+4.4}$ &    1 & 12   &  1.01 & $_{-0.04}^{+0.01}$   & 	1.14 & $_{-0.06}^{+0.06}$   &   91.2, 99.1, 100.0    &   1.7, 14.1, 46.4  &   91.2, 99.1, 100.0   &    1.7, 14.1, 46.4  \\ 
\midrule 
\multicolumn{2}{l}{ N = 100}	& &  & &	&	&	&	&	& & & &	  &   &  \\ [-2pt] 
\midrule 
 10.0  &  9.9 	& 10.1 & $_{-2.8}^{+3.5}$ 	& 8.8 & $_{-2.8}^{+2.8}$  &    1 & 12   &  0.99 & $_{-0.28}^{+0.34}$   & 	1.13 & $_{-0.36}^{+0.36}$   &   99.4, 100.0, 100.0   &   94.4, 99.7, 100.0  &   66.9, 95.8, 100.0   &    72.2, 93.3, 97.8  \\[3pt] 
 20.0  &  20.1 	& 20.5 & $_{-4.0}^{+4.5}$ 	& 17.8 & $_{-3.8}^{+3.8}$ &    2 & 11   &  0.98 & $_{-0.19}^{+0.21}$   & 	1.12 & $_{-0.24}^{+0.24}$   &   99.2, 99.9, 100.0    &   89.4, 99.8, 99.9  &   64.9, 94.6, 99.7   &    66.0, 90.7, 97.7  \\[3pt] 
 30.0  &  30.2 	& 30.6 & $_{-4.7}^{+5.1}$ 	& 26.6 & $_{-4.4}^{+4.4}$ &    2 & 11   &  0.98 & $_{-0.15}^{+0.16}$   & 	1.13 & $_{-0.18}^{+0.19}$   &   99.6, 100.0, 100.0   &   73.3, 99.7, 100.0  &   68.0, 94.8, 99.6   &    53.5, 85.3, 95.7  \\[3pt] 
 40.0  &  40.2 	& 40.7 & $_{-5.2}^{+5.3}$ 	& 35.4 & $_{-4.7}^{+4.7}$ &    2 & 11   &  0.98 & $_{-0.12}^{+0.13}$   & 	1.13 & $_{-0.15}^{+0.15}$   &   98.6, 100.0, 100.0   &   62.5, 98.7, 99.9  &   66.4, 94.9, 99.7   &    50.6, 86.5, 97.4  \\[3pt] 
 50.0  &  49.9 	& 50.4 & $_{-5.4}^{+5.4}$ 	& 43.9 & $_{-4.9}^{+4.9}$ &    1 & 12   &  0.99 & $_{-0.11}^{+0.11}$   & 	1.14 & $_{-0.13}^{+0.13}$   &   98.7, 100.0, 100.0   &   39.9, 97.3, 100.0  &   69.5, 94.8, 99.8   &    43.6, 80.3, 94.0  \\[3pt] 
 60.0  &  60.2 	& 60.9 & $_{-5.4}^{+5.3}$ 	& 53.1 & $_{-5.0}^{+5.0}$ &    1 & 12   &  0.99 & $_{-0.09}^{+0.09}$   & 	1.13 & $_{-0.11}^{+0.11}$   &   98.1, 100.0, 100.0   &   26.2, 89.2, 99.5  &   65.9, 94.9, 99.9   &    38.6, 75.6, 93.2  \\[3pt] 
 70.0  &  70.0 	& 70.7 & $_{-5.2}^{+4.9}$ 	& 61.8 & $_{-4.8}^{+4.8}$ &    1 & 12   &  0.99 & $_{-0.07}^{+0.07}$   & 	1.13 & $_{-0.09}^{+0.09}$   &   95.4, 100.0, 100.0   &   13.5, 74.3, 96.9  &   65.3, 94.9, 99.8   &    20.9, 57.4, 81.4  \\[3pt] 
 80.0  &  80.1 	& 80.6 & $_{-4.7}^{+4.3}$ 	& 70.5 & $_{-4.5}^{+4.5}$ &    1 & 12   &  0.99 & $_{-0.06}^{+0.05}$   & 	1.14 & $_{-0.07}^{+0.07}$   &   92.9, 100.0, 100.0   &   3.8, 49.5, 86.3  &   68.0, 94.5, 99.4   &    14.4, 50.0, 78.8  \\[3pt] 
 90.0  &  89.9 	& 90.4 & $_{-3.9}^{+3.2}$ 	& 79.0 & $_{-4.0}^{+4.0}$ &    0 & 12   &  1.00 & $_{-0.04}^{+0.04}$   & 	1.14 & $_{-0.06}^{+0.06}$   &   87.8, 99.8, 100.0    &   0.6, 18.0, 55.5  &   66.9, 94.2, 99.7   &    4.3, 22.6, 54.4  \\[3pt] 
 100.0 &  100.0 & 99.6 & $_{-2.2}^{+0.3}$ 	& 88.0 & $_{-3.2}^{+3.2}$ &    0 & 12   &  1.00 & $_{-0.02}^{+0.00}$   & 	1.14 & $_{-0.04}^{+0.04}$   &   92.9, 99.1, 100.0    &   0.0, 0.2, 2.6  &   92.9, 99.1, 100.0   &    0.0, 0.2, 2.6  \\ 
\bottomrule
\end{tabularx}
	\begin{tablenotes}[]\footnotesize
		\item[]\textit{Note} --- Uncertainties for $\bar{\theta}_{\mathrm{calc}}$ and $\bar{\theta}_{\mathrm{emit}}$ correspond to the mean values of the  68.3\% credible and confidence intervals, respectively.   These uncertainties are propagated for correction factors $\bar{f}_{\mathrm{calc}}$ and $\bar{f}_{\mathrm{calc}}$, which are calculated by dividing the true rate by $\bar{\theta}_{\mathrm{calc}}$ and $\bar{\theta}_{\mathrm{detect}}$, respectively. $\epsilon_{\text{calc}}$  and  $\epsilon_{\text{calc}}$  are the percent errors of the calculated mean occurrence and detection rates with respect to the true rate.
\end{tablenotes}
\end{ThreePartTable}  
\end{table*}
 
\setlength{\tabcolsep}{0.04in}
\begin{table*}\centering 
		\begin{ThreePartTable}
			\caption{Summary of mean simulation results: Low luminosity  \label{table:simulation_low-lum_uniform}} 
\begin{tabularx}{\textwidth}{rrr@{\hspace{0.1pt}}lr@{\hspace{0.1pt}}l@{\hspace{0.16in}}llr@{\hspace{0.1pt}}lr@{\hspace{0.1pt}}l@{\hspace{0.16in}}ll@{\hspace{0.16in}}llc}
   		&
		&
   		&
 		&
 		&
 		&
   		&
 		&
   		&
 		&
   		&
 		&
\multicolumn{2}{l}{Recovery rate for $\theta_{\mathrm{emit}}$ }  &
\multicolumn{2}{l}{Recovery rate for $\theta_{\mathrm{true}}$ } \\ 
$\bar{\theta}_{\mathrm{true}}$  							&
$\bar{\theta}_{\mathrm{emit}}$		    					&
\multicolumn{2}{c}{$\bar{\theta}_{\mathrm{calc}}$}   		&
\multicolumn{2}{c}{$\bar{\theta}_{\mathrm{detect}}$} 		&
$\abs{\bar{\epsilon}_{\mathrm{calc}}   }$   				&
$\abs{\bar{\epsilon}_{\mathrm{detect}} }$ 					&
\multicolumn{2}{c}{$\bar{f}_{\mathrm{calc}}$}   			&
\multicolumn{2}{c}{$\bar{f}_{\mathrm{detect}}$} 			&
$\theta_{\mathrm{calc}}$ CI 								&
$\theta_{\mathrm{detect}}$ CI 								&
$\theta_{\mathrm{calc}}$ CI 								&
$\theta_{\mathrm{detect}}$ CI \\
(\%)  						& 
(\%)  						& 
\multicolumn{2}{c}{(\%)}  	& 
\multicolumn{2}{c}{(\%)}  	& 
(\%)	  					& 
(\%)  						& 
	  						& 
  							& 
	  						& 
	  						& 
(\%)  						& 
(\%)  						& 
(\%)  						& 
(\%)  						\\ 
   		&
		&
   		&
   		&
   		&
   		&
 		&
 		&
 		&
   		&
 		&
 		&
68.3, 95.5, 99.7	&
68.3, 95.5, 99.7	&
68.3, 95.5, 99.7	&
68.3, 95.5, 99.7  	\\  [-2pt] 
\midrule  
\multicolumn{2}{l}{ N = 10}	& &  & &	&	&	&	&	& & & &	  &   &  \\ [-2pt] 
\midrule 
 10.0  &  10.1 	& 8.0 & $_{-3.6}^{+50.4}$ 	& 1.1 & $_{-1.1}^{+1.1}$ &    20 & 89   &  1.25 & $_{-0.57}^{+7.91}$   & 	9.26 & $_{-9.09}^{+7.89}$   &   90.4, 95.9, 100.0    &   43.1, 44.2, 44.6  &   89.6, 92.2, 99.7   &    9.8, 9.8, 9.8  \\[3pt] 
 20.0  &  20.0 	& 13.5 & $_{-6.0}^{+48.7}$ 	& 1.9 & $_{-1.8}^{+1.8}$ &    33 & 90   &  1.48 & $_{-0.66}^{+5.36}$   & 	10.31 & $_{-9.41}^{+8.13}$   &   85.9, 96.4, 100.0    &   20.1, 24.3, 26.7  &   84.8, 93.9, 99.9   &    16.0, 16.0, 16.0  \\[3pt] 
 30.0  &  30.2 	& 20.6 & $_{-9.0}^{+44.8}$ 	& 3.0 & $_{-2.6}^{+2.6}$ &    31 & 90   &  1.45 & $_{-0.63}^{+3.16}$   & 	10.07 & $_{-8.88}^{+7.30}$   &   76.3, 97.2, 100.0    &   9.4, 16.2, 22.6  &   78.5, 97.2, 100.0   &    1.5, 6.7, 23.2  \\[3pt] 
 40.0  &  39.9 	& 26.7 & $_{-11.4}^{+43.2}$ 	& 3.9 & $_{-3.4}^{+3.4}$ &    33 & 90   &  1.50 & $_{-0.64}^{+2.43}$   & 	10.15 & $_{-8.71}^{+7.94}$   &   73.5, 99.1, 100.0    &   4.3, 10.8, 20.7  &   68.6, 99.6, 100.0   &    3.8, 9.3, 30.4  \\[3pt] 
 50.0  &  49.2 	& 30.6 & $_{-13.0}^{+41.7}$ 	& 4.8 & $_{-3.9}^{+3.9}$ &    39 & 90   &  1.63 & $_{-0.69}^{+2.22}$   & 	10.48 & $_{-8.66}^{+7.56}$   &   67.1, 99.2, 100.0    &   1.7, 8.0, 19.3  &   61.6, 99.9, 100.0   &    0.7, 7.5, 20.8  \\[3pt] 
 60.0  &  59.9 	& 40.6 & $_{-17.0}^{+36.3}$ 	& 6.1 & $_{-4.9}^{+4.9}$ &    32 & 90   &  1.48 & $_{-0.62}^{+1.32}$   & 	9.90 & $_{-8.02}^{+7.24}$   &   61.3, 98.2, 99.7    &   1.3, 4.6, 13.7  &   66.3, 99.9, 100.0   &    0.3, 1.8, 9.4  \\[3pt] 
 70.0  &  70.0 	& 47.1 & $_{-19.5}^{+33.4}$ 	& 7.1 & $_{-5.5}^{+5.5}$ &    33 & 90   &  1.48 & $_{-0.61}^{+1.05}$   & 	9.86 & $_{-7.69}^{+7.33}$   &   59.8, 97.7, 99.3    &   0.4, 1.9, 7.1  &   55.5, 99.4, 100.0   &    0.0, 0.1, 1.6  \\[3pt] 
 80.0  &  79.7 	& 51.8 & $_{-21.3}^{+31.8}$ 	& 7.6 & $_{-6.0}^{+6.0}$ &    35 & 91   &  1.54 & $_{-0.63}^{+0.95}$   & 	10.54 & $_{-8.32}^{+7.47}$   &   61.8, 94.3, 96.8    &   0.0, 0.5, 3.0  &   59.5, 98.0, 100.0   &    0.0, 0.2, 1.2  \\[3pt] 
 90.0  &  89.5 	& 54.2 & $_{-21.5}^{+30.5}$ 	& 8.6 & $_{-6.3}^{+6.3}$ &    40 & 90   &  1.66 & $_{-0.66}^{+0.94}$   & 	10.45 & $_{-7.68}^{+7.31}$   &   60.0, 87.7, 90.2    &   0.0, 0.0, 1.1  &   59.4, 94.7, 99.9   &    0.0, 0.0, 0.4  \\[3pt] 
 100.0  &  100.0 	& 61.2 & $_{-24.1}^{+27.0}$ 	& 9.6 & $_{-7.0}^{+7.0}$ &    39 & 90   &  1.63 & $_{-0.64}^{+0.72}$   & 	10.44 & $_{-7.63}^{+6.96}$   &   64.9, 69.5, 70.2    &   0.0, 0.0, 0.0  &   64.9, 69.5, 70.2   &    0.0, 0.0, 0.0  \\
\midrule 
\multicolumn{2}{l}{ N = 20}	& &  & &	&	&	&	&	& & & &	  &   &  \\ [-2pt]
\midrule 
 10.0  &  10.4 	& 10.3 & $_{-4.9}^{+37.0}$ 	& 1.0 & $_{-0.9}^{+0.9}$ &    3 & 90   &  0.98 & $_{-0.47}^{+3.52}$   & 	9.66 & $_{-8.68}^{+8.31}$   &   86.4, 92.6, 99.7    &   19.6, 25.0, 28.1  &   83.5, 91.2, 99.2   &    17.5, 17.5, 17.5  \\[3pt] 
 20.0  &  19.8 	& 17.8 & $_{-8.4}^{+35.8}$ 	& 1.9 & $_{-1.7}^{+1.7}$ &    11 & 91   &  1.13 & $_{-0.53}^{+2.26}$   & 	10.55 & $_{-9.33}^{+8.88}$   &   79.5, 95.9, 100.0    &   5.2, 13.7, 23.1  &   79.6, 96.8, 99.9   &    1.4, 10.9, 31.9  \\[3pt] 
 30.0  &  30.6 	& 28.5 & $_{-12.7}^{+32.6}$ 	& 3.1 & $_{-2.5}^{+2.5}$ &    5 & 90   &  1.05 & $_{-0.47}^{+1.20}$   & 	9.62 & $_{-7.67}^{+7.35}$   &   70.2, 95.9, 100.0    &   1.5, 6.3, 16.9  &   66.9, 96.3, 100.0   &    0.0, 1.8, 7.5  \\[3pt] 
 40.0  &  40.4 	& 37.9 & $_{-16.0}^{+29.6}$ 	& 4.1 & $_{-3.0}^{+3.0}$ &    5 & 90   &  1.06 & $_{-0.45}^{+0.83}$   & 	9.78 & $_{-7.26}^{+6.99}$   &   61.1, 97.6, 100.0    &   0.5, 2.0, 6.5  &   55.4, 97.6, 100.0   &    0.0, 0.3, 4.2  \\[3pt] 
 50.0  &  50.1 	& 43.6 & $_{-17.9}^{+28.1}$ 	& 4.8 & $_{-3.4}^{+3.4}$ &    13 & 90   &  1.15 & $_{-0.47}^{+0.74}$   & 	10.53 & $_{-7.49}^{+7.24}$   &   57.4, 99.1, 100.0    &   0.0, 0.6, 2.6  &   51.5, 99.2, 100.0   &    0.0, 0.1, 0.8  \\[3pt] 
 60.0  &  60.8 	& 50.8 & $_{-19.8}^{+25.4}$ 	& 5.8 & $_{-3.8}^{+3.8}$ &    15 & 90   &  1.18 & $_{-0.46}^{+0.59}$   & 	10.43 & $_{-6.96}^{+6.85}$   &   59.7, 99.0, 100.0    &   0.0, 0.2, 1.0  &   51.6, 99.6, 100.0   &    0.0, 0.0, 0.0  \\[3pt] 
 70.0  &  69.5 	& 60.2 & $_{-22.9}^{+22.0}$ 	& 6.8 & $_{-4.3}^{+4.3}$ &    14 & 90   &  1.16 & $_{-0.44}^{+0.43}$   & 	10.32 & $_{-6.52}^{+6.81}$   &   63.6, 97.3, 100.0    &   0.0, 0.0, 0.1  &   61.6, 97.9, 100.0   &    0.0, 0.0, 0.0  \\[3pt] 
 80.0  &  80.2 	& 68.2 & $_{-25.0}^{+18.5}$ 	& 7.7 & $_{-4.8}^{+4.8}$ &    15 & 90   &  1.17 & $_{-0.43}^{+0.32}$   & 	10.33 & $_{-6.34}^{+6.47}$   &   68.9, 95.5, 99.7    &   0.0, 0.0, 0.0  &   69.9, 95.9, 100.0   &    0.0, 0.0, 0.0  \\[3pt] 
 90.0  &  89.7 	& 73.0 & $_{-25.8}^{+16.4}$ 	& 8.7 & $_{-5.2}^{+5.2}$ &    19 & 90   &  1.23 & $_{-0.43}^{+0.28}$   & 	10.31 & $_{-6.14}^{+6.24}$   &   71.3, 91.7, 98.1    &   0.0, 0.0, 0.0  &   70.8, 92.8, 99.7   &    0.0, 0.0, 0.0  \\[3pt] 
 100.0  &  100.0 	& 79.4 & $_{-27.2}^{+13.2}$ 	& 9.6 & $_{-5.5}^{+5.5}$ &    21 & 90   &  1.26 & $_{-0.43}^{+0.21}$   & 	10.40 & $_{-5.98}^{+6.06}$   &   72.6, 83.3, 86.1    &   0.0, 0.0, 0.0  &   72.6, 83.3, 86.1   &    0.0, 0.0, 0.0  \\
\midrule 
\multicolumn{2}{l}{ N = 50}	& &  & &	&	&	&	&	& & & &	  &   &  \\ [-2pt]
\midrule 
 10.0  &  10.2 	& 9.8 & $_{-5.6}^{+21.9}$ 	& 1.0 & $_{-0.8}^{+0.8}$ &    2 & 90   &  1.02 & $_{-0.58}^{+2.26}$   & 	10.35 & $_{-8.77}^{+8.32}$   &   83.5, 97.2, 99.8    &   2.8, 9.9, 21.2  &   86.2, 96.7, 99.4   &    1.3, 7.2, 29.0  \\[3pt] 
 20.0  &  20.1 	& 21.4 & $_{-10.5}^{+21.6}$ 	& 2.1 & $_{-1.5}^{+1.5}$ &    7 & 90   &  0.93 & $_{-0.46}^{+0.94}$   & 	9.65 & $_{-6.94}^{+6.67}$   &   70.1, 95.4, 99.7    &   0.1, 0.9, 3.6  &   66.7, 94.4, 99.6   &    0.0, 0.1, 1.7  \\[3pt] 
 30.0  &  29.9 	& 30.0 & $_{-13.4}^{+21.8}$ 	& 3.0 & $_{-1.9}^{+1.9}$ &    0 & 90   &  1.00 & $_{-0.45}^{+0.73}$   & 	10.16 & $_{-6.66}^{+6.42}$   &   63.2, 95.9, 100.0    &   0.0, 0.0, 0.5  &   61.3, 95.7, 99.8   &    0.0, 0.0, 0.0  \\[3pt] 
 40.0  &  40.2 	& 39.8 & $_{-15.9}^{+21.3}$ 	& 3.9 & $_{-2.4}^{+2.4}$ &    1 & 90   &  1.01 & $_{-0.40}^{+0.54}$   & 	10.15 & $_{-6.12}^{+5.91}$   &   62.1, 95.8, 99.8    &   0.0, 0.0, 0.0  &   60.2, 95.1, 99.7   &    0.0, 0.0, 0.0  \\[3pt] 
 50.0  &  49.7 	& 49.9 & $_{-18.1}^{+20.0}$ 	& 4.9 & $_{-2.8}^{+2.8}$ &    0 & 90   &  1.00 & $_{-0.36}^{+0.40}$   & 	10.20 & $_{-5.83}^{+5.61}$   &   60.3, 96.5, 99.9    &   0.0, 0.0, 0.0  &   58.6, 95.3, 100.0   &    0.0, 0.0, 0.0  \\[3pt] 
 60.0  &  60.0 	& 60.3 & $_{-20.1}^{+17.9}$ 	& 5.8 & $_{-3.1}^{+3.1}$ &    0 & 90   &  1.00 & $_{-0.33}^{+0.30}$   & 	10.31 & $_{-5.51}^{+5.29}$   &   58.7, 97.1, 100.0    &   0.0, 0.0, 0.0  &   57.8, 96.5, 100.0   &    0.0, 0.0, 0.0  \\[3pt] 
 70.0  &  70.0 	& 69.5 & $_{-20.8}^{+15.2}$ 	& 6.8 & $_{-3.4}^{+3.4}$ &    1 & 90   &  1.01 & $_{-0.30}^{+0.22}$   & 	10.25 & $_{-5.10}^{+4.98}$   &   64.3, 96.5, 99.9    &   0.0, 0.0, 0.0  &   59.3, 96.3, 100.0   &    0.0, 0.0, 0.0  \\[3pt] 
 80.0  &  80.1 	& 78.3 & $_{-21.6}^{+11.9}$ 	& 7.9 & $_{-3.7}^{+3.7}$ &    2 & 90   &  1.02 & $_{-0.28}^{+0.16}$   & 	10.08 & $_{-4.65}^{+4.55}$   &   73.4, 96.2, 99.7    &   0.0, 0.0, 0.0  &   71.4, 95.6, 99.6   &    0.0, 0.0, 0.0  \\[3pt] 
 90.0  &  90.0 	& 85.1 & $_{-21.8}^{+9.0}$ 	& 8.7 & $_{-3.8}^{+3.8}$ &    5 & 90   &  1.06 & $_{-0.27}^{+0.11}$   & 	10.32 & $_{-4.49}^{+4.46}$   &   79.0, 94.9, 99.6    &   0.0, 0.0, 0.0  &   79.5, 95.2, 99.4   &    0.0, 0.0, 0.0  \\[3pt] 
 100.0  &  100.0 	& 93.3 & $_{-21.8}^{+4.7}$ 	& 9.8 & $_{-4.0}^{+4.0}$ &    7 & 90   &  1.07 & $_{-0.25}^{+0.05}$   & 	10.19 & $_{-4.16}^{+4.20}$   &   85.0, 96.5, 98.8    &   0.0, 0.0, 0.0  &   85.0, 96.5, 98.8   &    0.0, 0.0, 0.0  \\
\midrule 
\multicolumn{2}{l}{ N = 100}	& &  & &	&	&	&	&	& & & &	  &   &  \\ [-2pt] 
\midrule 
 10.0  &  10.0 	& 9.5 & $_{-5.4}^{+13.2}$ 	& 0.9 & $_{-0.7}^{+0.7}$ &    5 & 91   &  1.05 & $_{-0.59}^{+1.46}$   & 	10.72 & $_{-7.80}^{+7.86}$   &   78.5, 98.6, 100.0    &   0.2, 0.9, 4.2  &   75.4, 97.6, 100.0   &    0.0, 0.0, 1.8  \\[3pt] 
 20.0  &  20.0 	& 20.1 & $_{-9.5}^{+14.9}$ 	& 2.0 & $_{-1.2}^{+1.2}$ &    0 & 90   &  1.00 & $_{-0.47}^{+0.74}$   & 	10.18 & $_{-6.24}^{+6.10}$   &   73.6, 98.2, 99.8    &   0.0, 0.0, 0.0  &   70.0, 97.1, 99.6   &    0.0, 0.0, 0.0  \\[3pt] 
 30.0  &  30.2 	& 30.3 & $_{-12.1}^{+16.1}$ 	& 3.0 & $_{-1.6}^{+1.6}$ &    1 & 90   &  0.99 & $_{-0.40}^{+0.53}$   & 	10.16 & $_{-5.44}^{+5.38}$   &   70.0, 96.0, 99.9    &   0.0, 0.0, 0.0  &   66.6, 95.5, 99.6   &    0.0, 0.0, 0.0  \\[3pt] 
 40.0  &  40.0 	& 40.2 & $_{-13.9}^{+16.4}$ 	& 3.8 & $_{-1.8}^{+1.8}$ &    1 & 90   &  0.99 & $_{-0.34}^{+0.41}$   & 	10.40 & $_{-4.89}^{+4.92}$   &   64.6, 95.1, 99.7    &   0.0, 0.0, 0.0  &   63.1, 93.1, 99.5   &    0.0, 0.0, 0.0  \\[3pt] 
 50.0  &  50.3 	& 49.9 & $_{-15.2}^{+16.3}$ 	& 4.8 & $_{-2.1}^{+2.1}$ &    0 & 90   &  1.00 & $_{-0.30}^{+0.33}$   & 	10.46 & $_{-4.49}^{+4.47}$   &   64.2, 94.0, 99.9    &   0.0, 0.0, 0.0  &   64.3, 92.6, 99.6   &    0.0, 0.0, 0.0  \\[3pt] 
 60.0  &  59.6 	& 59.4 & $_{-16.0}^{+15.6}$ 	& 5.6 & $_{-2.2}^{+2.2}$ &    1 & 91   &  1.01 & $_{-0.27}^{+0.26}$   & 	10.65 & $_{-4.20}^{+4.24}$   &   64.4, 95.3, 99.7    &   0.0, 0.0, 0.0  &   63.6, 94.6, 99.8   &    0.0, 0.0, 0.0  \\[3pt] 
 70.0  &  69.9 	& 69.4 & $_{-16.2}^{+13.7}$ 	& 6.7 & $_{-2.4}^{+2.4}$ &    1 & 90   &  1.01 & $_{-0.24}^{+0.20}$   & 	10.49 & $_{-3.82}^{+3.83}$   &   60.9, 96.4, 100.0    &   0.0, 0.0, 0.0  &   60.0, 95.4, 100.0   &    0.0, 0.0, 0.0  \\[3pt] 
 80.0  &  80.1 	& 79.5 & $_{-16.2}^{+10.7}$ 	& 7.7 & $_{-2.6}^{+2.6}$ &    1 & 90   &  1.01 & $_{-0.21}^{+0.14}$   & 	10.40 & $_{-3.55}^{+3.52}$   &   65.3, 97.3, 99.9    &   0.0, 0.0, 0.0  &   63.7, 96.9, 99.8   &    0.0, 0.0, 0.0  \\[3pt] 
 90.0  &  89.9 	& 88.0 & $_{-15.9}^{+7.3}$ 	& 8.7 & $_{-2.8}^{+2.8}$ &    2 & 90   &  1.02 & $_{-0.19}^{+0.08}$   & 	10.36 & $_{-3.31}^{+3.28}$   &   78.1, 97.3, 99.5    &   0.0, 0.0, 0.0  &   78.1, 97.1, 99.7   &    0.0, 0.0, 0.0  \\[3pt] 
 100.0  &  100.0 	& 96.0 & $_{-15.4}^{+3.1}$ 	& 9.7 & $_{-2.9}^{+2.9}$ &    4 & 90   &  1.04 & $_{-0.17}^{+0.03}$   & 	10.27 & $_{-3.08}^{+3.09}$   &   87.5, 98.1, 100.0    &   0.0, 0.0, 0.0  &   87.5, 98.1, 100.0   &    0.0, 0.0, 0.0  \\
\bottomrule
\end{tabularx}
	\begin{tablenotes}[]\footnotesize
		\item[]\textit{Note} --- Uncertainties for $\bar{\theta}_{\mathrm{calc}}$ and $\bar{\theta}_{\mathrm{emit}}$ correspond to the mean values of the  68.3\% credible and confidence intervals, respectively.   These uncertainties are propagated for correction factors $\bar{f}_{\mathrm{calc}}$ and $\bar{f}_{\mathrm{calc}}$, which are calculated by dividing the true rate by $\bar{\theta}_{\mathrm{calc}}$ and $\bar{\theta}_{\mathrm{detect}}$, respectively. $\epsilon_{\text{calc}}$  and  $\epsilon_{\text{calc}}$  are the percent errors of the calculated mean occurrence and detection rates with respect to the true rate.
\end{tablenotes}
\end{ThreePartTable}  
\end{table*}
 
\setlength{\tabcolsep}{0.04in}
\begin{table*}\centering 
		\begin{ThreePartTable}
			\caption{Summary of mean simulation results: Low luminosity, Low $\rms$  \label{table:simulation_low-rms_lowLum}} 
\begin{tabularx}{\textwidth}{rrr@{\hspace{0.1pt}}lr@{\hspace{0.1pt}}l@{\hspace{0.16in}}llr@{\hspace{0.1pt}}lr@{\hspace{0.1pt}}l@{\hspace{0.16in}}ll@{\hspace{0.16in}}llc}
   		&
		&
   		&
 		&
 		&
 		&
   		&
 		&
   		&
 		&
   		&
 		&
\multicolumn{2}{l}{Recovery rate for $\theta_{\mathrm{emit}}$ }  &
\multicolumn{2}{l}{Recovery rate for $\theta_{\mathrm{true}}$ } \\ 
$\bar{\theta}_{\mathrm{true}}$  							&
$\bar{\theta}_{\mathrm{emit}}$		    					&
\multicolumn{2}{c}{$\bar{\theta}_{\mathrm{calc}}$}   		&
\multicolumn{2}{c}{$\bar{\theta}_{\mathrm{detect}}$} 		&
$\abs{\bar{\epsilon}_{\mathrm{calc}}   }$   				&
$\abs{\bar{\epsilon}_{\mathrm{detect}} }$ 					&
\multicolumn{2}{c}{$\bar{f}_{\mathrm{calc}}$}   			&
\multicolumn{2}{c}{$\bar{f}_{\mathrm{detect}}$} 			&
$\theta_{\mathrm{calc}}$ CI 								&
$\theta_{\mathrm{detect}}$ CI 								&
$\theta_{\mathrm{calc}}$ CI 								&
$\theta_{\mathrm{detect}}$ CI \\
(\%)  						& 
(\%)  						& 
\multicolumn{2}{c}{(\%)}  	& 
\multicolumn{2}{c}{(\%)}  	& 
(\%)	  					& 
(\%)  						& 
	  						& 
  							& 
	  						& 
	  						& 
(\%)  						& 
(\%)  						& 
(\%)  						& 
(\%)  						\\ 
   		&
		&
   		&
   		&
   		&
   		&
 		&
 		&
 		&
   		&
 		&
 		&
68.3, 95.5, 99.7	&
68.3, 95.5, 99.7	&
68.3, 95.5, 99.7	&
68.3, 95.5, 99.7  	\\  [-2pt] 
\midrule  
\multicolumn{2}{l}{ N = 10}	& &  & &	&	&	&	&	& & & &	  &   &  \\ [-2pt] 
\midrule 
 10.0  &  10.6 	& 10.9 & $_{-6.1}^{+19.1}$ 	& 5.8 & $_{-4.7}^{+4.7}$ &    9 & 42   &  0.92 & $_{-0.51}^{+1.60}$   & 	1.73 & $_{-1.41}^{+1.21}$   &   92.4, 99.6, 100.0    &   70.1, 72.5, 72.5  &   86.7, 95.5, 99.4   &    39.3, 40.9, 40.9  \\[3pt] 
 20.0  &  19.3 	& 19.5 & $_{-9.8}^{+19.7}$ 	& 10.4 & $_{-7.5}^{+7.5}$ &    2 & 48   &  1.02 & $_{-0.51}^{+1.03}$   & 	1.92 & $_{-1.39}^{+1.23}$   &   85.5, 99.4, 100.0    &   60.6, 72.8, 74.2  &   62.0, 95.2, 99.4   &    61.5, 62.4, 62.4  \\[3pt] 
 30.0  &  29.8 	& 29.8 & $_{-13.1}^{+19.8}$ 	& 15.9 & $_{-9.4}^{+9.4}$ &    1 & 47   &  1.01 & $_{-0.44}^{+0.67}$   & 	1.89 & $_{-1.12}^{+1.12}$   &   80.3, 99.4, 100.0    &   54.3, 73.3, 80.1  &   62.2, 94.9, 99.5   &    31.2, 48.4, 77.5  \\[3pt] 
 40.0  &  39.8 	& 40.1 & $_{-15.8}^{+19.7}$ 	& 21.2 & $_{-11.1}^{+11.1}$ &    0 & 47   &  1.00 & $_{-0.39}^{+0.49}$   & 	1.89 & $_{-0.99}^{+1.03}$   &   77.2, 99.0, 100.0    &   49.2, 73.1, 85.5  &   60.9, 94.4, 100.0   &    47.3, 65.2, 89.5  \\[3pt] 
 50.0  &  50.1 	& 50.7 & $_{-17.5}^{+18.5}$ 	& 27.0 & $_{-12.3}^{+12.3}$ &    1 & 46   &  0.99 & $_{-0.34}^{+0.36}$   & 	1.85 & $_{-0.84}^{+0.89}$   &   73.3, 98.4, 100.0    &   41.4, 68.4, 86.0  &   57.2, 94.1, 100.0   &    45.3, 74.1, 83.5  \\[3pt] 
 60.0  &  60.6 	& 61.9 & $_{-18.9}^{+16.3}$ 	& 33.1 & $_{-12.9}^{+12.9}$ &    3 & 45   &  0.97 & $_{-0.30}^{+0.25}$   & 	1.81 & $_{-0.70}^{+0.78}$   &   72.5, 98.3, 99.9    &   34.4, 66.1, 84.0  &   57.1, 94.3, 100.0   &    34.0, 62.3, 85.9  \\[3pt] 
 70.0  &  69.9 	& 70.5 & $_{-19.5}^{+14.2}$ 	& 38.3 & $_{-13.9}^{+13.9}$ &    1 & 45   &  0.99 & $_{-0.27}^{+0.20}$   & 	1.83 & $_{-0.66}^{+0.67}$   &   76.4, 98.3, 100.0    &   21.4, 57.1, 80.5  &   58.1, 95.9, 99.7   &    8.1, 39.9, 62.8  \\[3pt] 
 80.0  &  80.3 	& 79.6 & $_{-20.0}^{+11.1}$ 	& 43.3 & $_{-13.9}^{+13.9}$ &    0 & 46   &  1.00 & $_{-0.25}^{+0.14}$   & 	1.85 & $_{-0.59}^{+0.64}$   &   79.8, 98.9, 99.9    &   13.4, 47.3, 72.6  &   70.1, 96.2, 99.8   &    14.3, 52.1, 76.5  \\[3pt] 
 90.0  &  89.8 	& 87.4 & $_{-20.1}^{+7.6}$ 	& 47.8 & $_{-14.3}^{+14.3}$ &    3 & 47   &  1.03 & $_{-0.24}^{+0.09}$   & 	1.88 & $_{-0.57}^{+0.58}$   &   84.4, 97.2, 99.9    &   6.4, 33.5, 62.0  &   81.5, 95.1, 99.2   &    6.5, 37.0, 65.8  \\[3pt] 
 100.0  &  100.0 	& 95.0 & $_{-19.5}^{+3.7}$ 	& 54.1 & $_{-14.8}^{+14.8}$ &    5 & 46   &  1.05 & $_{-0.22}^{+0.04}$   & 	1.85 & $_{-0.51}^{+0.48}$   &   87.2, 97.1, 99.8    &   4.0, 16.1, 47.4  &   87.2, 97.1, 99.8   &    4.0, 16.1, 47.4  \\
\midrule 
\multicolumn{2}{l}{ N = 20}	& &  & &	&	&	&	&	& & & &	  &   &  \\ [-2pt] 
\midrule 
 10.0  &  10.0 	& 10.1 & $_{-5.5}^{+11.9}$ 	& 5.4 & $_{-3.7}^{+3.7}$ &    1 & 46   &  0.99 & $_{-0.54}^{+1.17}$   & 	1.86 & $_{-1.29}^{+1.25}$   &   88.1, 100.0, 100.0    &   62.8, 74.4, 76.5  &   68.1, 97.2, 99.8   &    64.6, 64.9, 64.9  \\[3pt] 
 20.0  &  20.7 	& 21.0 & $_{-9.3}^{+13.6}$ 	& 11.3 & $_{-6.1}^{+6.1}$ &    5 & 44   &  0.95 & $_{-0.42}^{+0.62}$   & 	1.78 & $_{-0.96}^{+0.98}$   &   82.4, 99.6, 100.0    &   49.4, 77.6, 88.3  &   67.4, 95.7, 99.1   &    43.8, 69.1, 89.7  \\[3pt] 
 30.0  &  29.7 	& 30.8 & $_{-11.6}^{+14.7}$ 	& 16.4 & $_{-7.9}^{+7.9}$ &    3 & 45   &  0.97 & $_{-0.37}^{+0.46}$   & 	1.83 & $_{-0.88}^{+0.85}$   &   82.8, 99.3, 100.0    &   38.4, 73.9, 89.6  &   65.7, 94.1, 99.5   &    18.8, 59.4, 75.9  \\[3pt] 
 40.0  &  39.7 	& 40.5 & $_{-13.2}^{+15.1}$ 	& 21.5 & $_{-9.3}^{+9.3}$ &    1 & 46   &  0.99 & $_{-0.32}^{+0.37}$   & 	1.86 & $_{-0.80}^{+0.73}$   &   79.4, 98.5, 100.0    &   22.6, 63.9, 84.8  &   64.1, 93.0, 99.5   &    20.3, 61.3, 80.3  \\[3pt] 
 50.0  &  49.6 	& 50.6 & $_{-14.3}^{+15.1}$ 	& 27.1 & $_{-10.4}^{+10.4}$ &    1 & 46   &  0.99 & $_{-0.28}^{+0.29}$   & 	1.85 & $_{-0.71}^{+0.64}$   &   82.3, 99.5, 100.0    &   12.7, 53.1, 79.8  &   65.0, 93.6, 99.7   &    18.2, 47.3, 75.7  \\[3pt] 
 60.0  &  60.1 	& 61.6 & $_{-14.8}^{+14.2}$ 	& 32.7 & $_{-11.2}^{+11.2}$ &    3 & 45   &  0.97 & $_{-0.24}^{+0.22}$   & 	1.83 & $_{-0.63}^{+0.57}$   &   75.7, 98.7, 100.0    &   8.9, 38.6, 66.4  &   64.1, 93.4, 99.8   &    12.9, 42.2, 66.7  \\[3pt] 
 70.0  &  70.3 	& 71.8 & $_{-14.9}^{+12.2}$ 	& 38.3 & $_{-11.6}^{+11.6}$ &    3 & 45   &  0.98 & $_{-0.20}^{+0.17}$   & 	1.83 & $_{-0.55}^{+0.52}$   &   76.3, 98.9, 99.9    &   3.9, 24.3, 53.5  &   60.0, 94.9, 99.8   &    2.7, 16.1, 38.1  \\[3pt] 
 80.0  &  80.4 	& 81.3 & $_{-14.6}^{+9.7}$ 	& 43.6 & $_{-11.9}^{+11.9}$ &    2 & 45   &  0.98 & $_{-0.18}^{+0.12}$   & 	1.83 & $_{-0.50}^{+0.48}$   &   77.8, 99.8, 100.0    &   1.7, 12.6, 36.3  &   62.6, 97.8, 99.9   &    2.8, 15.2, 37.8  \\[3pt] 
 90.0  &  90.3 	& 89.8 & $_{-13.9}^{+6.3}$ 	& 49.0 & $_{-12.0}^{+12.0}$ &    0 & 46   &  1.00 & $_{-0.16}^{+0.07}$   & 	1.84 & $_{-0.45}^{+0.43}$   &   83.7, 98.9, 100.0    &   0.3, 4.6, 17.4  &   77.9, 98.1, 99.8   &    0.6, 5.7, 20.8  \\[3pt] 
 100.0  &  100.0 	& 97.3 & $_{-13.0}^{+2.2}$ 	& 54.1 & $_{-11.7}^{+11.7}$ &    3 & 46   &  1.03 & $_{-0.14}^{+0.02}$   & 	1.85 & $_{-0.40}^{+0.40}$   &   89.5, 98.7, 99.8    &   0.0, 0.5, 5.1  &   89.5, 98.7, 99.8   &    0.0, 0.5, 5.1  \\
\midrule 
\multicolumn{2}{l}{ N = 50}	& &  & &	&	&	&	&	& & & &	  &   &  \\ [-2pt] 
\midrule 
 10.0  &  10.2 	& 10.6 & $_{-4.7}^{+6.9}$ 	& 5.7 & $_{-3.1}^{+3.1}$ &    6 & 43   &  0.94 & $_{-0.41}^{+0.61}$   & 	1.76 & $_{-0.95}^{+0.91}$   &   85.2, 99.8, 100.0    &   46.5, 80.0, 91.4  &   65.3, 96.7, 99.9   &    53.6, 78.6, 88.9  \\[3pt] 
 20.0  &  20.4 	& 20.5 & $_{-6.8}^{+8.4}$ 	& 11.0 & $_{-4.2}^{+4.2}$ &    2 & 45   &  0.98 & $_{-0.32}^{+0.40}$   & 	1.82 & $_{-0.70}^{+0.70}$   &   81.7, 99.7, 99.9    &   16.9, 61.0, 83.2  &   66.8, 94.6, 99.5   &    26.8, 59.8, 79.2  \\[3pt] 
 30.0  &  30.1 	& 30.7 & $_{-8.3}^{+9.5}$ 	& 16.5 & $_{-5.1}^{+5.1}$ &    2 & 45   &  0.98 & $_{-0.26}^{+0.30}$   & 	1.82 & $_{-0.57}^{+0.57}$   &   83.3, 99.2, 100.0    &   6.8, 38.0, 70.2  &   67.5, 94.4, 99.5   &    4.6, 28.8, 55.8  \\[3pt] 
 40.0  &  39.9 	& 40.7 & $_{-9.3}^{+10.0}$ 	& 21.6 & $_{-5.9}^{+5.9}$ &    2 & 46   &  0.98 & $_{-0.22}^{+0.24}$   & 	1.85 & $_{-0.51}^{+0.48}$   &   79.5, 99.0, 100.0    &   2.1, 18.1, 50.8  &   66.6, 94.8, 99.6   &    4.0, 22.1, 47.3  \\[3pt] 
 50.0  &  50.0 	& 51.1 & $_{-9.8}^{+10.1}$ 	& 27.5 & $_{-6.4}^{+6.4}$ &    2 & 45   &  0.98 & $_{-0.19}^{+0.19}$   & 	1.82 & $_{-0.42}^{+0.41}$   &   79.0, 99.0, 100.0    &   0.6, 8.2, 31.8  &   65.0, 93.1, 99.7   &    1.5, 10.9, 33.1  \\[3pt] 
 60.0  &  60.0 	& 60.9 & $_{-10.1}^{+9.9}$ 	& 32.8 & $_{-6.6}^{+6.6}$ &    2 & 45   &  0.99 & $_{-0.16}^{+0.16}$   & 	1.83 & $_{-0.37}^{+0.38}$   &   80.6, 98.6, 100.0    &   0.1, 2.5, 14.2  &   66.1, 94.8, 99.5   &    0.5, 4.4, 16.1  \\[3pt] 
 70.0  &  70.1 	& 71.0 & $_{-10.0}^{+9.3}$ 	& 37.8 & $_{-6.8}^{+6.8}$ &    1 & 46   &  0.99 & $_{-0.14}^{+0.13}$   & 	1.85 & $_{-0.33}^{+0.34}$   &   78.4, 97.7, 100.0    &   0.0, 0.6, 4.1  &   65.2, 94.2, 99.4   &    0.0, 0.0, 2.6  \\[3pt] 
 80.0  &  79.8 	& 81.0 & $_{-9.4}^{+7.8}$ 	& 43.2 & $_{-6.9}^{+6.9}$ &    1 & 46   &  0.99 & $_{-0.12}^{+0.10}$   & 	1.85 & $_{-0.29}^{+0.30}$   &   75.7, 99.2, 100.0    &   0.0, 0.0, 0.6  &   65.1, 95.1, 99.8   &    0.0, 0.1, 0.2  \\[3pt] 
 90.0  &  90.0 	& 91.1 & $_{-8.5}^{+5.0}$ 	& 48.7 & $_{-7.0}^{+7.0}$ &    1 & 46   &  0.99 & $_{-0.09}^{+0.05}$   & 	1.85 & $_{-0.27}^{+0.27}$   &   79.2, 99.2, 100.0    &   0.0, 0.0, 0.3  &   66.9, 98.4, 99.9   &    0.0, 0.0, 0.0  \\[3pt] 
 100.0  &  100.0 	& 99.1 & $_{-7.0}^{+0.8}$ 	& 54.2 & $_{-7.0}^{+7.0}$ &    1 & 46   &  1.01 & $_{-0.07}^{+0.01}$   & 	1.85 & $_{-0.24}^{+0.24}$   &   95.3, 99.5, 100.0    &   0.0, 0.0, 0.0  &   95.3, 99.5, 100.0   &    0.0, 0.0, 0.0  \\
\midrule 
\multicolumn{2}{l}{ N = 100}	& &  & &	&	&	&	&	& & & &	  &   &  \\ [-2pt] 
\midrule 
 10.0  &  9.9 	& 10.3 & $_{-3.6}^{+4.6}$ 	& 5.5 & $_{-2.2}^{+2.2}$ &    3 & 45   &  0.97 & $_{-0.33}^{+0.43}$   & 	1.82 & $_{-0.72}^{+0.74}$   &   85.9, 99.7, 100.0    &   21.0, 68.0, 89.9  &   67.6, 95.8, 99.7   &    28.6, 62.7, 81.0  \\[3pt] 
 20.0  &  20.1 	& 20.8 & $_{-5.1}^{+6.0}$ 	& 11.1 & $_{-3.1}^{+3.1}$ &    4 & 44   &  0.96 & $_{-0.24}^{+0.27}$   & 	1.80 & $_{-0.50}^{+0.49}$   &   83.9, 99.6, 100.0    &   2.6, 27.3, 60.5  &   66.6, 94.8, 99.3   &    7.4, 31.5, 58.6  \\[3pt] 
 30.0  &  30.2 	& 31.1 & $_{-6.2}^{+6.8}$ 	& 16.6 & $_{-3.7}^{+3.7}$ &    4 & 45   &  0.97 & $_{-0.19}^{+0.21}$   & 	1.81 & $_{-0.41}^{+0.40}$   &   83.0, 99.5, 100.0    &   0.1, 6.0, 28.1  &   69.1, 94.7, 99.8   &    0.4, 7.0, 26.0  \\[3pt] 
 40.0  &  40.2 	& 41.0 & $_{-6.8}^{+7.2}$ 	& 21.9 & $_{-4.1}^{+4.1}$ &    2 & 45   &  0.98 & $_{-0.16}^{+0.17}$   & 	1.83 & $_{-0.35}^{+0.34}$   &   83.7, 99.8, 100.0    &   0.0, 0.6, 7.1  &   66.7, 95.2, 99.9   &    0.1, 1.8, 13.0  \\[3pt] 
 50.0  &  49.9 	& 50.6 & $_{-7.2}^{+7.3}$ 	& 27.1 & $_{-4.3}^{+4.3}$ &    1 & 46   &  0.99 & $_{-0.14}^{+0.14}$   & 	1.85 & $_{-0.30}^{+0.30}$   &   80.1, 99.3, 100.0    &   0.0, 0.0, 1.1  &   67.7, 95.1, 99.5   &    0.0, 0.3, 2.2  \\[3pt] 
 60.0  &  60.2 	& 61.1 & $_{-7.4}^{+7.3}$ 	& 32.4 & $_{-4.6}^{+4.6}$ &    2 & 46   &  0.98 & $_{-0.12}^{+0.12}$   & 	1.85 & $_{-0.26}^{+0.27}$   &   79.6, 99.1, 100.0    &   0.0, 0.0, 0.1  &   65.4, 94.6, 99.6   &    0.0, 0.0, 0.5  \\[3pt] 
 70.0  &  70.0 	& 71.1 & $_{-7.3}^{+6.9}$ 	& 38.1 & $_{-4.8}^{+4.8}$ &    2 & 46   &  0.98 & $_{-0.10}^{+0.10}$   & 	1.84 & $_{-0.23}^{+0.23}$   &   79.5, 98.6, 99.9    &   0.0, 0.0, 0.0  &   65.1, 94.1, 99.8   &    0.0, 0.0, 0.0  \\[3pt] 
 80.0  &  80.1 	& 81.1 & $_{-6.7}^{+6.0}$ 	& 43.5 & $_{-4.9}^{+4.9}$ &    1 & 46   &  0.99 & $_{-0.08}^{+0.07}$   & 	1.84 & $_{-0.21}^{+0.20}$   &   74.7, 97.7, 100.0    &   0.0, 0.0, 0.0  &   64.1, 93.4, 99.3   &    0.0, 0.0, 0.0  \\[3pt] 
 90.0  &  89.9 	& 91.0 & $_{-5.8}^{+4.2}$ 	& 48.6 & $_{-5.0}^{+5.0}$ &    1 & 46   &  0.99 & $_{-0.06}^{+0.05}$   & 	1.85 & $_{-0.19}^{+0.19}$   &   71.4, 98.8, 100.0    &   0.0, 0.0, 0.0  &   62.1, 96.4, 99.9   &    0.0, 0.0, 0.0  \\[3pt] 
 100.0  &  100.0 	& 99.6 & $_{-4.3}^{+0.3}$ 	& 54.2 & $_{-4.9}^{+4.9}$ &    0 & 46   &  1.00 & $_{-0.04}^{+0.00}$   & 	1.85 & $_{-0.17}^{+0.17}$   &   97.2, 99.4, 100.0    &   0.0, 0.0, 0.0  &   97.2, 99.4, 100.0   &    0.0, 0.0, 0.0  \\
\bottomrule
\end{tabularx}
	\begin{tablenotes}[]\footnotesize
		\item[]\textit{Note} ---Uncertainties for $\bar{\theta}_{\mathrm{calc}}$ and $\bar{\theta}_{\mathrm{emit}}$ correspond to the mean values of the  68.3\% credible and confidence intervals, respectively.   These uncertainties are propagated for correction factors $\bar{f}_{\mathrm{calc}}$ and $\bar{f}_{\mathrm{calc}}$, which are calculated by dividing the true rate by $\bar{\theta}_{\mathrm{calc}}$ and $\bar{\theta}_{\mathrm{detect}}$, respectively. $\epsilon_{\text{calc}}$  and  $\epsilon_{\text{calc}}$  are the percent errors of the calculated mean occurrence and detection rates with respect to the true rate.
	\end{tablenotes}
\end{ThreePartTable}  
\end{table*}




\bsp	
\label{lastpage}
\end{document}